\documentclass[12pt]{article}

\usepackage{algorithm}
\usepackage{algorithmic}
\usepackage[algo2e]{algorithm2e}
\RestyleAlgo{ruled}
\usepackage{amsthm,amsmath,amssymb}
\usepackage{natbib}
\usepackage{multirow}
\usepackage[pdftex]{graphicx}
\usepackage{subfigure}
\usepackage{makecell}
\usepackage{booktabs}
\usepackage{array}
\usepackage{fullpage}
\usepackage{url}
\usepackage{bm}
\usepackage{smile}
\usepackage{mathtools}
\usepackage{wrapfig}
\usepackage{lipsum}
\usepackage{mathrsfs}
\usepackage{dsfont}
\usepackage{titling}
\usepackage{epstopdf}
\usepackage{bbm}
\usepackage{relsize}
\usepackage{smile}
\usepackage{natbib}
\usepackage{multirow}
\usepackage{color}
\usepackage{mathtools}
\usepackage{caption}

\usepackage{natbib}
\usepackage{multirow}

\usepackage[usenames,dvipsnames,svgnames,table]{xcolor}
\usepackage[colorlinks=true,
            linkcolor=blue,
            urlcolor=blue,
            citecolor=blue]{hyperref}

\providecommand{\norm}[1]{\|#1\|}

\newcommand{\sign}[1]{\text{sign}(#1)}

\usepackage{comment}
\usepackage{tikz}
\tikzset{
	declare function={
		normcdf(\x,\m,\s)=1/(1 + exp(-0.07056*((\x-\m)/\s)^3 - 1.5976*(\x-\m)/\s));
		zeroone(\x)= (\x<=-2) * (\x*\x + 6*\x + 8)   +
		and(\x>-2, \x<=1) * (2 - \x - \x*\x)     +
		and(\x>1,  \x<=2) * (6 - 8*\x + 2*\x*\x) +
		(\x>2) * (-10 + 6*\x - \x*\x);
	}
}

\usepackage{mathtools}

\usepackage{enumitem}

	\makeatletter
	\newcommand{\mypm}{\mathbin{\mathpalette\@mypm\relax}}
	\newcommand{\@mypm}[2]{\ooalign{%
			\raisebox{.1\height}{$#1+$}\cr
			\smash{\raisebox{-.6\height}{$#1-$}}\cr}}
	\makeatother

\def\beq{\begin{equation}\begin{aligned}[b]}
\def\eeq{\end{aligned}\end{equation}}

\def\bq{\begin{equation*}\begin{aligned}[b]}
\def\eq{\end{aligned}\end{equation*}}

\def\wh{\widehat}
\def\wt{\widetilde}

\def\note#1{({\color{red}#1})}

\def\betadiff{\wh{\bbeta} - \bbeta^*}

\def\vdiffone{\wh{\bv}^{(1)} - \bv^*}

\def\boxit#1{\vbox{\hrule\hbox{\vrule\kern6pt\vbox{\kern6pt#1\kern6pt}\kern6pt\vrule}\hrule}}

\usepackage{geometry}
\usepackage{setspace}
\begin{document}

\title{Test of Significance for High-dimensional Thresholds with Application to Individualized Minimal Clinically Important Difference}


\author{
Huijie Feng\thanks{Department of Statistics and Data Science, Cornell University, Ithaca, NY 14850, USA; e-mail: \texttt{hf279@cornell.edu}.}~~~~~
Jingyi Duan\thanks{Department of Statistics and Data Science, Cornell University, Ithaca, NY 14850, USA; e-mail: \texttt{jd2222@cornell.edu}. The first two authors contribute equally to the paper.}~~~~~
Yang Ning\thanks{Department of Statistics and Data Science, Cornell University, Ithaca, NY 14850, USA; e-mail: \texttt{yn265@cornell.edu}.}~~~~~
Jiwei Zhao\thanks{Department of Biostatistics and Medical Informatics, University of Wisconsin-Madison, Madison, WI 53726, USA; e-mail: \texttt{jiwei.zhao@wisc.edu}.}
}

\maketitle
\thispagestyle{empty}

\newpage
\setcounter{page}{1} 
\thispagestyle{empty}

\begin{center}
{\Large Test of Significance for High-dimensional Thresholds with Application to Individualized Minimal Clinically Important Difference}
\end{center}

\begin{abstract}
This work is motivated by learning the individualized minimal clinically important difference, a vital concept to assess clinical importance in various biomedical studies. We formulate the scientific question into a high-dimensional statistical problem where the parameter of interest lies in an individualized linear threshold. The goal is to develop a hypothesis testing procedure for the significance of a single element in this parameter as well as of a linear combination of this parameter.
The difficulty dues to the high-dimensional nuisance in developing such a testing procedure, and also stems from the fact that this high-dimensional threshold model is nonregular and the limiting distribution of the corresponding estimator is nonstandard.
To deal with these challenges, we construct a test statistic via a new bias-corrected smoothed decorrelated score approach, and establish its asymptotic distributions under both null and local alternative hypotheses.
We propose a double-smoothing approach to select the optimal bandwidth in our test statistic and provide theoretical guarantees for the selected bandwidth.
We conduct simulation studies to demonstrate how our proposed procedure can be applied in empirical studies. We apply the proposed method to a clinical trial where the scientific goal is to assess the clinical importance of a surgery procedure.
\end{abstract}

{\bf Key Words:}
bandwidth selection;
high-dimensional statistical inference;
kernel method;
nonstandard asymptotics.

\newpage

\setcounter{equation}{0}

\section{Introduction}

\subsection{Motivation: Individualized Minimal Clinically Important Difference (iMCID) under High-dimensionality}

{
In clinical studies, instead of statistical significance, the effect of a treatment or intervention is widely assessed through clinical significance.
}
By leveraging patient-reported outcomes (PRO) that are directly collected from the patients without a third party's interpretation, the aim of assessing clinical significance is to provide clinicians and policy makers the clinical effectiveness of the treatment or intervention.
For example, in our motivating study, the ChAMP randomized controlled trial \citep{bisson2015design}, the interest is to identify \emph{the smallest WOMAC pain score change} such that the corresponding improvement and beyond can be claimed as clinically significant.
In \cite{jaeschke1989measurement}, this concept was firstly and formally introduced
as the minimal clinically important difference (MCID), ``the smallest difference in score in the domain of interest which patients perceive as beneficial and which would mandate a change in the patient's management''.

There are roughly three approaches to determine the magnitude of MCID \citep{lassere2001foundations, erdogan2016minimal, angst2017minimal, jayadevappa2017minimal}: distribution-based, opinion-based, and anchor-based. Although adopted in various studies \citep{wyrwich1999linking, wyrwich1999further, samsa1999determining, norman2003interpretation, bellamy2001towards}, the first two approaches are usually criticized \citep{mcglothlin2014minimal} due to, for example,  ``distribution-based methods are not derived from individual patients'' and ``expert opinion may not be a valid and reliable way to determine what is important to patients''.
The third approach, anchor-based, conceptually determines the MCID by incorporating both certainty of effective treatment encoded as a continuous variable and the patient's satisfaction collected from the anchor question.
It is clinically evident \citep{wells2001minimal} that the magnitude of MCID would depend on various factors such as the demographic variables and the patients' baseline status.
For example, in a shoulder pain reduction study \citep{heald1997shoulder}, because of the higher expectation for complete recovery, the healthier patients with mild pain at baseline often deemed greater pain reduction as ``meaningful'' than the ones who suffered from chronic disease.
Therefore, it is of scientific interest to generally estimate the individualized MCID (iMCID) based on each individual patient's clinical profile as well as to quantify the uncertainties of those estimates.

Nowadays, there is an increasing use and advancing development of EHR-based (electronic health records) studies in clinical research.
The EHR data are complex, diverse and high-dimensional \citep{abdullah2020visual}.
The rich information contained in the EHR data could facilitate the determination and quantification of iMCID.
Therefore, there is a pressing need to develop statistical methods that incorporate the high-dimensional data into both magnitude determination and uncertainty quantification of iMCID.

\subsection{Problem Formulation}

To facilitate the presentation, we first introduce some notation.
Let $X\in \RR$ be a continuous variable representing the score change collected from the PRO, e.g., the WOMAC pain score change from baseline to one year after surgery in the ChAMP trial.
Let $Y=\pm 1$ be a binary variable derived from the patient's response to the anchor question, where $Y=1$ represents an improved health condition and $Y=-1$ otherwise.
We use a $d$-dimensional vector $\bZ$ to denote the patient's clinical profile including demographic variables, clinical biomarkers, disease histories, among many others.
Suppose the data we observe are $n$ i.i.d. samples $\{(x_i,y_i,\bz_i)\}_{i=1}^n$ of $(X,Y,\bZ)$.
We focus on the high-dimensional setting, i.e., $d\gg n$.

Firstly, if there were no covariate $\bZ$, the MCID can be estimated by $\argmax_\tau \{\PP(X\geq \tau\mid Y=1) + \PP(X<\tau\mid Y=-1)\}$,
which is equivalent to
\beq\label{eq:pMCID}
\argmin_\tau ~ \EE[w(Y)L_{01}\{Y(X-\tau)\}],
\eeq
where $L_{01}(u)=\frac12\{1-\sign u\}$ is the 0-1 loss, $\sign u=1$ if $u\geq0$ and $-1$ otherwise, $w(1)=1/\pi$, $w(-1)=1/(1-\pi)$ and $\pi=\PP(Y=1)$.
When the high-dimensional covariate $\bZ$ is available, as the focus of this paper, the natural idea is to consider the iMCID with a functional form of $\bZ$, say $\tau(\bZ)$. In clinical practice, a simple structure, such as linear, is preferred due to its transparency and convenience for interpretation, especially for high-dimensional data.
Therefore, we focus on the linear structure $\tau(\bZ)=\bbeta^T\bZ$ in this paper.
{
The objective thus becomes
\beq\label{eq_risk_ori}
\bbeta^* = \argmin_{\bbeta}R(\bbeta), ~~\textrm{where}~~R(\bbeta)=\EE\big[
w(Y)L_{01}\{Y(X - \bbeta^{T}\bZ)\}
\big],
\eeq
and the expectation is with respect to the joint distribution of $(X, Y, \bZ)$.
}
Throughout this paper we assume that $\bbeta^*$ exists and is unique---the existence and uniqueness can be verified under specific models; see Section S3.1 in the Supplement for details.
Denote $\bbeta^* = (\theta^*,\bgamma^{*T})^T$, where $\theta^*$ is an arbitrary one-dimensional component of $\bbeta^*$ and $\bgamma^*$ represents the rest of the parameter which is high-dimensional.
In this paper, we start from considering the hypothesis testing procedure for the parameter $\theta^*$. With a simple reparametrization, the same procedure can be applied to infer the iMCID $\bc_0^T\bbeta^*$ for some  fixed and known vector $\bc_0\in\RR^d$.


It is worthwhile to mention that, although the motivation of this paper is to study iMCID, our formulation of this problem can be similarly applied to other scenarios as well, such as the covariate-adjusted Youden index \citep{xu2014model}, one-bit compressed sensing \citep{boufounos20081}, linear binary response model \citep{manski1975maximum,manski1985semiparametric}, and personalized medicine \citep{wang2018quantile}. Interested readers could refer to \cite{feng2022nonregular} for those examples.

\subsection{From Estimation to Inference}

Incorporating high-dimensional data in the objective, i.e., moving forward from (\ref{eq:pMCID}) to (\ref{eq_risk_ori}), is not trivial, even for the purpose of estimation only.
Recently, \cite{ban2019} established the rate of convergence of the (penalized) maximum score estimator for (\ref{eq_risk_ori}) in growing dimension {that $d$ is allowed to grow with $n$.}
In a related work, \cite{feng2022nonregular} proposed a regularized empirical risk minimization framework with a smoothed surrogate loss for {estimating the high-dimensional parameter $\bbeta^*$, and showed the estimation problem is nonregular in that there do not exist estimators of $\bbeta^*$ with root-$n$ convergence rate uniformly over a proper parameter space.}

Under (\ref{eq_risk_ori}), developing a valid statistical inference procedure is challenging, even for fixed dimensional setting.
\cite{manski1975maximum,manski1985semiparametric} considered the binary response model $Y = \sign{X-\bZ^T\bbeta + \epsilon}$, where $\epsilon$ may depend on $(X,\bZ)$ but with $\text{Median}(\epsilon|X,\bZ) = 0$. It can be shown that the true coefficient $\bbeta^*$ can be equivalently defined via (\ref{eq_risk_ori}) with $w(-1)=w(1)=1/2$. The maximum score estimator is proposed  to estimate $\bbeta$, and is later shown to have  a non-Gaussian limiting distribution \citep{kim1990cube}.
To tackle the challenge of nonstandard limiting distribution of the maximum score estimator, \cite{horowitz1992smoothed} proposed the smoothed maximum score estimator which is asymptotically normal in fixed dimension.

On top of the nonregularity of the problem (\ref{eq_risk_ori}), the high-dimensionality of the parameter adds an
additional layer of complexity for inference.
The reason is that the estimator that minimizes the penalized loss function does not have a tractable standard limiting distribution under high dimensionality, due to the bias induced by the penalty term.
For regular models (e.g., generalized linear models), there is a growing literature on correcting the bias from the penalty for valid inference, such as \cite{javanmard2014confidence,zhang2014confidence,van2014asymptotically,belloni2015uniform,ning2017general,cai2017confidence,fang2017testing,neykov2018unified,feng2019high,fang2020test},
among others.
Their main idea is to firstly construct a consistent estimator of the high dimensional parameter via proper regularization, and then remove the bias (via debiasing or decorrelation) in order to develop valid inferential statistics.
While these methods enjoy great success under regular models, it remains unclear whether they can be applied to conduct valid inference in nonregular models
such as the problem we consider in this paper.
To the best of our knowledge, our work is the first that provides valid  inferential tools for nonregular models in high dimension.

\subsection{Our Contributions}

In this paper, we propose a unified hypothesis testing framework for the one-dimensional parameter $\theta^*$ as well as for the iMCID encoded as a linear combination of $\bbeta^*$. We start from considering the hypothesis testing problem $H_0: \theta^* = 0$ versus $H_1: \theta^* \neq 0$, where we treat $\bgamma^*$ as
a high-dimensional nuisance parameter. Built on the smoothed surrogate estimation framework \citep{feng2022nonregular}, we propose a bias corrected smoothed decorrelated score to form the score test statistic.

There are several new ingredients in the construction of our score statistic. First, the score function is derived based on a smoothed surrogate loss to overcome the nonregularity due to the nonsmoothness of the 0-1 loss. Second, unlike the existing works on high-dimensional inference, the score function from the smoothed surrogate loss is asymptotically biased. By explicitly estimating the bias term, we derive a new bias corrected score. Third, the decorrelation step, developed by \cite{ning2017general} for regular models, is applied to reduce the uncertainty of estimating high-dimensional nuisance parameters. Compared to \cite{ning2017general},  the adoption of the smoothed loss and the corresponding bias correction step are new, which also make our  inference much more challenging than the existing works.  Theoretically, we show that under some conditions, the proposed score test statistic converges in distribution to a standard Gaussian distribution under the null hypothesis. We further establish the local asymptotic power of the test statistic when $\theta^*$ deviates from $0$ in a local neighborhood. In particular, we give the conditions under which the test statistic has asymptotic power one.

When constructing  the bias corrected smoothed decorrelated score, we need to specify a bandwidth parameter,  whose optimal choice depends on the unknown smoothness of the data distribution.  We further propose a  double-smoothing approach to select the optimal bandwidth by minimizing the mean squared error (MSE) of the score function. To our knowledge, such bandwidth selection procedures have not been studied for high-dimensional models.
We show that under some extra smoothness assumptions, the ratio of the data-driven bandwidth to the theoretically optimal bandwidth converges to one in probability. Moreover, the proposed score test statistic with the data-driven bandwidth still converges in distribution to a standard Gaussian distribution under the null hypothesis.

\subsection{Paper Structure and Notation}

The organization of this paper is as follows.
In Section \ref{mainmethod}, we first provide some background on the estimation of iMCID then introduce the bias corrected smoothed deccorelated score and the associated test statistic.
In Section \ref{theory}, we discuss the theoretical properties of the score test. The data-driven bandwidth selection is addressed in Section \ref{adaptivity}.
The corresponding results for $\bc_0^T\bbeta^*$, the linear combination of parameters, are briefly summarized in Section~\ref{sec:imcid}.
Sections \ref{simulation} and \ref{sec_data} contain simulation studies and a real data example, respectively.
All the technical details and proofs are contained in the Supplement.

Throughout the paper, we adopt the following notation.
	For any set $\cS$, we write $|\cS|$ for its cardinality.
	For any vector $\bv\in \RR^d$, we use $\bv_{\cS}$ to denote the subvector of $\bv$ with entries indexed by the set $\cS$, and define its $\ell_q$ norm as $\|\bv\|_q = (\sum_{j=1}^d |\bv_j|^q)^{1/q}$ for some real number $q\ge 0$. For any matrix $\bM \in \RR^{d_1 \times d_2}$, we denote $\norm{\bM}_{\max} = \max_{i,j}|M_{ij}|$. For any two sequences $a_n$ and $b_n$, we write $a_n \lesssim b_n$  if there exists some positive constant $C$ such that $a_n \le Cb_n$ for any $n$. We let $a_n \asymp b_n$ stand for $a_n\lesssim b_n$ and $b_n \lesssim a_n$. Denote $a\vee b=\max (a,b)$ and $a\wedge b=\min(a,b)$. 
{
	For function $F(\theta,\bgamma)$, we denote  $\nabla_{\theta}F(\theta,\bgamma)$ and $\nabla_{\bgamma}F(\theta,\bgamma)$ as the first order derivatives, and  $\nabla^2_{\theta,\theta}F(\theta,\bgamma)$ the second order derivative.
}

\section{Methodology}\label{mainmethod}
\subsection{Review of Penalized Smoothed Surrogate Estimation}\label{background}

Under high dimensionality that $d\gg n$, estimating $\bbeta^*$ via the empirical risk minimization (\ref{eq_risk_ori}) induces challenges from both statistical and computational perspectives.
The non-smoothness of $L_{01}(u)$ would cause the estimator to have a nonstandard convergence rate, which  happens even in the fixed low dimensional case \citep{kim1990cube}.
Moreover, minimizing the empirical risk function based on the 0-1 loss is computationally NP-hard and is often very difficult to implement.
To tackle these challenges, \cite{feng2022nonregular} considered the following smoothed surrogate risk
\beq \label{eq_risk_smooth}
R_{\delta}(\bbeta) = \EE\bigg[
w(Y)L_{\delta,K}\big\{
Y(X-\bbeta^T\bZ)
\big\}
\bigg],
\eeq
where $L_{\delta,K}(u) = \int_{u/\delta}^{\infty}K(t)dt$ is a smoothed approximation of $L_{01}(u)$, $K$ is a kernel function defined in Section \ref{theory} and $\delta >0 $ is a bandwidth parameter. As the bandwidth $\delta$ shrinks to 0, $L_{\delta,K}(u)$ converges pointwisely to $L_{01}(u)$ (for any $u\neq 0$), from which it can be shown that $\bbeta^*$ also minimizes the smoothed risk $R_{\delta}(\bbeta)$ up to a small approximation error. They further proposed the following penalized smoothed surrogate estimator
\beq\label{eq_hatbeta}
\wh{\bbeta} := \argmin_{\bbeta} R_\delta^n(\bbeta)+ P_{\lambda}(\bbeta),
\eeq
where $P_{\lambda}(\bbeta)$ is some sparsity inducing penalty (e.g., Lasso) with a tuning parameter $\lambda$, and
$R_\delta^n(\bbeta)$ is the corresponding empirical risk
\beq\label{eq_erisk_smooth}
R_{\delta}^n(\bbeta) = \frac 1n \sum_{i = 1}^{n}\bar{R}^i_{\delta}(\bbeta) = \frac 1n\sum_{i = 1}^{n}w(y_i)L_{\delta,K} \Big(y_i(x_i - \bbeta^T\bz_i)\Big).
\eeq
Computationally, the empirical surrogate risk $R_{\delta}^n(\bbeta)$ is a smooth function of $\bbeta$, which renders the optimization more tractable. Statistically, under some conditions, the estimator $\wh{\bbeta}$ is shown to be rate-optimal, i.e., the convergence rate of $\wh{\bbeta}$ matches the minimax lower bound up to a logarithmic factor. We refer to \cite{feng2022nonregular} for the detailed results.

\subsection{Bias Corrected Smoothed Decorrelated Score}\label{method}

While \cite{feng2022nonregular} showed that the penalized smoothed surrogate estimator $\wh\bbeta$ is consistent, it does not automatically equip with a practical inferential procedure for $\bbeta^*$, mainly because of the sparsity inducing penalty.
In practice, how to draw valid statistical inference is often the ultimate goal. In our motivating example, it is of critical importance to quantify the uncertainty of $\bc_0^T\bbeta^*$ where $\bc_0$ represents the realized value of a new patient's clinical profile. In other words, we would like to develop a testing procedure for
\beq\label{eq:test MCID}
H_{0L}: \bc_0^T\bbeta^* = 0 \mbox{ versus } H_{1L}: \bc_0^T\bbeta^* \neq 0.
\eeq
In this section, we focus on a special case of (\ref{eq:test MCID}), the hypothesis test for $\theta^*$,
\beq\label{eq:test theta}
H_0: \theta^* = 0 \mbox{ versus } H_1: \theta^* \neq 0,
\eeq
where we treat $\bgamma$ as the nuisance parameter.
Once the results for (\ref{eq:test theta}) are clear, we can extend them to (\ref{eq:test MCID}), to be presented in Section~\ref{sec:imcid}.

For (\ref{eq:test theta}), we propose a new bias corrected smoothed decorrelated score test.
It is well known that the classical score test is constructed based on the magnitude of the gradient of the loglikelihood, or more generally, the empirical risk function associated with $R(\bbeta)$ in (\ref{eq_risk_ori}).  However, this construction breaks down in our problem due to the following two reasons.


First, to construct the score statistic, one needs to plug in some estimate of the nuisance parameter $\bgamma$ such as $\wh{\bgamma}$ obtained by partitioning $\wh{\bbeta} = (\wh{\theta},\wh{\bgamma}^T)^T$ in (\ref{eq_hatbeta}). However, since $\bgamma$ is a high-dimensional parameter, the estimation error from $\wh{\bgamma}$ may become the leading term in the asymptotic analysis of the score function.
To deal with the high-dimensional nuisance parameter, we use the decorrelated score, where the key idea is to project the score of the parameter of interest to a high-dimensional nuisance space \citep{ning2017general}.
On the population level, it takes the form
\beq\label{eq_pop_decor}
\nabla_\theta R(\theta,\bgamma)-\bomega^{*T}\nabla_{\bgamma} R(\theta,\bgamma),
\eeq
where the decorrelation vector is $\bomega^* = \big(\nabla^2_{\bgamma,\bgamma}R(\bbeta^*)\big)^{-1}\nabla^2_{\bgamma,\theta}R(\bbeta^*)$.
When $R(\theta,\bgamma)$ corresponds to the expected loglikelihood function of the data, the definition of $\bomega^*$ coincides with that in \cite{ning2017general}.
In general, however, $R(\theta,\bgamma)$ is not always the loglikelihood function, so we define $\bomega^*$ as $\big(\nabla^2_{\bgamma,\bgamma}R(\bbeta^*)\big)^{-1}\nabla^2_{\bgamma,\theta}R(\bbeta^*)$ in order to mitigate the bias from estimating $\bgamma$.
We refer to the review paper \citep{neykov2018unified} for further discussions.

Second, even if the above decorrelated score approach can successfully remove the effect of the  high-dimensional nuisance parameter, one cannot construct the sample based decorrelated score from (\ref{eq_pop_decor}), as the sample version of $R(\theta,\bgamma)$ is non-differentiable, leading to the so called non-standard inference. To circumvent this issue, we approximate $R(\theta,\bgamma)$ in (\ref{eq_pop_decor}) by the smoothed surrogate risk $R_\delta(\theta,\bgamma)$ in (\ref{eq_risk_smooth}), that is
\beq\label{eq_pop_decor2}
\nabla_\theta R(\theta,\bgamma)-\bomega^{*T}\nabla_{\bgamma} R(\theta,\bgamma)=\big\{\nabla_\theta R_\delta(\theta,\bgamma)-\bomega^{*T}\nabla_{\bgamma} R_\delta(\theta,\bgamma)\big\}-\textrm{approximation bias}.
\eeq
Since the empirical version of $R_\delta(\theta,\bgamma)$ is smooth, we define the (empirical) smoothed decorrelated score function as $S_\delta(\theta,\bgamma) = \nabla_{\theta} R^n_\delta(\theta,\bgamma) - \bomega^{*T}\nabla_{\bgamma} R^n_\delta(\theta,\bgamma)$.
With $\bgamma$ estimated by $\wh\bgamma$, the estimated score function is then naturally defined as
\beq\label{eq_estimated_score}
\wh{S}_\delta (\theta,\wh{\bgamma}) =  \nabla_{\theta} R^n_\delta(\theta,\wh{\bgamma}) - \wh{\bomega}^T\nabla_{\bgamma} R^n_\delta(\theta,\wh{\bgamma}),
\eeq
where $\wh{\bomega}$, to be defined more precisely in Section \ref{sec_details}, is an estimator of $\bomega^*$.

In view of (\ref{eq_pop_decor2}) and (\ref{eq_estimated_score}), the sample version of $\nabla_\theta R_\delta(\theta,\bgamma)-\bomega^{*T}\nabla_{\bgamma} R_\delta(\theta,\bgamma)$ is given by $\wh{S}_\delta (\theta,\wh{\bgamma})$ and therefore, to construct a valid score function, it remains to estimate the approximation bias in (\ref{eq_pop_decor2}). To proceed,
{
we first analyze the population version of this approximation bias,
}
which is simply $\bv^{*T}\nabla R_\delta(\bbeta^*)$ at $\bbeta=\bbeta^*$, where $\bv^* = (1,-\bomega^{*T})^T$. After some analysis, we can show that the magnitude of the approximation bias depends on the smoothness of $f(x|y,\bz)$, the conditional density of $X$ given $Y$ and $\bZ$. To obtain an explicit form of the approximation bias, we assume that $f(x|y,\bz)$ is $\ell$th order differentiable for some $\ell\geq 2$, to be defined more precisely in Section \ref{theory}. Under this assumption, we can show that as the bandwidth parameter $\delta\rightarrow 0$,
$\bv^{*T}\nabla R_\delta(\bbeta^*) = \delta^{\ell}\mu^*(1 + o(1))$,
where
\begin{eqnarray}
\mu^*:=\bv^{*T}\bb^*&=&\bv^{*T}
\Big(\int K(u)\frac{u^\ell}{\ell!}du\Big) \sum_{y \in \{-1,1\}}w(y)\int y\bz f^{(\ell)}(\bbeta^{*T}\bz|y,\bz)f(y,\bz) d\bz,\nonumber\\
&=&
\underbrace{\Big(\int K(u)\frac{u^\ell}{\ell!}du\Big)}_{\gamma_{K,\ell}} \bv^{*T}\underbrace{\EE \Big[w(Y)Y\bZ f^{(\ell)}(\bbeta^{*T}\bZ|Y,\bZ)\Big]}_{T^{(\ell)}(\bbeta^*)},\label{eq:definemustar}
\end{eqnarray}
and $f^{(\ell)}(x|y,\bz)$ denotes the $\ell$th order derivative of $f(x|y,\bz)$ with respect to $x$.

To estimate the approximation bias $\bv^{*T}\nabla R_\delta(\bbeta^*)$, it suffices to estimate $\mu^*$. From (\ref{eq:definemustar}), once $f^{(\ell)}(x|y,\bz)$ at $x=\bbeta^{*T}\bz$ is estimated, we can construct a plug-in estimator for $\mu^*$. To be specific, assume that a pilot kernel estimator with some kernel function $U$  and bandwidth $h$ is available to estimate $f^{(\ell)}(\bbeta^{*T}\bz|y,\bz)$. Then we can estimate $\mu^*$ by
\beq\label{pilot_bias}
\wh{\mu}= \gamma_{K,\ell}\wh{\bv}^{T}\wh{T}_{h,U}^{(\ell),n}(\wh{\bbeta}),
\eeq
where $\wh{T}_{h,U}^{(\ell),n}(\wh{\bbeta}):= \frac 1n \sum_{i=1}^n w(y_i)y_i \frac{ \bz_{i}}{h^{1+\ell}}U^{(\ell)}\big(\frac{ \wh{\bbeta}^{T}\bz_i - x_i}{h}\big)$
and $\wh\bv = (1,-\wh\bomega^{T})^T$.

The last step to construct a valid score test is to find the asymptotic variance of the smoothed decorrelated score $S_\delta(\bbeta^*)$. Lemma \ref{lemma_ori_normality_new} in the next section shows that the asymptotic variance of the standardized decorrelated score  $(n\delta)^{1/2}S_{\delta}(\bbeta^*)$ is $\sigma^{*2} = \bv^{*T}\bSigma^*\bv^*$, where
\begin{eqnarray}
\bSigma^* &:=& \sum_{y \in \{-1,1\}}w(y)^2\int\bz\bz^T\int K(u)^2duf(\bbeta^{*T}\bz|y,\bz)f(y,\bz)d\bz,\nonumber\\
&=& \underbrace{\Big(\int K(u)^2du\Big)}_{\tilde{\mu}_K} \underbrace{\EE\Big[w(Y)^2\bZ\bZ^T f(\bbeta^{*T}\bZ|Y,\bZ)\Big]}_{H(\bbeta^*)},\label{eq:defineSigmastar}
\end{eqnarray}
and thus $\sigma^*$ can be estimated by
\beq\label{eq_est_var}
\wh{\sigma} = \sqrt{\tilde{\mu}_K \wh{\bv}^T\wh{H}^{n}_{g,L}(\wh{\bbeta})\wh{\bv}},
\eeq
where
$\wh{H}_{g,L}^{n}(\wh{\bbeta}) =
\frac{1}{n}\sum_{i=1}^n w^2(y_i)\bz_i\bz_i^T\frac 1 g L(\frac{x_i - \wh{\bbeta}^T\bz_i}{g})$
with some kernel function $L$ and bandwidth $g$. In Section S4 in the Supplement, we propose an alternative kernel-free estimator of $\sigma^*$, which does not require any additional kernel function or bandwidth. We show that the estimator is still consistent for $\sigma^*$ but may have a slower convergence rate than $\wh{\sigma}$ here.

Equipped with the smoothed decorrelated score $\wh{S}_\delta (\theta,\wh{\bgamma})$ in (\ref{eq_estimated_score}), the estimate of the approximation bias $\delta^{\ell}\wh{\mu}$ in (\ref{pilot_bias}) and the estimate of the asymptotic variance $\wh{\sigma}^2$ in (\ref{eq_est_var}),  we define the bias corrected smoothed decorrelated score statistic as
\beq\label{eq_test_statistic_1}
\wh{U}_n = \sqrt{n\delta}\Big(\frac{\wh{S}_\delta(0,\wh{\bgamma}) - \delta^{\ell}\wh{\mu}}{\wh{\sigma}}\Big).
\eeq
\begin{remark}
Compared to the existing decorrelated score approach \citep{ning2017general}, our methodological innovation  is to develop an explicit bias correction step to remove the approximation bias in (\ref{eq_pop_decor2}) induced by the smoothed surrogate risk. From the theoretical aspect, our test statistic $\wh{U}_n$ is rescaled by $(n\delta)^{1/2}$ rather than the classical $n^{1/2}$ factor, which leads to the non-standard rate of the decorrelated score not only under the null but also under local alternatives; see Section \ref{theory}.
\end{remark}



\subsection{Detailed Implementation}\label{sec_details}

For numerical implementation, we follow the path-following algorithm presented in \cite{feng2022nonregular} to compute the initial estimator $\wh\bbeta$. For the estimator $\wh\bomega$, recall that $ \bomega^*$ satisfies  $\nabla^2_{\bgamma,\bgamma}R(\bbeta^*)\bomega^*=\nabla^2_{\bgamma,\theta}R(\bbeta^*)$. Since $\nabla^2 R(\bbeta^*)$ can be approximated by the Hessian of the smoothed surrogate loss $\nabla^2 R^n_{\delta}(\bbeta^*)$, we consider the following Dantzig type estimator $\wh{\bomega}$, where
\beq\label{eq_def_dantzig}
\wh{\bomega} = \argmin_{\bomega} \norm{\bomega}_1 ~~~~~~~s.t.~~ \norm{\nabla^2_{\bgamma,\theta} R^n_{\delta}(\wh{\bbeta}) - \nabla^2_{\bgamma,\bgamma} R^n_{\delta}(\wh{\bbeta})\bomega }_\infty \leq \lambda',
\eeq
for some tuning parameter $\lambda'>0$. 

For implementing $\wh{U}_n$, we note that the analysis of the asymptotic distribution of $\wh{U}_n$ is complicated by the dependence between the estimator $\wh{\bbeta}$ and $S_\delta(\theta,\bgamma)$. To decouple the dependence and ease theoretical development, we apply the cross-fitting technique to construct the bias corrected smoothed decorrelated score. Specifically, instead of utilizing the same set of samples for estimating $\wh{\bbeta},\wh{\bomega}$ and constructing the score function $S_\delta(\theta,\bgamma)$, we will firstly estimate $\wh{\bbeta}$ using one set of samples, and then use the rest of samples for estimating $\wh{\bomega}$ and constructing $S_\delta(\theta,\bgamma)$.
We can further switch the samples and aggregate the decorrelated score.
Without loss of generality, assume the sample size $n$ is even and we divide the samples into two halves with equal size for this purpose. Formally, denote $\wh{\bbeta}^{(i)},\wh{\bomega}^{(i)},i=1,2$ as the estimator based on the $i$th fold of the samples, $\cN_i$, and similarly $\nabla R_\delta^{n_{(i)}}(\bbeta), \nabla^2 R_\delta^{n_{(i)}}(\bbeta)$ as the corresponding gradient and Hessian. Define
$$
\wh{S}^{(1)}_\delta(\theta,\wh{\bgamma}^{(2)})=\nabla_{\theta} R^{n_{(1)}}_\delta(\theta,\wh{\bgamma}^{(2)}) - \wh{\bomega}^{(1)T}\nabla_{\bgamma} R^{n_{(1)}}_\delta(\theta,\wh{\bgamma}^{(2)}),
$$
and $\wh{S}^{(2)}_\delta(\theta,\wh{\bgamma}^{(1)})$ in a similar way. The estimated decorrelated score via cross-fitting is
\beq\label{eq_cross_dscore}
\wh{S}_\delta(\theta,\wh{\bgamma}) =\frac12\big(\wh{S}^{(1)}_\delta(\theta,\wh{\bgamma}^{(2)}) + \wh{S}^{(2)}_\delta(\theta,\wh{\bgamma}^{(1)})\big ).
\eeq
Similarly, we define the cross-fitted estimators $\wh{\mu}$ and $\wh{\sigma}$ as
\beq
\wh{\mu}=& \frac 12 \gamma_{K,\ell} (\wh{\bv}^{(1)T}\wh{T}_{h,U}^{(\ell),n_{(1)}}(\wh{\bbeta}^{(2)})+ \wh{\bv}^{(2)T}\wh{T}^{(\ell),n_{(2)}}_{h,U}(\wh{\bbeta}^{(1)})),\\
\wh{\sigma}^2 =& \frac {\tilde{\mu}_K}{2} \bigg[\wh{\bv}^{(1)T}\wh{H}^{n_{(1)}}_{g,K}(\wh{\bbeta}^{(2)})\wh{\bv}^{(1)}
+ \wh{\bv}^{(2)T}\wh{H}^{n_{(2)}}_{g,K}(\wh{\bbeta}^{(1)})\wh{\bv}^{(2)}
\bigg],\label{eq_cross_mu}
\eeq
where
\beq
\wh{T}_{h,U}^{(\ell),n_{(1)}}(\wh{\bbeta}^{(2)})=& \frac{1}{|\cN_1|}\sum_{i\in \cN_1} w(y_i)y_i \frac{ \bz_{i}}{h^{1+\ell}}U^{(\ell)}\bigg(\frac{\wh{\bbeta}^{(2)T}\bz_i - x_i}{h}\bigg),\\
\wh{H}_{g,L}^{n_{(1)}}(\wh{\bbeta}^{(2)}) =&
\frac{1}{|\cN_1|}\sum_{i\in \cN_1} w^2(y_i)\bz_i\bz_i^T\frac 1 g L(\frac{x_i - \wh{\bbeta}^{(2)T}\bz_i}{g}),
\eeq
and similarly for
$\wh{T}_{h,U}^{(\ell),n_{(2)}}(\wh{\bbeta}^{(1)}), \wh{H}_{g,L}^{n_{(2)}}(\wh{\bbeta}^{(1)})$. Given $\wh{S}_\delta(\theta,\wh{\bgamma})$ in (\ref{eq_cross_dscore}) and the above estimators $\wh{\mu}$ and $\wh{\sigma}$, we can form the score test statistic $\wh{U}_n$ in the same way as in (\ref{eq_test_statistic_1}).

\section{Theory}\label{theory}
\subsection{Assumptions}

In this paper, we consider the following definition of function smoothness.

\begin{definition}\label{def_smooth}
	We say the conditional density $f(x|y,\bz)$ of $X$ given $Y,\bZ$ is $\ell$th order smooth, if for any $\bz$ and $y\in \{-1,1\}$, the conditional density $f(x|y,\bz)$ is $\ell$-times continuously differentiable in $x$ with derivatives $f^{(i)}(x|y,\bz)$ bounded by a constant $C$, $|f^{(i)}(x|y,\bz)|\leq C$ for $i = 1,\dotso,\ell$, and  $f^{(\ell)}(x|y,\bz)$ is H\"older continuous with some exponent $0 < \zeta \leq 1$, that is, for any $\bz, \triangle$ and $y\in \{-1,1\}$, $|f^{(\ell)}(x+\triangle|y,\bz) - f^{(\ell)}(x|y,\bz)| \leq L \triangle^\zeta$,
	where $L>0$ is some constant.
\end{definition}

\begin{assumption}\label{ass_smooth}
	We assume $f(x|y,\bz)$ is $\ell$th order smooth with some integer $\ell \geq 2$.
\end{assumption}

Assumption \ref{ass_smooth} concerns the smoothness of $f(x|y,\bz)$. To see why the smoothness condition is important, notice that the gradient functions of (\ref{eq_risk_smooth}) and (\ref{eq_risk_ori}) are
\beq
\nabla R_{\delta}(\bbeta) =& \sum_{y \in \{-1,1\}} w(y) \int yz \Big[\int \frac{1}{\delta}K(\frac{y(x - \bbeta^T\bz)}{\delta})f(x|y,\bz)dx\Big] f(y,\bz) d\bz\\
\nabla R(\bbeta) =& \sum_{y \in \{-1,1\}} w(y) \int yzf(\bbeta^T\bz|y,\bz)f(y,\bz)d\bz,
\eeq
from which we can see that $f(\bbeta^T\bz|y,\bz)$ in $\nabla R(\bbeta)$ is substituted by its kernel approximation $\int \frac{1}{\delta}K(\frac{y(x - \bbeta^T\bz)}{\delta})f(x|y,\bz)dx$, and thus the difference between $\nabla R_{\delta}(\bbeta)$ and $\nabla R(\bbeta)$ naturally depends on the smoothness of $f(x|y,\bz)$.

Notice that our smoothness condition in Definition \ref{def_smooth} is slightly stronger than the standard H\"older smoothness condition in the nonparametric literature \citep{tsybakovintroduction}. In particular, we require that $f^{(\ell)}(x|y,\bz)$ is H\"older continuous with some exponent $0 < \zeta \leq 1$. This additional assumption is essential to show the rate of the bias estimator $\hat\mu$ in (\ref{pilot_bias}). The H\"older class condition in Assumption \ref{ass_smooth} can be relaxed to a variation of Nikol'ski class condition \citep{tsybakovintroduction}; see Section S3.2 in the Supplement for details.

\begin{assumption}\label{ass_kernel}
	We assume $K(t)$ is a kernel function with bounded support that satisfies: $K(t) = K(-t)$, $|K(t)| \leq K_{\max} < \infty \; \forall\; t\in \RR$, $\int K(t)dt = 1$, $\int K^2(t)dt < \infty$, and $|K'|<\infty$. We also assume that $K$  degenerates at the boundaries. A kernel is said to be of order $\ell \geq 1$ if it satisfies
	$\int t^jK(t)dt = 0, \;\forall\;j = 1,\dotso,\ell-1$, $\int t^\ell K(t)dt \neq 0$, and $\int |t|^q |K(t)|dt $ are bounded by a constant for any $q\in [\ell,\ell+1]$.
\end{assumption}

{
Assumption~\ref{ass_kernel} above is about the kernel function $K(t)$ that we first introduced in the surrogate risk $R_\delta(\bbeta)$ in (\ref{eq_risk_smooth}).
We provide a list of commonly-seen second-order, fourth-order and sixth-order kernel functions in Section S3.3 of the Supplement.
}

We now impose regularity conditions on $(X,Y,\bZ)$.

\begin{assumption}\label{ass_proportion}
	There exists a constant $c > 0$ such that $c \leq \PP(Y = 1)\leq 1 -c$ and the weight function $w(\cdot)$ is positive and upper bounded by a constant.
\end{assumption}

\begin{assumption}\label{ass_moment}
	We assume $\max_{1\leq j\leq d} |Z_j| \leq M_n$ for some $M_n$ that possibly depends on $n$, where $M_n^2 \leq C\sqrt{n\delta/\log(d)}$ for some constant $C > 0$.
	We also assume that $\EE[|Z_j|^4|Y = y]$ is bounded by a constant for $y\in \{1,-1 \}$.
\end{assumption}


\begin{assumption}\label{ass_projection_norm}
	We assume $\sigma^* = \sqrt{\bv^{*T}\bSigma^*\bv^*}$ is bounded away from 0 and infinity by some constants, and $|\mu^*|=|\bv^{*T}\bb^*|$ is also upper bounded by a constant.
\end{assumption}
Assumption \ref{ass_moment} requires the boundedness of $\bZ$ and the fourth order moment.
{
Notice that if each component of $\bZ$ is sub-Gaussian with bounded sub-Gaussian norm, Assumption \ref{ass_moment} is satisfied with high probability with $M_n \asymp \sqrt{\log d}$ providing $(\log d)^3/(n\delta) = O( 1)$ which is a mild assumption.
For binary covariates $Z_j\in \{0,1\}$, it holds that $M_n=1$.
}
Assumption \ref{ass_projection_norm} ensures that the asymptotic variance of the smoothed decorrelated score $\sigma^*$ does not degenerate and the approximation bias $\mu^*$ is bounded. In Section S3.4 in the Supplement, we verify that under mild conditions, Assumptions \ref{ass_smooth}-\ref{ass_projection_norm} hold under the binary response model.
Finally, we impose the following assumption on  the estimators of $\bbeta^*$ and $\bomega^*$.

\begin{assumption}\label{ass_estimators}
	Assume there are estimators $\wh{\bbeta}$ and $\wh{\bv} = (1,-\wh{\bomega}^T)^T$ with
	\[
	\norm{\betadiff}_1 \lesssim \eta_1(n) ~~~~ \text{and} ~~~~ \norm{\wh{\bv} - \bv^*}_1/\norm{\bv^*}_1 \lesssim \eta_2(n),
	\]
	for some non-random sequences $\eta_1(n),\eta_2(n)$ converging to $0$ as $n\rightarrow \infty$.
	
\end{assumption}

It is shown by \citet{feng2022nonregular} that, under some conditions, the estimator $\wh{\bbeta}$ in (\ref{eq_hatbeta}) achieves the (near) minimax-optimal rate   $\eta_1(n)=\sqrt{s}(\frac{s\log(d)}{n})^{\ell/(2\ell+1)}$
where $s=\|\bbeta^*\|_0$. For $\wh{\bv} = (1,-\wh{\bomega}^T)^T$, we assume $\norm{\wh{\bv} - \bv^*}_1\lesssim \norm{\bv^*}_1\eta_2(n)$. Notice that the term $\norm{\bv^*}_1$ is not absorbed into $\eta_2(n)$ only for notational simplicity.
In Lemma S7 in the Supplement, we show that a Dantzig type estimator $\wh{\bv}$ could attain the fast rate $\eta_2(n)$.
	
\subsection{Theoretical Results}

We start from the following lemma which characterizes the asymptotic distribution of the decorrelated score function evaluated at the true parameter $\bbeta^*$.
\begin{lemma}\label{lemma_ori_normality_new}
	Under Assumptions \ref{ass_smooth} - \ref{ass_projection_norm}, if $(\norm{\bv^*}_1M_n)^3/(n\delta)^{1/2} = o(1)$ and $\delta = o(1)$, then
	\beq
	\sqrt{n\delta}\frac{\bv^{*T}(\nabla R_\delta^n(\bbeta^*) - \nabla R_\delta(\bbeta^*))}{\sqrt{\bv^{*T}\bSigma^*\bv^*}} \stackrel{d}{\rightarrow} N(0,1),
	\eeq
	where
	\beq\label{eq_bias}
	\bv^{*T}\nabla R_\delta(\bbeta^*) = \delta^{\ell}\bv^{*T}\bb^*(1 + o(1)).
	\eeq
\end{lemma}
Asymptotically, the bias and standard deviation of $\bv^{*T}\nabla R_\delta^n(\bbeta^*)$ can be seen from this lemma.
Since $\mu^*=\bv^{*T}\bb^*$ and $\sigma^*=\sqrt{\bv^{*T}\bSigma^*\bv^*}$ are both bounded by constants, the asymptotic bias and standard deviation are of order $\delta^\ell$ and $(n\delta)^{-1/2}$, respectively. Thus, choosing $\delta= c n^{-1/(2\ell + 1)}$ for any constant $c>0$ attains the optimal bias and variance trade-off. 
Note that in this lemma we require $(\norm{\bv^*}_1M_n)^3/(n\delta)^{1/2} = o(1)$  to verify the Lindeberg condition in the central limit theorem, which holds as long as $\delta$ does not shrink to zero too fast.




Our first main theorem characterizes the asymptotic normality of the decorrelated score under the null hypothesis with nuisance parameters $\bgamma^*$ and $\bomega^*$ estimated by those in Assumption \ref{ass_estimators}.

\begin{theorem}\label{theorem_score}
	Under Assumptions \ref{ass_smooth} - \ref{ass_estimators}, if $(\norm{\bv^*}_1M_n)^3/(n\delta)^{1/2} = o(1)$, $\frac{\log(d)}{n\delta^3} = o(1)$, $n\delta^{2\ell+1}= O(1)$, and
	\beq &(n\delta)^{1/2}\norm{\bv^*}_1\bigg(\frac{\eta_1(n)}{\delta}\vee \eta_2(n)\bigg)\bigg(\sqrt{\frac{\log(d)}{n\delta}} \vee \delta^{\ell} \vee M_n^2\eta_1(n)\bigg)  = o(1),\label{eq_theorem_score_condition}
	\eeq
	then under $H_0:\theta^* = 0$, it holds that
$\frac{\sqrt{n\delta}\wh{S}_\delta(0,\wh{\bgamma}) - \sqrt{n\delta^{2\ell+1}}\mu^*}{\sigma^*} \stackrel{d}{\rightarrow}  N(0,1)$.
\end{theorem}

Theorem \ref{theorem_score} implies that the decorrelated score with some high-dimensional plug-in estimators $\wh\bgamma$ and $\wh\bomega$ has the same asymptotic distribution as in Lemma \ref{lemma_ori_normality_new}. Several conditions are needed to show this result. The first condition $(\norm{\bv^*}_1M_n)^3/(n\delta)^{1/2} = o(1)$ is from Lemma \ref{lemma_ori_normality_new}, and the second condition $\frac{\log(d)}{n\delta^3} = o(1)$ is also mild as long as $\delta$ does not go to zero too fast. The third condition $n\delta^{2\ell+1}= O(1)$ guarantees that the higher order bias of the decorrelated score can be ignored and therefore it suffices to only correct for the leading bias term in (\ref{eq_bias}).

We now elaborate the condition (\ref{eq_theorem_score_condition}). Roughly speaking, the term $\sqrt{\frac{\log(d)}{n\delta}} \vee \delta^{\ell} \vee M_n^2\eta_1(n)$ comes from the bound for $\|\nabla R^{n_{(1)}}_\delta(\theta,\wh{\bgamma}^{(2)})-\nabla R(\theta,{\bgamma})\|_\infty$. Indeed, the cross-fitting technique guarantees the independence between $\wh{\bgamma}^{(2)}$ and $\nabla R^{n_{(1)}}_\delta(\theta,{\bgamma})$, which plays a key role in the analysis.  Condition (\ref{eq_theorem_score_condition}) simply means that this bound interacting with the estimation error of $\wh\bgamma$ and $\wh\bomega$ is sufficiently small.
We can further simplify the condition (\ref{eq_theorem_score_condition}) by plugging the order of $\eta_1(n)$ derived in \cite{feng2022nonregular} and $\eta_2(n)=s'(\log(d)/n)^{(\ell-1)/(2\ell + 1)}$ derived from Lemma S7 in the Supplement where $s'=\|\bomega^*\|_0$.

Recall that in our score statistic $\hat U_n$ in (\ref{eq_test_statistic_1}), we plug in the estimators $\hat\mu$ and $\hat\sigma$ for $\mu^*$ and $\sigma^*$. In Lemmas S8 and S9 in the Supplement, we establish the rate of convergence of  $\hat\mu$ and $\hat\sigma$. Under the assumption that $|\hat\mu-\mu^*|=o_p(1)$ and $|\hat\sigma-\sigma^*|=o_p(1)$, the Slutsky's theorem implies that the bias corrected decorrelated score statistic $\wh{U}_n\stackrel{d}{\rightarrow}  N(0,1)$ under the null hypothesis.

Accordingly, given the desired significance level $\alpha$, we define the test function as
\bq
T_{DS}=I(|\wh{U}_n|>\Phi^{-1}(1-\alpha/2)),
\eq
where $\Phi^{-1}(\cdot)$ is the inverse function of the cdf of the standard normal distribution. Thus, our result shows that the Type I error of the test $T_{DS}$ converges to $\alpha$ asymptotically, i.e., $\PP(T_{DS}=1| H_0)\rightarrow \alpha$.

Now denote
$\nabla^2_{\theta|\bgamma}R(\bbeta^*) = \nabla^2_{\theta\theta}R(\bbeta^*) - \nabla^2_{\theta\bgamma}R(\bbeta^*)(\nabla^2_{\bgamma\bgamma}R(\bbeta^*))^{-1}\nabla^2_{\bgamma\btheta}R(\bbeta^*)$.
Our second main theorem characterizes the limiting behavior of $\wh{U}_n$ under the local alternative hypothesis $H_1: \theta^* = \tilde{C}n^{-\phi}$ for some constants $\tilde C\neq 0$ and $\phi>0$.


\begin{theorem}\label{theorem_power}
Assume the conditions in Theorem \ref{theorem_score} and in Lemmas S8 and S9 of the Supplement, and further
	\beq\label{eq_power_condition}
	&\norm{\bv^*}_1^2M_n^4n^{1-4\phi}/\delta = o(1),\quad (n\delta)^{1/2}
	\norm{\bv^*}_1(\eta_1(n)\vee \eta_2(n))M_n n^{-\phi} = o(1),
	\eeq
and that $\wh{\mu}, \wh{\sigma}$ are consistent estimators of $\mu^*, \sigma^*$. Then, by choosing the optimal bandwidth $\delta \asymp n^{-1/(2\ell + 1)}$, the following results hold under the local alternative hypothesis $H_1: \theta^* = \tilde{C}n^{-\phi}$.
	\begin{enumerate}
		\item[1.] When $\phi=\frac{\ell}{2\ell+1}$, it holds that
$\wh{U}_n  \stackrel{d}{\rightarrow} N(-\xi,1)$,
where $\xi=\tilde{C}\nabla^2_{\theta|\bgamma}R(\bbeta^*)/\sigma^{*}$ is assumed to be a constant.
		\item[2.] When $\phi<\frac{\ell}{2\ell+1}$, it holds that for any fixed $t$,
$\lim_{n\rightarrow\infty}\PP(|\wh{U}_n| > t) = 1$.
	\end{enumerate}
\end{theorem}

In addition to the conditions imposed in Theorem \ref{theorem_score} and Lemmas S8 and S9, we further require two additional conditions involving the magnitude of $\theta^*$ in (\ref{eq_power_condition}). The first condition  $\norm{\bv^*}_1^2M_n^4n^{1-4\phi}/\delta = o(1)$ is imposed to ensure the local asymptotic normality (LAN) in terms of the parameter $\theta^*$.
The second condition in (\ref{eq_power_condition}) is similar to (\ref{eq_theorem_score_condition}), which controls
the magnitude of $\norm{\nabla R(0,\bgamma^*)}_\infty$ and $\norm{\nabla^2 R_\delta(\theta^*,\bgamma^*) - \nabla^2 R_\delta(0,\bgamma^*) }_{\max}$ under the alternative hypothesis.

This theorem implies that the proposed test converges in distribution to a normal distribution with mean $-\xi$, when the contiguous alternatives approach the null hypothesis at a rate $n^{-\frac{\ell}{2\ell+1}}$.  In addition, if the alternatives deviate  from the null hypothesis in the magnitude larger than $n^{-\frac{\ell}{2\ell+1}}$ (i.e., $\phi<\frac{\ell}{2\ell+1}$), the asymptotic power of our test is 1. In other words, our test can successfully detect the nonzero $\theta^*$ whose magnitude exceeds the order of $n^{-\ell/(2\ell + 1)}$. In contrast, for regular models, the local alternative that is detectable is of the standard parametric rate $n^{-1/2}$.


{
\begin{remark}
  By choosing the optimal bandwidth $\delta \asymp n^{-1/(2\ell + 1)}$, all the conditions in Theorem~\ref{theorem_power} can be simplified and summarized as
  \begin{eqnarray*}
  && M_n^2 \leq n^{\ell/(2\ell+1)}(\log d)^{-1/2}, \log d = o(n^{(2\ell-2)/(2\ell+1)}), \norm{\bv^*}_1 M_n^2 = o(n^{2\phi-(\ell+1)/(2\ell+1)}), \mbox{ and }\\
  && \norm{\bv^*}_1 M_n = o(n^{\ell/(6\ell+3)} \wedge n^{\phi-\ell/(2\ell+1)}\{\eta_1(n)\vee \eta_2(n)\}^{-1}),
  \end{eqnarray*}
  where $\eta_1(n)=s^{(4\ell+1)/(4\ell+2)} (\log d/n)^{\ell/(2\ell+1)}=o(1)$, $\eta_2(n)=s' (\log d/n)^{(\ell-1)/(2\ell+1)}=o(1)$ and $\norm{\bv^*}_1 \{(\log d)^{1/2}\vee n^{\ell/(2\ell+1)}M_n^2 \eta_1(n)\} \{n^{1/(2\ell+1)}\eta_1(n)\vee \eta_2(n)\} = o(1)$.
  Consider the extreme case $\ell\to \infty$, it can be verified that $\log d = n^{1/5}$, $s=s'=n^{1/10}$, $\norm{\bv^*}_1=n^{1/20}$, $M_n=n^{1/50}$ would satisfy all of these conditions when $\phi=1/2$.

In addition, these conditions could be further simplified if one is willing to assume $\norm{\bv^*}_1=O(1)$ and $M_n = O(1)$.
If that is the case, condition (\ref{eq_theorem_score_condition}) becomes
\bq \label{eq_simple_rate}
&s^{(4\ell+1)/(4\ell+2)}(s^{(4\ell+1)/(4\ell+2)}\vee s') n^{-(\ell -1)/(2\ell + 1)}(\log d)^{(4\ell - 1)/(4\ell + 2)} = o(1)
\eq
when taking $\delta \asymp n^{-1/(2\ell + 1)}$. If we consider the extreme case with $\ell\rightarrow\infty$,
it suffices to have $s(s\vee s')\log d= o(n^{1/2})$ in order to satisfy all of the conditions when $\phi=1/2$.
\end{remark}
}

\section{Data-Driven Bandwidth Selection}\label{adaptivity}

In the previous section, we establish the theoretical property of the bias corrected smoothed decorrelated score when the underlying conditional density $f(x|y,\bz)$ is $\ell$th order smooth. However, this smoothness parameter $\ell$ is typically unknown in practice, leading to the following two complications. First, in Assumption \ref{ass_kernel}, a kernel function $K$ of the same order is applied, which implicitly requires the knowledge on the smoothness parameter $\ell$. In practice, the choice of kernel functions is often determined by the user's preference rather than the theory. Since high order kernels may exacerbate the problem of variability, choosing low order kernels of 2 or 4 is often recommended (even if the density is more smooth); see \cite{hardle1992regression}. Second, the optimal bandwidth $\delta \asymp n^{-1/(2\ell + 1)}$ that balances the asymptotic bias and variance of the decorrelated score in Lemma \ref{lemma_ori_normality_new} also depends on the unknown $\ell$. It is well known from the nonparametric literature that the choice of bandwidth is an extremely important problem of both theoretical and practical values \citep{silverman1986density, bowman1984alternative, sheather1991reliable, hall1992smoothed,jones1996brief}.

In this section, we focus on how to choose the bandwidth $\delta$ in a data-driven manner. In view of the above discussion on the kernels, we assume that a low order kernel $K$ is chosen (for simplicity, we still denote its order by $\ell$) and meanwhile the underlying conditional density has a higher order smoothness parameter.

\begin{assumption}\label{ass_sb_additional}
	We assume that the kernel $K$ is of order $\ell$ and $f(x|y,\bz)$ is $(\ell + r)$th order smooth for some $\ell \geq 2$ and $r > 0$.
\end{assumption}
We define the optimal bandwidth $\delta^*$ as the one that minimizes the MSE of the smoothed decorrelated score:
\beq\label{eq_deltastar}
\delta^* = \argmin_{\delta}M(\delta), ~~\textrm{where}~~M(\delta)=\EE[(\bv^{*T}\nabla R_\delta^n(\bbeta^*))^2].\eeq
A direct bias-variance decomposition of $M(\delta)$ gives
\beq
M(\delta)= \frac1n\EE[(\bv^{*T}\nabla \bar{R}_\delta^1(\bbeta^*))^2] + \frac{n-1}{n}(\bv^{*T}\nabla R_\delta(\bbeta^*))^2:= \frac1n V(\delta) + \frac{n-1}{n}SB(\delta),
\eeq
where $\nabla \bar{R}_\delta^1(\bbeta^*)$ is defined in (\ref{eq_erisk_smooth}), $V(\delta)=\EE[(\bv^{*T}\nabla \bar{R}_\delta^1(\bbeta^*))^2]$ is used as a proxy for the variance, and $SB(\delta)=(\bv^{*T}\nabla R_\delta(\bbeta^*))^2$  denotes the squared error.
To estimate $\delta^*$, our main idea is to construct estimators  $\wh{V}(\delta)$ and $\wh{SB}(\delta)$ for $V(\delta)$ and $SB(\delta)$ and then estimate $\delta^*$ by
$$
\wh\delta = \argmin_{\delta}\wh M(\delta)~~\textrm{where}~~\wh{M}(\delta) = \frac{1}{n}\wh{V}(\delta) + \frac{n-1}{n}\wh{SB}(\delta).
$$

From the proof of Lemma \ref{lemma_ori_normality_new}, we can show that ${SB}(\delta)=(\delta^{\ell}\mu^*)^2(1+o(1))$ and ${V}(\delta)=\delta^{-1}\sigma^{*2}(1+o(1))$ as $\delta\rightarrow 0$, and thus $\sigma^{*2}/\delta$ and $(\delta^{\ell}\mu^*)^2$ are the asymptotic versions of $V(\delta)$ and $SB(\delta)$, respectively.
As a result, one may attempt to estimate the optimal bandwidth by minimizing the asymptotic MSE $\hat\sigma^{2}/\delta+(\delta^{\ell}\hat\mu)^2$ with the plug-in estimators $\hat\sigma$ and $\hat\mu$ developed in the previous section. However, the asymptotic MSE depends on the unknown smoothness $\ell$ and therefore is not appropriate for bandwidth selection in practice.

{
Instead, we propose to estimate $V(\delta)$ and $SB(\delta)$ using a different strategy.
Our estimates are still in the cross-fitting fashion, but when we estimate the bias $B(\delta)$ that dues to the approximation using the kernel function $K(\cdot)$ of order $\ell$ with bandwidth $\delta$, we will have to use a new pilot kernel function $J(\cdot)$ of order $r$ with bandwidth $b$.
Essentially when the target function has higher order smoothness than the kernel function applied for estimation, a different kernel smoothing procedure has to be used for estimating the bias.
This is motivated by the ``double-smoothing'' technique in nonparametric statistics  \citep{hardle1992regression,neumann1995automatic}, and is also related to the ``smoothed cross validation'' approach \citep{hall1992smoothed}.

To be specific, we estimate $V(\delta)$ with the following  moment estimator
\beq\label{eq_adaptive_sq_ori_1}
\wh{V}(\delta) =  \frac {1}{2} \big(
\wh{\bv}^{(1)T}\wh{\bGamma}^{(1)}(\delta)\wh{\bv}^{(1)} + \wh{\bv}^{(2)T}\wh{\bGamma}^{(2)}(\delta)\wh{\bv}^{(2)} \big),
\eeq
where $\wh{\bGamma}^{(1)}(\delta) = \frac{1}{|\cN_1|}\sum_{i\in \cN_1} \nabla \bar{R}_\delta^i(\wh{\bbeta}^{(2)})\nabla \bar{R}_\delta^i(\wh{\bbeta}^{(2)})^T$ and similarly for $\wh{\bGamma}^{(2)}$.
To estimate the squared bias $SB(\delta)$, note that the bias term $B(\delta)$ can be written as
\beq\label{eq_adaptive_sq_ori}
B(\delta) = \bv^{*T} \nabla R_\delta(\bbeta^*)
= \bv^{*T} \underbrace{(\nabla R_\delta(\bbeta^*) - \nabla R(\bbeta^*))}_{A(\bbeta^*,\delta)}
=\int_u K(u) \biggl\{
\bv^{*T}(\nabla R(u\delta,\bbeta^*) -\nabla R(\bbeta^*))
\biggr\} du,
\eeq
where we use $\nabla R(u\delta,\bbeta^*) = \sum_y w(y)\int_{\bz} \bz y
f(u\delta+ \bbeta^{*T}\bz|y,\bz) f(y,\bz)d\bz $ to denote the population gradient with a bias induced by $u\delta$, with a bit abuse of notation.
We propose to estimate $B(\delta)$ by
\beq\label{eq_bias_split}
\wh{B}(\delta) = \frac 12 \bigg(
\wh{\bv}^{(1)T} \frac{1}{|\cN_1|}\sum_{i \in \cN_1}A_i(\wh{\bbeta}^{(2)},\delta) + \wh{\bv}^{(2)T} \frac{1}{|\cN_2|}\sum_{i \in \cN_2}A_i(\wh{\bbeta}^{(1)},\delta)
\bigg),
\eeq
where
\beq
A_i(\wh \bbeta,\delta) =&\int_u K(u)w(y_i)\frac{\bz_iy_i}{b}
\Big[
J(\frac{x_i - \wh{\bbeta}^T\bz_i - u\delta}{b}) - J(\frac{x_i - \wh{\bbeta}^{T}\bz_i}{b})
\Big]du,
\eeq
and $J(\cdot)$ is the aforementioned new pilot kernel function of order $r$ with bandwidth $b$. Essentially, we substitute $\nabla R(u\delta,\bbeta^*)$ and $\nabla R(\bbeta^*)$ in (\ref{eq_adaptive_sq_ori}) with their corresponding second smoothers through kernel function $J(\cdot)$.
We now estimate the squared bias by $\wh{SB}(\delta)= \wh{B}(\delta)^2$.
}

To analyze theoretical properties of the estimates, let's define $\Delta = [q_1 n^{-1+\epsilon_1},q_2n^{-1+\epsilon_2}]$ as the range of bandwidth $\delta$ for some constants
$0 < q_1\leq q_2$ and $0<\epsilon_1 < \epsilon_2<1$. Since the optimal bandwidth $\delta^*$ is of order $n^{-1/(2\ell + 1)}$, we can guarantee $\delta^*\in \Delta$ for some suitable $\epsilon_1$ and $\epsilon_2$. Under some conditions, the uniform convergence rates of $\wh{V}(\delta) $ and $\wh{SB}(\delta)$ are given by
$$
|\wh{V}(\delta) - V(\delta)| \lesssim \frac{\psi_1(n,\delta)}{\delta},~~~~	|\wh{SB}(\delta) - SB(\delta)|\lesssim \delta^{2\ell} \psi_2(n,\delta),
$$
uniformly over all $\delta \in \Delta$, where
\begin{eqnarray*}
\psi_1(n,\delta) &=& \norm{\bv^*}_1^2\bigg(\eta_2(n) \vee \sqrt{\frac{\log(n\vee d)}{n\delta}} \vee M_n\eta_1(n)\bigg),\\
\psi_2(n,\delta) &=& \norm{\bv^*}_1^2\bigg(\sqrt{\frac{\log(n\vee d)}{nb^{2\ell+1}}} \vee (\delta\vee b)^r \vee M_n \eta_1(n)(1 \vee \frac{M_n\eta_1(n)}{\delta^{\ell}}) \vee \eta_2(n)\bigg).
\end{eqnarray*}
We refer to Lemmas S11 and S12 in Section S2 in the Supplement for the formal statement of the results and further interpretations of the rates.
With these two lemmas, we further establish the convergence rate of the data-driven bandwidth $\hat\delta$,  i.e., $\frac{\wh{\delta} - \delta^*}{\delta^*} \lesssim C_{n,\delta^*}$, where $C_{n,\delta^*}=\psi_1(n,\delta^*) \vee \psi_2(n,\delta^*)$, see Theorem S2 in the Supplement. That is, our data-driven bandwidth $\hat\delta$ is consistent.

For notational simplicity, let $\hat U_n(\delta)$ denote the bias corrected smoothed decorrelated score test statistic with a pre-specified bandwidth  parameter $\delta$. The main result in this section shows that $\hat{U}_{n}(\hat\delta)$ with the data-driven bandwidth $\hat\delta$ is still asymptotically normal under the null hypothesis.

\begin{theorem}\label{theorem_adapt_main_short}
	Under the conditions in Theorem S3 in the Supplement and $H_{0}: \theta^{*}=0$, it holds that
\[
\hat{U}_n({\hat{\delta}}) \stackrel{d}{\rightarrow} N(0,1).
\]
\end{theorem}

There are two major challenges in the analysis of $\hat{U}_n({\hat{\delta}})$. First, the estimator $\hat{\delta}$ and the decorrelated score statistic are generally dependent with each other, which prevents the direct use of many concentration inequalities such as Bernstein's. To decouple the dependence, similar to Section \ref{sec_details}, we carefully design a cross-fitting approach by splitting the data into three folds. Due to the space constraint, we leave the detailed algorithm to Section S2 in the Supplement.
Second, different from Theorem \ref{theorem_score} which presents the asymptotic normality of $\hat{U}_n({{\delta}})$ with a fixed $\delta$,
the uncertainty of $\hat{\delta}$ needs to be incorporated in the proof of Theorem \ref{theorem_adapt_main_short}. In particular, we use concentration inequalities to take care of the higher order error terms in the Taylor expansion of $\hat{U}_n({\hat{\delta}})$ with respect to both $\delta$ and $\bbeta$. This leads to much more involved analysis than that in Theorem \ref{theorem_score}.

\section{Hypothesis Testing for the Linear Combination}\label{sec:imcid}

Our results presented thus far are for the hypothesis testing (\ref{eq:test theta}), where $\theta$ is simply a single element in $\bbeta$. In applications, researchers may be more interested in the hypothesis testing (\ref{eq:test MCID}), a linear combination of the parameter $\bbeta$. For example, as we mentioned at the beginning of Section~\ref{method}, in the study of inferring iMCID, it is of interest to test $\bc_0^T\bbeta^*=0$ where $\bc_0$ represents the realized value of a new patient's clinical profile and we assume $c_{01}\neq 0$ without loss of generality.
The methods and results for (\ref{eq:test theta}) are essential.
Below, we show that, in a parallel manner, all of them can be developed for (\ref{eq:test MCID}) which is more applicable in scientific applications.
We also point out, technically, our methods and results can be further extended to a more general null hypothesis $H_{0M}: \Mb\bbeta^*=\zero$ where $\Mb\in\RR^{m\times d}$ and $m$ is a fixed integer.

To test the hypothesis (\ref{eq:test MCID}), consider the one to one reparametrization $(\theta,\bgamma)\rightarrow (\xi,\bgamma)$, where $\xi=\bc_0^T\bbeta$. Under this new set of parameters, the null hypothesis can be written as $H_{0L}: \xi^*=0$, and the smoothed surrogate loss reduces to $R^n_{\delta}(\frac{\xi-\bc_{02}^T\bgamma}{c_{01}},\bgamma)$, where we write $\bc_0 = (c_{01},\bc_{02}^T)^T$ with $\bc_{02}\in \RR^{d-1}$.
Define $\bC = \begin{bmatrix}
       \frac{1}{c_{01}} & 0          \\[0.3em]
       \frac{-\bc_{02}}{c_{01}} & \bI_{d-1}
     \end{bmatrix} \in \RR^{d\times d}$.
From the chain rule, we can show that
\beq
&\nabla_{(\xi,\bgamma)} R^n_{\delta}(\frac{\xi-\bc_{02}^T\bgamma}{c_{01}},\bgamma) = \bC \nabla R_\delta^n(\theta,\bgamma),~~\nabla^2_{(\xi,\bgamma),(\xi,\bgamma)} R^n_{\delta}(\frac{\xi-\bc_{02}^T\bgamma}{c_{01}},\bgamma)
= \bC\nabla^2 R_\delta^n(\theta,\bgamma)\bC^T,
\eeq
and similarly for $R_\delta(\frac{\xi-\bc_{02}^T\bgamma}{c_{01}},\bgamma)$ and $R(\frac{\xi-\bc_{02}^T\bgamma}{c_{01}},\bgamma)$.
Therefore, following the same idea as in Section \ref{method}, we define the smoothed decorrelated score as
\beq
S^L_{\delta}(\xi,\bgamma,\bomega^*_L)=\nabla_\xi R^n_{\delta}(\frac{\xi-\bc_{02}^T\bgamma}{c_{01}},\bgamma)-\bomega^{*T}_L\nabla_{\bgamma} R^n_{\delta}(\frac{\xi-\bc_{02}^T\bgamma}{c_{01}},\bgamma),
\eeq
where
$\bomega^*_L = \Big [
\nabla^2_{\bgamma,\bgamma} R(\frac{\xi-\bc_{02}^T\bgamma}{c_{01}},\bgamma)
\Big ]^{-1}\nabla^2_{\bgamma,\xi} R(\frac{\xi-\bc_{02}^T\bgamma}{c_{01}},\bgamma).$

We write $\bv^*_L = (1,\bomega^{*T}_L)^T$ and denote $\mu^*_{L}, \sigma^*_{L}$ as  the (scaled) asymptotic bias and standard deviation of the score function $S^L_{\delta}(\xi,\bgamma,\bomega^*_L)$, i.e.,
\beq
\mu^*_{L} = \bv^{*T}_L\bC\bb^*,~~~~~~ \sigma^*_L =  \sqrt{\bv^{*T}_L\bC\bSigma^*\bC^T\bv^*_L},
\eeq
where $\bb^*$ and $\bSigma^*$ are defined in (\ref{eq:definemustar}) and (\ref{eq:defineSigmastar}), respectively. From above we can see that the estimation methods for $\bomega^*, \mu^*$ and $\sigma^*$ proposed in Section \ref{method} can be easily extended to obtain corresponding estimators for
$\bomega^*_L, \mu^*_{L}$ and $\sigma^*_L$.
Given these estimators, we define the test statistics for $H_{0L}$ as
$$
\wh{U}^L_n = (n\delta)^{1/2}\frac{S^L_{\delta}(0,\hat\bgamma,\hat\bomega_L) - \delta^{\ell}\wh{\mu}_L}{\wh{\sigma}_L}.
$$
Accordingly, all the parallel results presented in Section~\ref{theory} and Section~\ref{adaptivity} can be developed. In the interest of space, we only present the following result that
characterizes the asymptotic distribution of $\wh U_n^L$ under the null hypothesis $H_{0L}$. All other parallel results are omitted.

\begin{theorem}\label{theorem_linear_hypothesis}
If Assumptions \ref{ass_smooth} - \ref{ass_estimators} hold with $\mu^*,\sigma^*,\bv^*,\wh\bv$ substituted by $\mu^*_L,\sigma^*_L,\bv^*_L,\wh\bv_L$, and in addition $\wh\mu_{L}$ and $\wh\sigma_L$ are consistent estimators of $\mu^*_{L}$ and $\sigma^*_L$, respectively, then under the same conditions as in Theorem \ref{theorem_score} and the null hypothesis $H_{0L}:\xi^* = 0$, it holds that
	$\wh U_n^L \stackrel{d}{\rightarrow}  N(0,1)$.
\end{theorem}

\section{Simulation Studies}\label{simulation}
In this section, we evaluate the empirical performance of the proposed methods. Although many models can be formulated as special cases of our problem (\ref{eq_risk_ori}), here we mainly consider the following binary response model
\beq\label{eq_binary_reponse_model}
Y = \sign{X - \bbeta^{*T}\bZ + \epsilon},
\eeq
where $\epsilon$ possibly depends on $X$ and $\bZ$ but the median of $\epsilon$ given $X$ and $\bZ$ is 0. 

\subsection{Experiments with pre-specified bandwidth}\label{simulation_fixed}
In the first set of experiments, we evaluate the performance of the proposed test statistic with pre-specified bandwidth.
We use Gaussian kernel $K$ of order 2 with bandwidth pre-specified at $\delta = n^{-1/5}$. The choices of other tuning parameters are detailed in Section S5 in the Supplement.
Throughout this subsection, we consider sample size $n = 800$, dimension $d = 100,500,1000$ and generate $\beta^*_2,\dotso,\beta^*_{s}$ by sampling from a uniform distribution within $[1,2]$ for $s = 3,10$. The first coordinate $\beta_1^*$ would vary depending on the purpose of the experiment, and the rest coordinates of $\bbeta^*$ are all set to 0. After that, the coefficient vector is then normalized such that $\norm{\bbeta^*}_2 = 1$. We generate $X \sim N(0,1)$ and $\bZ \sim N(0,\bSigma_{\rho})$, where $(\bSigma_{\rho})_{jk}=\rho^{|j-k|}$ with $\rho = 0.2, 0.5, 0.7$. 
For all cases, the simulations are repeated 250 times.

In the first scenario, we let $\epsilon \sim N(0, 0.2^2(1 + 2(X - \bbeta^{*T}\bZ)^2))$, which is referred to as Heteroskedastic Gaussian scenario later on.
We compare the proposed smoothed decorrelated score test (SDS) with the decorrelated score test method (DS) \citep{ning2017general} and Honest confidence region method (Honest) \citep{belloni2013honest} from the ``hdm'' package.
We fix the significance level at $0.05$ and firstly evaluate the performance of the tests under the null hypothesis $H_0: \beta^*_1 = 0$. In this case, we set $\beta^*_1 = 0$.
Note that the R code for the DS and Honest approaches is tailored for the high-dimensional logistic regression, which differs from the above data generating process.

Table \ref{table_simu_manski} reports the empirical Type I error rate under the first scenario. The error rate from the SDS method is generally close to the nominal significance level 0.05, which empirically verifies the theoretical results in Theorem \ref{theorem_score}. For both Honest and DS methods, the empirical Type I error rate seems to be consistently higher or lower than the nominal level. This is expected as these two methods only work for the logistic regression.
By taking a closer look at the Normal Q-Q plot of the test statistics, we observe that the distribution of the test statistics from the Honest and DS methods deviate substantially from Gaussian, as opposed to those yield by the proposed SDS method. Please see Section S6.2 in the Supplement for more details.

\begin{table}[h]
	\caption{The empirical Type I error rate of the tests under the Heteroskedastic Gaussian scenario from SDS, DS and Honest methods.}\label{table_simu_manski}
	\setlength\extrarowheight{1pt}
	\begin{center}
		\begin{tabular}{ c c| c c c | c c c }
			\hline
			&& &$s = 3$&&&$s=10$&\\
			d&method&$\rho = 0.2$&$\rho = 0.5$&$\rho = 0.7$&$\rho = 0.2$&$\rho = 0.5$&$\rho = 0.7$\\
			\hline
			100&SDS&5.6\%&5.0\%&6.4\%&4.8\%&4.8\%&5.2\%\\
			&DS&1.2\%&2.0\%&2.0\%&2.0\%&1.8\%&1.8\%\\
			&Honest&5.2\%&5.6\%&7.6\%&5.4\%&5.2\%&6.8\%\\
			500&SDS&4.8\%&4.4\%&5.6\%&5.6\%&5.0\%&4.8\%\\
			&DS&0.2\%&0.4\%&0.4\%&0.2\%&0.0\%&0.4\%\\
			&Honest&7.0\%&10.8\%&7.6\%&8.2\%&6.8\%&7.2\%\\
			1000&SDS&4.4\%&6.0\%&5.6\%&5.0\%&5.4\%&5.0\%\\
			&DS&0.0\%&0.4\%&0.2\%&0.0\%&0.0\%&0.4\%\\
			&Honest&10.0\%&10.4\%&12.6\%&12.4\%&6.4\%&15.2\%\\
			\hline
		\end{tabular}
	\end{center}	
\end{table}

\begin{table}[h]
	\caption{The empirical Type I error rate of the tests under the Heteroskedastic Uniform scenario from SDS, DS and Honest methods.}\label{table_simu_uniform}
	\setlength\extrarowheight{1pt}
	\begin{center}
		\begin{tabular}{ c c| c c c | c c c }
			\hline
			&& &$s = 3$&&&$s=10$&\\
			d&method&$\rho = 0.2$&$\rho = 0.5$&$\rho = 0.7$&$\rho = 0.2$&$\rho = 0.5$&$\rho = 0.7$\\
			\hline
			100&SDS&6.0\%&5.6\%&6.8\%&7.2\%&6.8\%&7.2\%\\
			&DS&2.0\%&0.4\%&0.8\%&1.6\%&0.8\%&1.2\%\\
	
			&Honest&9.6\%&17.6\%&19.2\%&9.6\%&18.4\%&21.6\%\\
			500&SDS&6.4\%&6.8\%&6.4\%&6.0\%&7.6\%&7.2\%\\
			&DS&1.6\%&0.8\%&0.0\%&0.8\%&1.2\%&0.4\%\\
			&Honest&11.2\%&15.6\%&20.0\%&13.6\%&17.2\%&22.8\%\\
			1000&SDS&8.4\%&7.6\%&8.8\%&9.2\%&7.6\%&8.0\%\\
			&DS&0.0\%&0.4\%&0.4\%&1.2\%&0.8\%&0.0\%\\
        	&Honest&13.6\%&15.2\%&22.4\%&16.0\%&19.2\%&20.8\%\\
			\hline
		\end{tabular}
	\end{center}	
\end{table}

In the second scenario, we let $\epsilon \sim 0.2\cdot \text{Unif}(-G(X,\bZ),G(X,\bZ))$, where $G(x,\bz) = \sqrt{1 + 2(x - \bbeta^{*T}\bz)^2}$. In other words, the error $\epsilon$ follows a uniform distribution such that its range depends on the covariates $X,\bZ$ (we will call it Heteroskedastic Uniform scenario). Similar to the Heteroskedastic Gaussian case, we compare SDS method with DS and Honest methods and study the empirical Type I error rate. 
From Table \ref{table_simu_uniform} we can see that the proposed method yields Type I error close to the nominal level as opposed to the other two. The Normal QQ-plots in Section S6.2 in the Supplement further confirm the asymptotic normality of our SDS test statistics. The above results  suggest that in practice, if the underlying data generating process is the binary response model, our proposed approach provides valid inferential results while the existing approaches fail.

Next, we investigate the empirical power of the SDS method. We use the same data generating processes as in the above two scenarios, but instead of setting $\beta^*_1 = 0$, we vary $\beta^*_1$ in the grid $\{0.02,0.05,0.075, 0.10,0.15,0.20,0.25,0.30\}$ for the Heteroskedastic Gaussian case, and $\{0.025,0.05,0.075, 0.10,0.125,0.15, 0.175\}$ for the Heteroskedastic Uniform case. Similarly, we consider $s = 3, 10$, $d = 100, 500, 1000$ and $\rho = 0.2,0.5,0.7$. Figure \ref{fig_Power_1} shows the empirical rejection rate of the SDS method when $s = 10$ (see Section S6.3 for the results when $s = 3$). Note that we do not compare with the DS and Honest methods for the empirical power, because these two tests do not maintain the desired Type I error in our scenarios. We can see that for all considered cases, the empirical power converges to 1 as the magnitude of the signal $\beta^*_1$ becomes larger, which agrees with Theorem \ref{theorem_power}. In addition, we find that the dimension $d$ has minor effects on the empirical power,  which is reasonable as Theorem \ref{theorem_power} only depends on $\log d$ via the condition (\ref{eq_power_condition}). Finally, we note that the power of the test deteriorates as the correlation of the design increases.

\begin{figure}[h]
	\centering
	{\subfigure[]{\includegraphics[width=80mm,height = 80mm]{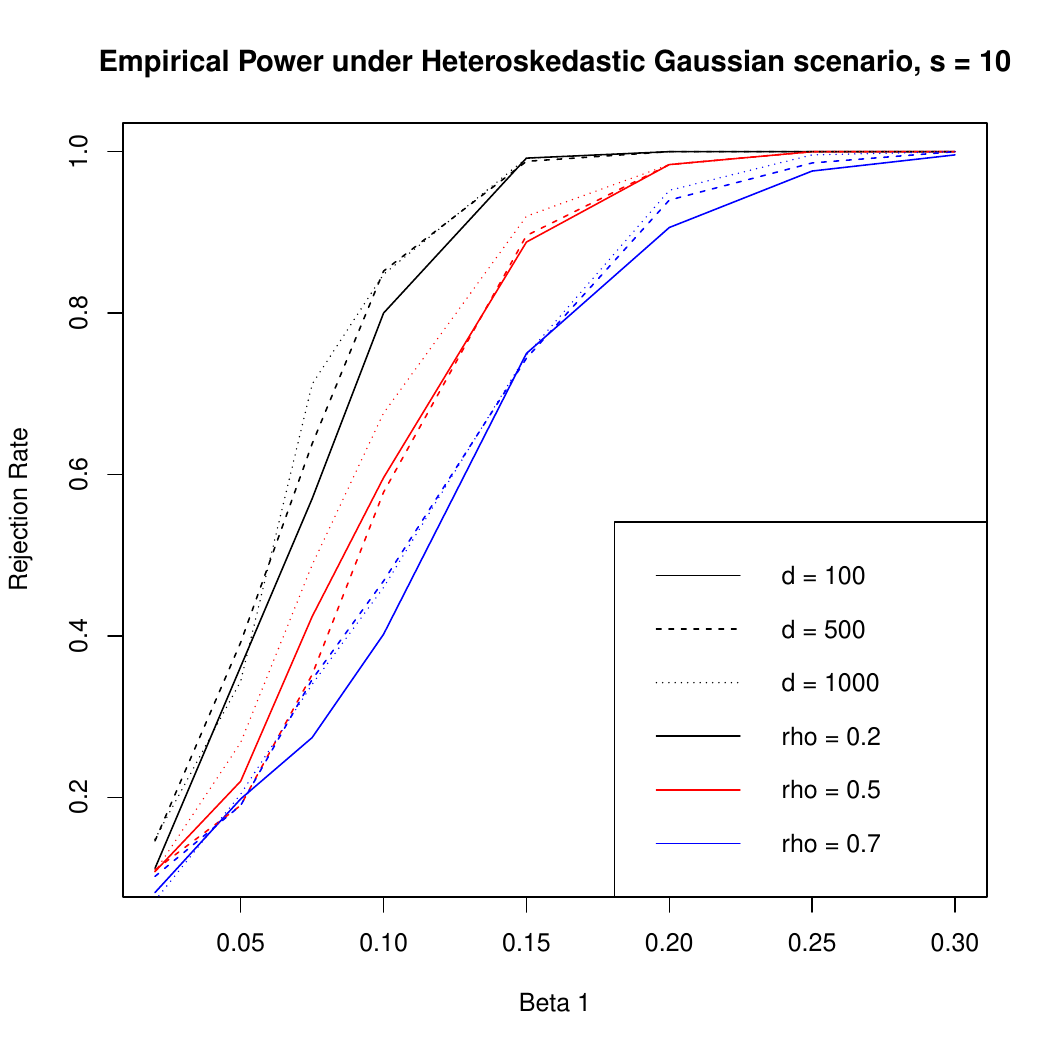}}
	}
	{\subfigure[]{\includegraphics[width=80mm,height = 80mm]{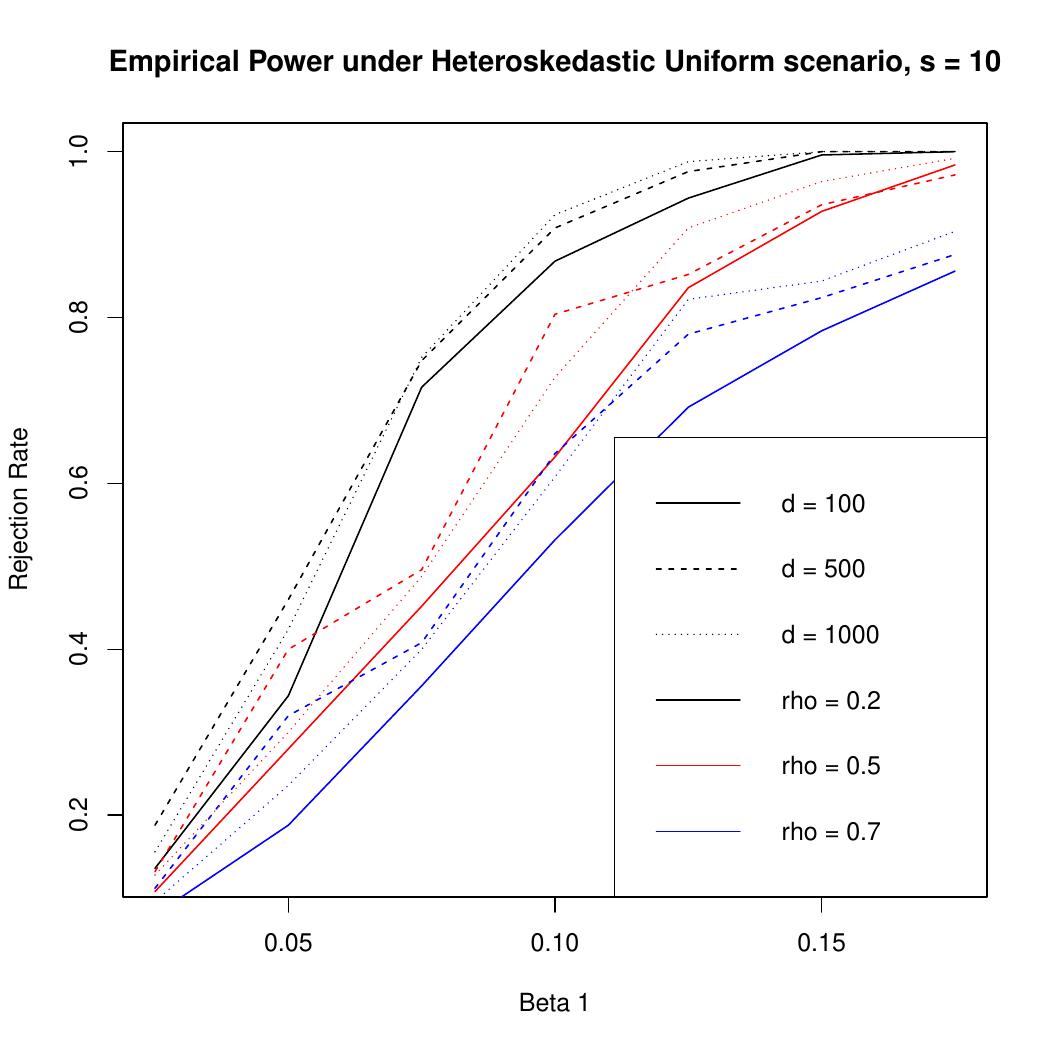}}
	}
	\caption{Empirical rejection rate of the proposed test under both scenarios with $s = 10$, $d = 100, 500, 1000$ and $\rho = 0.2,0.5,0.7$.}   \label{fig_Power_1}
\end{figure}

\subsection{Experiments with data-driven bandwidth}\label{sec_sim_band}
In the next set of experiments, we study the empirical performance of the data-driven bandwidth selection approach.
We firstly study the type I error and power of our SDS method for testing $H_0: \beta^*_1 = 0$ versus $H_1: \beta^*_1 \neq 0$ with data-driven bandwidth.
We consider the same data generating processes as in  Section \ref{simulation_fixed} with $n = 800, d = 100, s  = 3,10$ and $\rho = 0.2,0.5,0.7$. We seek for the minimizer of the estimated MSE over $\delta \in [0.1,1.2]$ and each experiment is repeated 250 times.  After $\wh\delta$ is obtained, we plug-in it into the test statistic and estimate the bias and variance as discussed in Section S5 in the Supplement. With the same implementations, we also evaluate the empirical power of the test by varying $\beta^*_1$ in the same grid as in Section \ref{simulation_fixed}.

\begin{table}[h]
	\caption{The empirical Type I error rate of the tests under the Heteroskedastic Gaussian and Uniform scenarios with data-driven bandwidth $\wh\delta$.}\label{table_adaptive_type1}
	\setlength\extrarowheight{1pt}
	\begin{center}
	\begin{tabular}{ c | c c c | c c c }
			\hline
			& &$s = 3$&&&$s=10$&\\
			Data generating process&$\rho = 0.2$&$\rho = 0.5$&$\rho = 0.7$&$\rho = 0.2$&$\rho = 0.5$&$\rho = 0.7$\\
			\hline
			Heteroskedastic Gaussian&6.8\%&7.2\%&5.6\%&8.4\%&7.2\%&6.4\%\\
			Heteroskedastic Uniform&8.8\%&8.0\%&7.6\%&8.4\%&6.8\%&5.2\%\\
			\hline
		\end{tabular}
	\end{center}	
\end{table}

\begin{figure}[h]
	\centering
	{\subfigure[]{\includegraphics[width=80mm,height = 80mm]{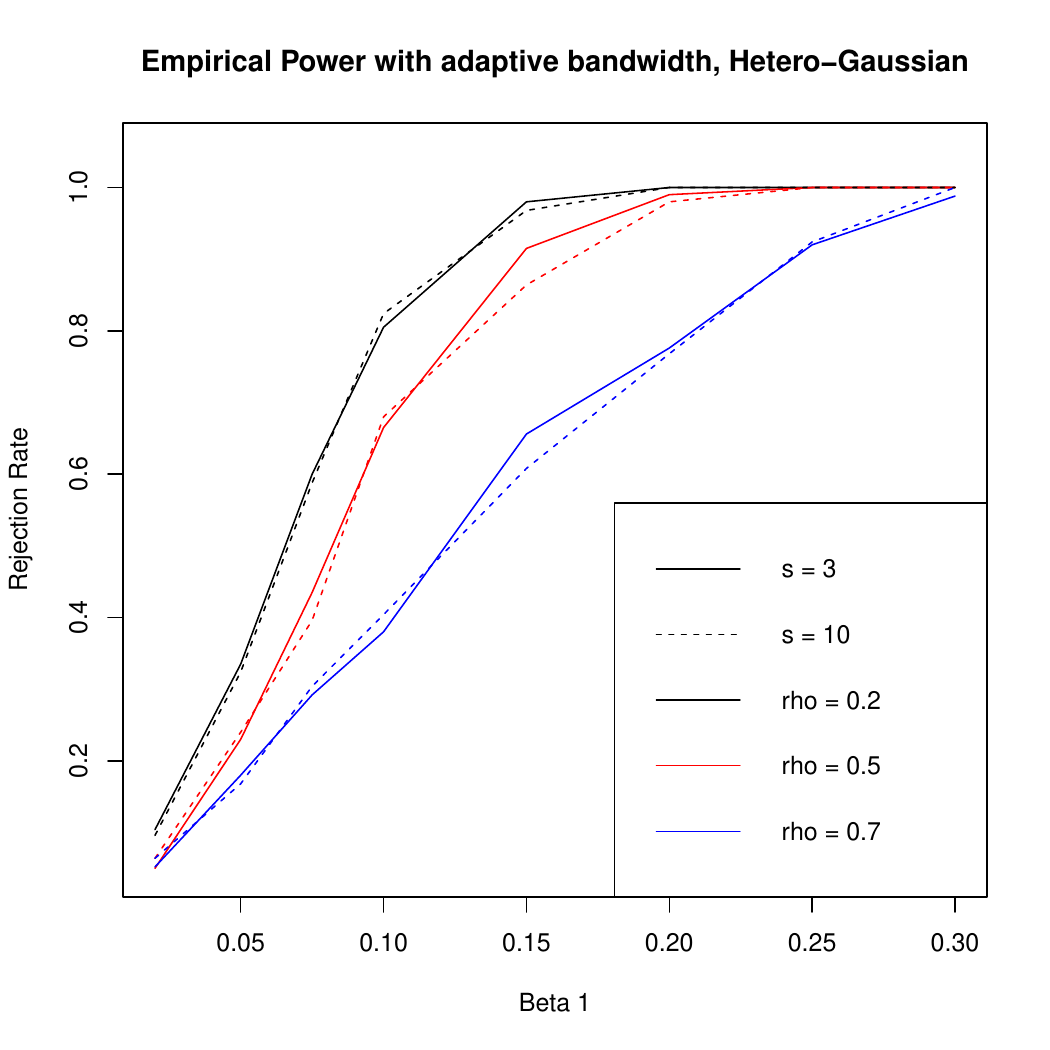}}
	}
	{\subfigure[]{\includegraphics[width=80mm,height = 80mm]{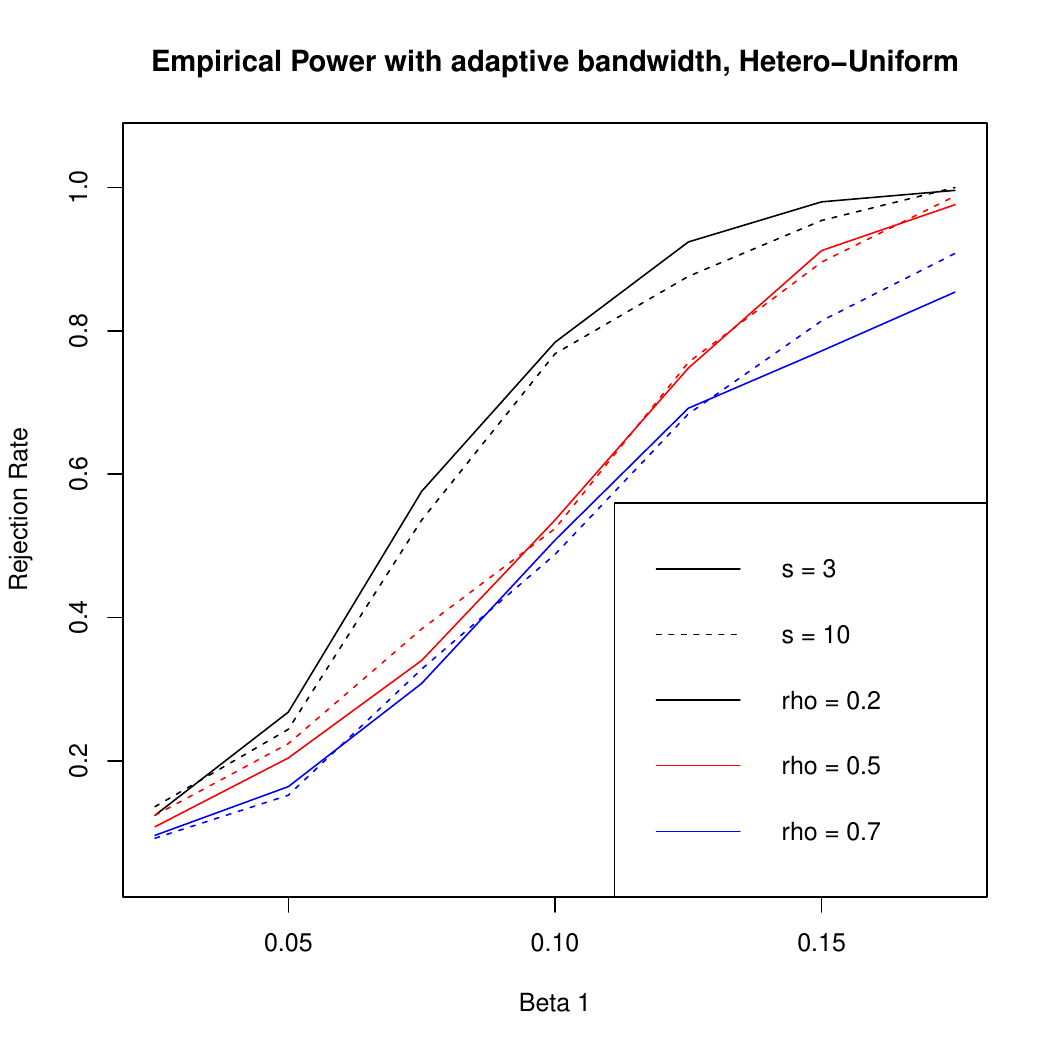}}
	}
	\caption{Empirical power of the tests under the Heteroskedastic Gaussian and Uniform scenarios with data-driven bandwidth $\wh\delta$.}   \label{fig_adaptive}
\end{figure}

Table \ref{table_adaptive_type1} shows the empirical Type I error rate over 250 repetitions  when $\beta^*_1 = 0$, and Figure \ref{fig_adaptive} shows the empirical power for different $\beta^*_1 \neq 0$ in these scenarios. Similar to the case when the bandwidth $\delta$ is pre-specified, the empirical Type I errors are generally close to the nominal level 0.05, and the empirical power converges to 1 as $\beta^*_1$ becomes larger. We refer to Section S6 in the Supplement for further numerical results.

From all of these numerical results, we recommend using the data-driven bandwidth selection approach in practice.


\section{Analysis of ChAMP Trial}\label{sec_data}

In this section we analyze the ChAMP (Chondral Lesions And Meniscus Procedures) trial \citep{bisson2017patient}, which contains clinical information about $n = 138$ patients undergoing arthroscopic partial meniscectomy (APM), a knee surgery for meniscal tears. The response variable is $Y = 1$ if the patient is healthy/satisfactory and $-1$ otherwise, obtained from the SF-36 survey. The continuous measurement $X$ encodes the WOMAC pain score change from the baseline to one-year after the surgery.
The dataset also contains $d = 160$ additional variables from the patient's clinical profile, denoted by $\bZ$. The scientific question is to determine the iMCID, defined as a linear combination of the covariates $\bbeta^T\bZ$, such that the treatment of debriding chondral lesions can be claimed as clinically significant by comparing the WOMAC pain score change with this individualized threshold. As we can see, this application naturally fits into our formulation (\ref{eq_risk_ori}) with weight function $w(y) = 1/\PP(Y = y)$. The goal of the analysis is to address this question by providing valid inferential results for each component of $\bbeta$.



We apply the proposed SDS test for $H_{0j}: \beta_j=0$ versus $H_{1j}: \beta_j\neq 0$, where $1\leq j\leq d$. We use the same tuning parameter setting for estimating $\bbeta^*,\omega^*$ and the asymptotic bias and variance of the score function following Section S5. For comparison, we also apply the DS and Honest methods discussed in Section \ref{simulation_fixed}.

\begin{table}[h]
	\caption{The three significant variables (with p-value $<0.05/d$) identified by the proposed SDS method, and their corresponding p-values obtained from the DS and Honest methods.}\label{table_champ}
	\begin{center}
		\begin{tabular}{ c |c c c }
			\hline
			p-value& KQOL\_6wk& flex\_inj\_pre &KSymp\_3mo\\
			\hline
			SDS& 3.583e-07& 5.575e-05&  4.482e-05\\
			DS&0.0168 & 0.0340 & 0.0213 \\
			Honest&0.0591&0.0241 &   0.4383\\
			\hline
		\end{tabular}
	\end{center}	
\end{table}

Table \ref{table_champ} lists the three significant variables (i.e., those with p-value $<0.05/d =$3.125e-04)  from the proposed SDS approach.
Interestingly, all of them are clinically relevant and can provide meaningful implications for iMCID. The significance of the variable KQOL\_6wk, which represents the KOOS score for quality of life at 6-week, definitely indicates how the patients recover at a relatively early stage after the surgery.
The variable flex\_inj\_pre means the degree of flexion right before the surgery. Its significance recommends that the baseline disease severity would affect the magnitude of iMCID---this similar phenomenon was also discovered in the clinical literature for other types of diseases, such as the shoulder pain reduction study \citep{heald1997shoulder}.
The third variable KSymp\_3mo is the KOOS score for other symptoms at 3-month. In some previous analysis of ChAMP trial where only estimate is available but without inference results, this variable has the second largest coefficient \citep{feng2022nonregular}.

The results from the DS and the Honest methods are different.
Firstly, the DS method only yields 1 significant variable and the Honest method yields 13. From Table \ref{table_champ}, the three significant variables identified by the proposed SDS method cannot be identified by either DS or Honest.
In general, compared to the proposed SDS method, DS identifies fewer significant variables while Honest identifies more. This phenomenon is also evident from the simulation results in Section \ref{simulation_fixed}.
On the other hand, the results from the three methods do not completely contradict with each other. For instance, the two significant variables, KQOL\_6wk and KSymp\_3mo, identified from the proposed SDS method, has the fourth and fifth smallest p-values in the DS method.
Please refer to Section S6.7 in the Supplement for more results of this analysis.

In general, recall that DS and Honest methods are devised for the logistic regression, while our
proposed SDS method can produce valid inference results under the binary response model (\ref{eq_binary_reponse_model}), which is more flexible since the distribution of $\epsilon$ is left unspecified.
Therefore, we expect that the significant variables identified by SDS are potentially more reliable and clinically more relevant. Our results presented in this session echo this rationale.

\section*{Supplemental Material}

The Supplement include the technical proofs, some more detailed theoretical results and discussions, and additional numerical results.

\section*{Acknowledgment}

Ning is supported in part by National Science Foundation (NSF) CAREER award DMS-1941945 and NSF award DMS-1854637.
Zhao is supported in part by NSF award DMS-2122074 and a start-up grant from the University of Wisconsin-Madison.
The authors would like to thank the Editor, an Associate Editor, and two reviewers for their insightful comments which have helped improve the manuscript substantially.

\bibliographystyle{asa}
\bibliography{ref}

 \renewcommand\theequation{S\arabic{equation}}
 \renewcommand\thelemma{S\arabic{lemma}}
\renewcommand{\thesection}{S\arabic{section}}
 \renewcommand\theassumption{S\arabic{assumption}}
 \renewcommand\thecorollary{S\arabic{corollary}}
 \renewcommand\thetheorem{S\arabic{theorem}}
 \renewcommand\theproposition{S\arabic{proposition}}
  \renewcommand\theremark{S\arabic{remark}}

\newpage
\begin{center}
    \bf{\Large Supplemental Material for ``Test of Significance for High-dimensional Thresholds with Application to Individualized Minimal Clinically Important Difference''
    }
\end{center}







\section{Proofs of Main Results}
\subsection{Proof of Lemma \ref{lemma_ori_normality_new}}\label{pf_lemma_ori_normality_new}
	Recall that
	\beq
	\nabla \bar{R}_{\delta}^i(\bbeta^{*}) = w(y_i)\frac{y_i\bz_i}{\delta}K\bigg(\frac{y_i(x_i - \bbeta^{*T}\bz_i)}{\delta}\bigg).
	\eeq
	Let $T_i =\frac{ \sqrt{\delta}\bv^{*T}(\nabla \bar{R}_\delta^i(\bbeta^*) - \nabla R_\delta(\bbeta^*))}{\sqrt{\bv^{*T}\bSigma^*\bv^*}}$, we know by definition $\EE T_i = 0$.
	Consider
	\beq
	Var[\bv^{*T}\nabla \bar{R}_\delta^i(\bbeta^*)] =&
	\EE [(\bv^{*T}\nabla \bar{R}_\delta^i(\bbeta^*))^2] - (\bv^{*T}\nabla R_\delta(\bbeta^*))^2.
	\eeq
	Here 	
	\beq
	\EE [(\bv^{*T}\nabla \bar{R}_\delta^i(\bbeta^*))^2] =& \sum_{y \in \{-1,1\}}w(y)^2\int\frac{(\bv^{*T}\bz)^2}{\delta^2}\int K^2(\frac{x - \bbeta^{*T}\bz}{\delta})f(x|y,\bz) dxf(y,\bz)d\bz\\
	=&\sum_{y \in \{-1,1\}}w(y)^2\int\frac{(\bv^{*T}\bz)^2}{\delta}\int K^2(u)f(u\delta + \bbeta^{*T}\bz|y,\bz) duf(y,\bz)d\bz\\
	=&\frac{1}{\delta}\sum_{y \in \{-1,1\}}w(y)^2\int(\bv^{*T}\bz)^2\int K^2(u)(f(\bbeta^{*T}\bz|y,\bz) + u\delta f'(\bbeta^{*T}\bz|y,\bz)
+ o(\delta))\\
& duf(y,\bz)d\bz\\
	=&\frac{1}{\delta}(\bv^{*T}\bSigma^*\bv^*(1+ o(1))),
	\eeq
	where the second equality is due to a change of variable, and the third equality is due to Assumption \ref{ass_smooth}. Meanwhile, we know
	\beq
	\bv^{*T}\nabla R_\delta(\bbeta^*) =& \sum_{y \in \{-1,1\}}w(y)y\int\frac{(\bv^{*T}\bz)}{\delta}\int K(\frac{x - \bbeta^{*T}\bz}{\delta})f(x|y,\bz) dxf(y,\bz)d\bz\\
	=&\sum_{y \in \{-1,1\}}w(y)y\int(\bv^{*T}\bz)\int K(u)f(u\delta + \bbeta^{*T}\bz|y,\bz) duf(y,\bz)d\bz\\
	=&\sum_{y \in \{-1,1\}}w(y)y\int(\bv^{*T}\bz)\int K(u)	\frac{(u\delta)^{\ell}}{\ell!}\bigg(f^{(\ell)}(\bbeta^{*T}\bz|y,\bz) + \cO((u\delta)^{\zeta})\bigg)
	duf(y,\bz)d\bz\\
	=&\delta^{\ell}\bv^{*T}\bb^*(1 + o(1)).
	\eeq
	This together with Assumption \ref{ass_projection_norm} implies that $Var[\bv^{*T}\nabla R_\delta^i(\bbeta^*)] = \frac{1}{\delta}\bv^{*T}\bSigma^*\bv^*(1+ o(1))$ and therefore $Var(T_i) = 1 + o(1).$ Now we verify the Lyapunov condition
	\beq\label{eq_asymp_0}
	\frac{1}{n^{3/2}}\sum_i^n \EE|T_i|^3 =& \frac{1}{n^{3/2}}\sum_i^n\EE\bigg|
	\frac{ \sqrt{\delta}\bv^{*T}(\nabla \bar{R}_\delta^i(\bbeta^*) - \nabla R_\delta(\bbeta^*))}{\sqrt{\bv^{*T}\bSigma^*\bv^*}}\bigg|^3.
	\eeq
	By Assumption \ref{ass_projection_norm},
	\beq\label{eq_asymp_1}
	\EE\bigg|
	\frac{ \sqrt{\delta}\bv^{*T}(\nabla \bar{R}_\delta^i(\bbeta^*) - \nabla R_\delta(\bbeta^*))}{\sqrt{\bv^{*T}\bSigma^*\bv^*}}
	\bigg|^3
	\lesssim& \delta^{3/2}\EE\bigg|
	\bv^{*T}(\nabla \bar{R}_\delta^i(\bbeta^*) - \nabla R_\delta(\bbeta^*))
	\bigg |^3\\
	\lesssim &\delta^{3/2}(\EE |\bv^{*T}\nabla \bar{R}_\delta^i(\bbeta^*)|^3 + |\bv^{*T}\nabla R_\delta(\bbeta^*)|^3),
	\eeq
	where we can show that
	\beq
	\EE |\bv^{*T}\nabla \bar{R}_\delta^i(\bbeta^*)|^3 \leq& \norm{\bv^*}_1^3\sum_{y \in \{-1,1\}}\int\delta\max_{j}\bigg|\frac{y\bz_j}{\delta}K(u)
	\bigg|^3f(u\delta + \bbeta^{*T}\bz,y,\bz)dud\bz\\
	\lesssim&\norm{\bv^*}_1^3M_n^3/\delta^2,
	\eeq
	and from Lemma \ref{lemma_bias} it's easy to see that the first term on the RHS of (\ref{eq_asymp_1}) is dominant.  This implies that
	\beq
	\frac{1}{n^{3/2}}\sum_i^n \EE|T_i|^3 = \cO((\norm{\bv^*}_1M_n)^3/(n\delta)^{1/2}).
	\eeq
	Therefore, under the condition of this Lemma the Lyapunov condition holds, which completes the proof by applying Lindeberg Feller Central Limit Theorem.

\subsection{Proof of Theorem \ref{theorem_score}}\label{pf_theorem_score}
	It suffices to show that $(n\delta)^{1/2}|\wh{S}_\delta(\wh{\bbeta}_0) - S_\delta(\bbeta^*)| = o_{\PP}(1)$ where $\wh{\bbeta}_0 = (0,\wh{\gamma})$. Here we only show that $ (n\delta)^{1/2}|\wh{S}_\delta^{(1)}(\wh{\bbeta}_0^{(2)}) - S_\delta(\bbeta^*)| = o_{\PP}(1)$  and the desired result shall follow naturally from Lemma \ref{lemma_ori_normality_new}.
	
	By definition, we have
	\beq
	&(n\delta)^{1/2}|\wh{S}_\delta^{(1)}(\wh{\bbeta}_0^{(2)})- S_\delta(\bbeta^*)|\\
	=& (n\delta)^{1/2}|\wh{\bv}^{(1)T}\nabla R_\delta^{n_{(1)}}(\wh{\bbeta}_0^{(2)}) - \bv^{*T}\nabla R_\delta^{n_{(1)}}(\bbeta^*)|\\
	\leq&(n\delta)^{1/2}|\bv^{*T}(\nabla R_\delta^{n_{(1)}}(\wh{\bbeta}_0^{(2)}) - \nabla R_\delta^{n_{(1)}}(\bbeta^*))| + (n\delta)^{1/2}|(\wh{\bv}^{(1)} - \bv^{*})\nabla R_\delta^{n_{(1)}}(\wh{\bbeta}_0^{(2)})|\\
	:=&I_1 + I_2.
	\eeq
	The fact that $\bv^{*T}\nabla^2_{\cdot\bgamma}R(\bbeta^*) = 0$ and Lemma \ref{lemma_second-order-samplediff} implies that 
	\beq
	I_1 =& (n\delta)^{1/2}|\bv^{*T}\nabla^2_{\cdot\bgamma}R^{n_{(1)}}_{\delta}(\bbeta^*)(\wh{\bbeta}^{(2)} - \bbeta^*)| + o_{\PP}(1)\\
	\leq&(n\delta)^{1/2}\norm{\bv^*}_1\big(
	\norm{\nabla^2_{\cdot\bgamma}R^{n_{(1)}}_{\delta}(\bbeta^*) - \nabla^2_{\cdot\bgamma}R_{\delta}(\bbeta^*)}_{\max} + \norm{\nabla^2_{\cdot\bgamma}R_{\delta}(\bbeta^*) - \nabla^2_{\cdot\bgamma}R(\bbeta^*)}_{\max}
	\big)\norm{\wh{\bbeta}^{(2)} - \bbeta^*}_1\\
& + o_{\PP}(1)\\
	\lesssim& (n\delta)^{1/2}\norm{\bv^*}_1\frac{\eta_1(n)}{\delta}(\sqrt{\frac{\log(d)}{n\delta}} + \delta^{\ell})+ o_{\PP}(1)\\
	=&o_{\PP}(1).
	\eeq
	Similarly, since $\nabla R(\bbeta^*) = 0$, imply that
	\beq
	&I_2 = (n\delta)^{1/2}|(\wh{\bv}^{(1)} - \bv^{*})(\nabla R_\delta^{n_{(1)}}(\wh{\bbeta}_0^{(2)}) - \nabla R(\bbeta^*))|\\
	\leq&(n\delta)^{1/2}\norm{\wh{\bv}^{(1)} - \bv^*}_1\norm{\nabla R_\delta^{n_{(1)}}(\wh{\bbeta}_0^{(2)}) - \nabla R(\bbeta^*)}_\infty\\
	\leq&(n\delta)^{1/2}\norm{\wh{\bv}^{(1)} - \bv^*}_1
	\bigg(
	\norm{\nabla R_\delta^{n_{(1)}}(\wh{\bbeta}_0^{(2)}) - \nabla R_\delta(\wh{\bbeta}_0^{(2)})}_\infty\\ &~~~~~~~~~~~~~~~~~~~~~~~~~~~~~ + \norm{\nabla R_\delta(\wh{\bbeta}_0^{(2)}) - \nabla R(\wh{\bbeta}_0^{(2)})}_\infty +
	\norm{\nabla R(\wh{\bbeta}_0^{(2)}) - \nabla R(\bbeta^*)}_\infty
	\bigg).
	\eeq
	Since $\wh{\bbeta}_0^{(2)}$ depends on the set of samples that is disjoint with $\cN_1$, Lemma \ref{lemma_variance} together with Lemma \ref{lemma_addi_conditional} implies that
	\[
	\norm{\nabla R_\delta^{n_{(1)}}(\wh{\bbeta}_0^{(2)}) - \nabla R_\delta(\wh{\bbeta}_0^{(2)})}_\infty = \cO_{\PP}(\sqrt{\frac{\log(d)}{n\delta}}).
	\]
	This in combine with Lemma \ref{lemma_bias} and \ref{lemma_sample_diff} further implies that
	$$
	I_2 \lesssim (n\delta)^{1/2}\norm{\bv^*}_1\eta_2(n)\bigg(
	M_n \eta_1(n) + \sqrt{\frac{\log(d)}{n\delta}} + \delta^{\ell}
	\bigg) = o_{\PP}(1).	
	$$
	Putting all above together with the results from Lemma~\ref{lemma_consistent_bias_pilot} and Lemma~\ref{prop_var_pilot} as well as the Slutsky's theorem, we obtain that the bias corrected decorrelated score statistic $\wh{U}_n\stackrel{d}{\rightarrow}  N(0,1)$ under the null hypothesis.

\subsection{Proof of Theorem \ref{theorem_power}}\label{pf_theorem_power}
In the following, we shall prove the following general version of Theorem \ref{theorem_power}.

\begin{theorem}\label{theorem_power_2}
Under the conditions in Theorem \ref{theorem_score} and Lemmas \ref{lemma_consistent_bias_pilot} and \ref{prop_var_pilot} in the supplemental materials, we further assume
	\beq\label{eq_power_condition}
	&\norm{\bv^*}_1^2M_n^4n^{1-4\phi}/\delta = o(1),\quad (n\delta)^{1/2}
	\norm{\bv^*}_1(\eta_1(n)\vee \eta_2(n))M_n n^{-\phi} = o(1),
	\eeq
	 and $\wh{\mu}, \wh{\sigma}$ are consistent estimators of $\mu^*, \sigma^*$. Then the following results hold under the local alternative hypothesis $H_1: \theta^* = \tilde{C}n^{-\phi}$.
	\begin{itemize}
		\item If $\tilde{C}(n\delta)^{1/2}n^{-\phi}\nabla^2_{\theta|\bgamma}R(\bbeta^*)\sigma^{*-1} \rightarrow \xi$ for some constant $\xi$, then it holds that
		\beq\label{eq_theorem_power_0}
		\wh{U}_n  \stackrel{d}{\rightarrow} N(-\xi,1).
		\eeq
		\item If $\tilde{C}(n\delta)^{1/2}n^{-\phi}\nabla^2_{\theta|\bgamma}R(\bbeta^*)\sigma^{*-1} \rightarrow \infty$, then for any fixed $t$, it holds that \beq\label{eq_theorem_power_1}
		\lim_{n\rightarrow\infty}\PP(|\wh{U}_n| > t) = 1.
		\eeq
	\end{itemize}
\end{theorem}

Proof: Recall that $\wh{S}_\delta(0,\wh{\bgamma}) = \frac12 \big( \wh{S}^{(1)}_\delta(0,\wh{\bgamma}^{(2)}) + \wh{S}^{(2)}_\delta(0,\wh{\bgamma}^{(1)})  \big)$. Let's focus on $\wh{S}^{(1)}_\delta(0,\wh{\bgamma}^{(2)}) $. By definition
	\beq
	&\wh{S}^{(1)}(0,\wh{\bgamma}^{(2)})= \wh{\bv}^{(1)T}\nabla R_\delta^{n_{(1)}}(0,\wh{\bgamma}^{(2)})\\
	=&\underbrace{(\vdiffone)^T\nabla R_\delta^{n_{(1)}}(0,\wh{\bgamma}^{(2)})}_{I_1} + \bv^{*T}\nabla R_\delta^{n_{(1)}}(0,\wh{\bgamma}^{(2)})\\
	=&I_1 + \underbrace{\bv^{*T}\big(\nabla R_\delta^{n_{(1)}}(0,\wh{\bgamma}^{(2)}) - \nabla R_\delta^{n_{(1)}}(0,\bgamma^*)\big)}_{I_2} + \bv^{*T}\nabla R_\delta^{n_{(1)}}(0,\bgamma^*) \\
	=&I_1 + I_2 + \bv^{*T}\nabla R_\delta^{n_{(1)}}(\bbeta^*) - \theta^*\bv^{*T}\nabla^2_{\cdot\theta} R(\bbeta^*) + \underbrace{\bv^{*T}(\nabla R_\delta^{n_{(1)}}(0,\bgamma^*) - \nabla R_\delta^{n_{(1)}}(\bbeta^*)) + \theta^*\bv^{*T}\nabla^2_{\cdot\theta}R(\bbeta^*)}_{I_3}\\
	=&S^{(1)}(\bbeta^*)- \theta^*\bv^{*T}\nabla^2_{\cdot\theta}R(\bbeta^*) + I_1 + I_2 + I_3.
	\eeq
	For $I_1$, we know
	$$
	I_1 \leq \norm{\vdiffone}_1\norm{\nabla R_\delta^n(0,\wh{\bgamma})}_\infty = \cO_{\PP}(\norm{\bv^*}_1\eta_2(n)(\sqrt{\log(d)/(n\delta)} + \delta^{\ell} + M_n(\theta^* + \eta_1(n))))  = o_{\PP}((n\delta)^{-1/2}).
	$$
	Following a similar proof of Lemma \ref{lemma_second-order-samplediff}, we can show that $$I_2 = \cO_{\PP}\big(\norm{\bv^*}_1\eta_1(n)\big[
	\delta^{\ell-1} + \sqrt{\log(d)/(n\delta^3)} + M_n\theta^*
	\big]\big) = o_{\PP}((n\delta)^{-1/2}).$$
	Finally, for $I_3$, since by definition
	\beq
	\bv^{*T}\nabla^2_{\cdot\theta}R(\bbeta^*) = \nabla^2_{\theta|\bgamma}R(\bbeta^*),
	\eeq
	Lemma \ref{lemma_power_approximation} implies that  $I_3 = o_{\PP}((n\delta)^{-1/2})$.
	Put all pieces together, we have shown that
	\[
	|\wh{S}^{(1)}(0,\wh{\bgamma}^{(2)}) - (S^{(1)}(\bbeta^*)- \theta^*\nabla^2_{\theta|\bgamma}R(\bbeta^*))| = o_{\PP}((n\delta)^{-1/2}).
	\]
	A similar result will also hold for $\wh{S}^{(2)}(0,\wh{\bgamma}^{(1)})$, and thus we conclude
	\beq
	|\wh{S}(0,\wh{\bgamma}) - (S(\bbeta^*)- \theta^*\nabla^2_{\theta|\bgamma}R(\bbeta^*))| = o_{\PP}((n\delta)^{-1/2}).
	\eeq
	At this point the conclusion of (\ref{eq_theorem_power_0}) is shown. To show (\ref{eq_theorem_power_1}), the above formula also implies that
	\beq
	\PP(|(n\delta)^{1/2}\wh{S}(0,\wh{\bgamma})/\sigma^*| \leq t) \leq \PP(L(n) \leq (n\delta)^{1/2}S(\bbeta^*)/\sigma^* \leq U(n)) + o(1),
	\eeq
	where $L(n) = -t - q(n) + \tilde{C}\sigma^{*-1}\nabla^2_{\theta|\bgamma}R(\bbeta^*)n^{-\phi}(n\delta)^{1/2}$, $U(n) = t + q(n) + \tilde{C}\sigma^{*-1}\nabla^2_{\theta|\bgamma}R(\bbeta^*)n^{-\phi}(n\delta)^{1/2}$, and $q(n) = o(1)$ is some deterministic sequence. Since $n^{-\phi}(n\delta)^{1/2}\nabla^2_{\theta|\bgamma}R(\bbeta^*)\sigma^{*-1}\rightarrow\infty$, it is easily seen that $\PP(|(n\delta)^{1/2}\wh{S}(0,\wh{\bgamma})/\sigma^*| \leq t) \rightarrow 0$. Since $\wh{\sigma}$ is consistent, for $n$ large enough we will get $|\wh{\sigma}/\sigma - 1|\leq 3$, which finally implies the desired result.
	This completes the proof.

\subsection{Proof of Theorem \ref{theorem_linear_hypothesis} }\label{pf_theorem_linear_hypothesis}
The proof is very similar to the proof of Theorem \ref{theorem_score} so we only give a sketch here. Following the same derivation of Lemma \ref{lemma_ori_normality_new}, we can show that
\beq
(n\delta)^{1/2}\frac{S_\delta^L(0,\bgamma^*,\bomega^*_L) - \delta^{\ell}\mu^*_L}{\sigma^*_L} \stackrel{d}{\rightarrow} N(0,1).
\eeq

Next, following the derivation of Theorem \ref{theorem_score}, we can show that $I_1 = (n\delta)^{1/2}|\bv^{*T}_L\bC(\nabla R_\delta^{n_{(1)}}(\wh{\bbeta}_0^{(2)}) - \nabla R_\delta^{n_{(1)}}(\bbeta^*))|$ and $I_2 = (n\delta)^{1/2}|(\wh{\bv}_L^{(1)} - \bv_L^{*})\bC\nabla R_\delta^{n_{(1)}}(\wh{\bbeta}_0^{(2)})|$ are $o_{\PP}(1)$, which further implies that $(n\delta)^{1/2}|S_\delta^{L(1)}(0,\wh\bgamma^{(2)},\wh\bomega^{(1)}_L) - S_\delta^L(0,\bgamma^*,\bomega^*_L)| = o_{\PP}(1)$. This shall hold similarly for $S_\delta^{L(2)}(0,\wh\bgamma^{(1)},\wh\bomega^{(2)}_L)$. Invoking Slutsky's theorem completes the proof.


\subsection{Proofs of additional technical lemmas}

\begin{lemma}\label{lemma_bias}
	Under Assumptions \ref{ass_smooth} - \ref{ass_moment}, for any fixed $\bbeta$, we have\beq
	&\norm{\nabla R_\delta(\bbeta) - \nabla R(\bbeta)}_{\infty} \lesssim \delta^{\ell},\\
	&\norm{\nabla^2 R_\delta(\bbeta) - \nabla R(\bbeta)}_{\max }\lesssim \delta^{\ell-1}.
	\eeq
\end{lemma}

\begin{proof}[Proof of Lemma \ref{lemma_bias}]
	We focus on proving the result for $\norm{\nabla^2 R_\delta(\bbeta) - \nabla^2 R(\bbeta)}_{\max }$ and the proof for the other one is very similar. By definition
	\beq
	&\norm{\nabla^2 R_\delta(\bbeta) - \nabla^2 R(\bbeta)}_{\max }\\
	=&\max_{j,k} |\sum_y w(y)\int_{\bz} y z_jz_k (\int \frac{-1}{\delta^2}K'(\frac{x - \bbeta^T\bz}{\delta}) f(x|y,\bz) dx - f'(\bbeta^T\bz|y,\bz)) f(y,\bz) d\bz|\\
	=&\max_{j,k} |\sum_y w(y)\int_{\bz} y z_jz_k (\int \frac{-1}{\delta}K'(u) f(u\delta + \bbeta^T\bz|y,\bz) du - f'(\bbeta^T\bz|y,\bz)) f(y,\bz) d\bz|\\
	=&\max_{j,k} |\sum_y w(y)\int_{\bz} y z_jz_k \int K(u) (f'(u\delta + \bbeta^T\bz|y,\bz)  - f'(\bbeta^T\bz|y,\bz))du f(y,\bz) d\bz|\\
	=&\max_{j,k} |\sum_y w(y)\int_{\bz} y z_jz_k \int K(u)\frac{(u\delta)^{\ell-1}}{(\ell-1)!} f^{(\ell)}(\tau u\delta + \bbeta^T\bz|y,\bz)du f(y,\bz) d\bz|\\
	\lesssim&\cO(\delta^{\ell-1}),
	\eeq
	where the first equality is by definition, the second equality follows from a change of variable, the third equality follows from an integration by parts, the last equality follows from Assumptions \ref{ass_smooth}, \ref{ass_kernel} and the last inequality follows from Assumptions \ref{ass_proportion}, \ref{ass_moment}. The proof is complete.
\end{proof}

\begin{lemma}\label{lemma_variance}Under Assumptions \ref{ass_smooth} - \ref{ass_moment}, for any fixed $\bbeta$, we have with probability greater than $1-\cO(d^{-1})$\beq
	&\norm{\nabla R^n_\delta(\bbeta) - \nabla R_\delta(\bbeta)}_{\infty} \lesssim \sqrt{\frac{\log(d)}{n\delta}},\\
	&\norm{\nabla^2 R^n_\delta(\bbeta) - \nabla^2 R_\delta(\bbeta)}_{\max }\lesssim \sqrt{\frac{\log(d)}{n\delta^3}}.
	\eeq
\end{lemma}

\begin{proof}[Proof of Lemma \ref{lemma_variance}]
	We focus on proving the result for $\norm{\nabla^2 R^n_\delta(\bbeta) - \nabla^2 R_\delta(\bbeta)}_{\max }$ and the proof for the other one is very similar. Denote
	\beq
	T_{ijk} =& (\nabla^2 \bar{R}^i_\delta(\bbeta) - \nabla^2 R_\delta(\bbeta))_{jk}\\
	=&-w(y_i)y_i\frac{z_{ij}z_{ik}}{\delta^2}K'(\frac{x_i - \bbeta^T\bz_i}{\delta}) - (\nabla^2 R_\delta(\bbeta))_{jk}.
	\eeq
	By definition, we know $\EE[T_{ijk}] = 0$. Meanwhile for $Var[T_{ijk}]$, we know
	\beq
	\EE[(\nabla^2 \bar{R}^i_\delta(\bbeta))_{jk}^2] =& \sum_y \int w(y)^2 \frac{z_j^2z_k^2}{\delta^4}K'^2(\frac{x - \bbeta^T\bz}{\delta})f(x|y,\bz)dx f(y,\bz)d\bz\\
	=&\sum_y \int w(y)^2\frac{z_j^2z_k^2}{\delta^3}K'^2(u)f(u\delta + \bbeta^T\bz|y,\bz)du f(y,\bz)d\bz\\
	=&\cO(\frac{1}{\delta^3}),
	\eeq
	and with a similar derivation $(\nabla^2 R_\delta(\bbeta))^2_{jk} = \cO(\frac{1}{\delta^2})$. This shows that $Var[T_{ijk}] = \cO(\frac{1}{\delta^3})$. Since $|T_{ijk}|\lesssim \frac{M_n^2}{\delta^2}$, by Bernstein inequality we can show that with probability greater than $1 - \cO(d^{-1})$
	\beq
	\norm{\nabla^2 R^n_\delta(\bbeta) - \nabla^2 R_\delta(\bbeta)}_{\max }\lesssim \sqrt{\frac{\log(d)}{n\delta^3}}.
	\eeq
	This completes the proof.
\end{proof}

\begin{lemma}\label{lemma_sample_diff}
	Under Assumptions \ref{ass_smooth} - \ref{ass_moment},  it holds that\beq
	&\norm{\nabla R (\wh{\bbeta}) - \nabla R (\bbeta^*)}_{\infty} \lesssim M_n\norm{\wh{\bbeta} - \bbeta^*}_1,\\
	&\norm{\nabla^2 R (\wh{\bbeta}) - \nabla^2 R (\bbeta^*)}_{\max} \lesssim M_n\norm{\wh{\bbeta} - \bbeta^*}_1.
	\eeq
\end{lemma}

\begin{proof}[Proof of Lemma \ref{lemma_sample_diff}]
	Here we prove the second inequality and the first one should follow similarly. By definition,
	\beq
	&\norm{\nabla^2 R (\wh{\bbeta}) - \nabla^2 R (\bbeta^*)}_{\max}\\
	=&\max_{j,k}\bigg|
	\sum_yw(y)\int_{\bz}z_jz_k \big[
	f'(\wh{\bbeta}^T\bz|y,\bz) - f'(\bbeta^{*T}\bz|y,\bz)
	\big]f(y,\bz)d\bz\bigg|\\
	\leq&\norm{\wh{\bbeta} - \bbeta^*}_1M_n|f''|_\infty \max_{j,k}
	\sum_{y}\EE[|Z_jZ_k||Y = y]\\
	\lesssim&M_n\norm{\wh{\bbeta} - \bbeta^*}_1.
	\eeq
\end{proof}

\begin{lemma}\label{lemma_second-order-samplediff}
	Under the conditions in Theorem \ref{theorem_score}, we have
	\[
	(n\delta)^{1/2}\bigg|
	\bv^{*T}\bigg(
	\nabla R_\delta^{n_{(j)}}(\wh{\bbeta}^{(k)}) - \nabla R_\delta^{n_{(j)}}(\bbeta^*) - \nabla^2 R_\delta^{n_{(j)}}(\bbeta^*)(\wh{\bbeta}^{(k)} - \bbeta^*)
	\bigg)
	\bigg| = o_{\PP}(1),
	\]
	for $(j,k)\in \{(1,2),(2,1)\}$.
\end{lemma}
\begin{proof}[Proof of Lemma \ref{lemma_second-order-samplediff}]
	With some algebra we obtain
	\beq
	&\bigg|\bv^{*T}\bigg(
	\nabla R_\delta^{n_{(j)}}(\wh{\bbeta}^{(k)}) - \nabla R_\delta^{n_{(j)}}(\bbeta^*) - \nabla^2 R_\delta^{n_{(j)}}(\bbeta^*)(\wh{\bbeta}^{(k)} - \bbeta^*)
	\bigg)\bigg|\\
	=& \bigg|\frac{1}{\cN_j}\sum_{i\in\cN_j}w(y_i)\frac{\bz_i^T\bv^*y_i}{\delta}\int_{\frac{x_i - \bbeta^{*T}\bz_i}{\delta}}^{\frac{x_i - \wh{\bbeta}^{(k)T}\bz_i}{\delta}}K''(t)(\frac{x_i - \wh{\bbeta}^{(k)T}\bz_i}{\delta} - t)dt\bigg|\\
	\leq&\norm{\bv^*}_1\bignorm{\frac{1}{\cN_j}\sum_{i\in\cN_j} w(y_i) \frac{\bz_iy_i}{\delta}\int_{\frac{x_i - \bbeta^{*T}\bz_i}{\delta}}^{\frac{x_i - \wh{\bbeta}^{(k)T}\bz_i}{\delta}}K''(t)(\frac{x_i - \wh{\bbeta}^{(k)T}\bz_i}{\delta} - t)dt}_\infty.
	\eeq
	Now we start to analyze $\bignorm{\frac{1}{\cN_j}\sum_{i\in\cN_j} w(y_i) \frac{\bz_iy_i}{\delta}\int_{\frac{x_i - \bbeta^{*T}\bz_i}{\delta}}^{\frac{x_i - \wh{\bbeta}^{(k)T}\bz_i}{\delta}}K''(t)(\frac{x_i - \wh{\bbeta}^{(k)T}\bz_i}{\delta} - t)dt}_\infty$. Denote $G_i=  w(y_i) \frac{\bz_iy_i}{\delta}\int_{\frac{x_i - \bbeta^{*T}\bz_i}{\delta}}^{\frac{x_i - \wh{\bbeta}^{(k)T}\bz_i}{\delta}}K''(t)(\frac{x_i - \wh{\bbeta}^{(k)T}\bz_i}{\delta} - t)dt$ and $G_{im},1\leq m\leq d$ as its coordinates.
	Consider the event $A := \{\norm{\wh{\bbeta}^{(k)} - \bbeta^*}_{1} \lesssim C\eta_1(n)\}$ for some constant $C$. Notice that for each $m$
	\beq
	&\sum_{y}w(y)\int_{\bz}\frac{z_my}{\delta}\int_x	\int_{\frac{x - \bbeta^{*T}\bz}{\delta}}^{\frac{x - \wh{\bbeta}^{(k)T}\bz}{\delta}}K''(t)(\frac{x - \wh{\bbeta}^{(k)T}\bz}{\delta} - t)dtf(x|y,\bz)dxf(y,\bz)d\bz\\
	\stackrel{(u = (x - \bbeta^{*T}\bz)/\delta)}{=}&\sum_{y}w(y)\int_{\bz}z_my\int_u	\int_{u}^{u+\triangle}K''(t)(u+\triangle- t)dtf(u\delta+\bbeta^{*T}\bz|y,\bz)duf(y,\bz)d\bz\\
	=&\sum_{y}w(y)\int_{\bz}z_my\int_u	\int_{u}^{u+\triangle}K''(t)(u+\triangle- t)dtf(\bbeta^{*T}\bz|y,\bz)duf(y,\bz)d\bz\\
	&+ \sum_{y}w(y)\int_{\bz}z_my\int_u	u\delta \int_{u}^{u+\triangle}K''(t)(u+\triangle- t)dtf'(\tau u\delta+\bbeta^{*T}\bz|y,\bz)duf(y,\bz)d\bz,
	\eeq
	where $\triangle = \frac{(\wh{\bbeta}^{(k)} - \bbeta)^T\bz}{\delta}$ and $\tau \in [0,1]$. Here the first step is by definition and the last step is from the mean value theorem. Since
	\beq
	\int_u	\int_{u}^{u+\triangle}K''(t)(u+\triangle- t)dtdu = 0,
	\eeq
	we can show that the first term on the RHS of the last step is 0. The second term on the RHS can be bounded by $C'\delta|\triangle|^2 \lesssim M_n^2\eta_1(n)^2/\delta$ for some constant $C'$ on event $A$. This implies that $\EE[G_{im}|A] \lesssim M_n^2\eta_1(n)^2/\delta$.
	
	Now we look at its variance. For the second moment, for each $i\in \cN_j,m = 1,\dotso,d$, we have
	\beq
	&\EE\bigg[\bigg(w(y_i)
	\frac{z_{im}y_i}{\delta}\int_{\frac{x_i - \bbeta^{*T}\bz_i}{\delta}}^{\frac{x_i - \wh{\bbeta}^{(k)}\bz_i}{\delta}}K''(t)(\frac{x_i - \wh{\bbeta}^{(k)}\bz_i}{\delta} - t)dt\bigg)^2\bigg | A
	\bigg]\\
	=&\EE\bigg[\sum_{y}w(y)^2\int_{\bz}\frac{z_m^2}{\delta}\int_u	\bigg(\int_{u}^{u+\triangle}K''(t)(u+\triangle- t)dt\bigg)^2f(u\delta+\bbeta^{*T}\bz|y,\bz)duf(y,\bz)d\bz| A
	\bigg]\\
	\leq&\EE\bigg[\sum_{y}w(y)^2\int_{\bz}\frac{z_m^2}{\delta}2|f|_\infty|K''|^2_\infty\triangle^4f(y,\bz)d\bz| A
	\bigg]\\
	\lesssim&\frac{M^4_n \eta_1(n)^4}{\delta^5}.
	\eeq
	Also we know for each $i\in \cN_j,m = 1,\dotso,d$, $\bigg|\frac{z_{im}y_i}{\delta}\int_{\frac{x_i - \bbeta^{*T}\bz_i}{\delta}}^{\frac{x_i- \wh{\bbeta}^{(k)}\bz_i}{\delta}}K''(t)(\frac{x_i - \wh{\bbeta}^{(k)T}\bz_i}{\delta} - t)dt\bigg|\lesssim \frac{M_n}{\delta}(\frac{M_n\eta_1(n)}{\delta})^2$ on $A$, and therefore
	\beq
	Var[G_{im}|A] \lesssim \frac{M^4_n \eta_1(n)^4}{\delta^5}.
	\eeq
	Therefore, applying Bernstein inequality with $M_n \sqrt{\frac{\log(d)}{n\delta}} = \cO(1)$, we can obtain that
	\beq
	\PP\bigg(\max_m|\frac{1}{\cN_j}\sum_{i\in \cN_j}G_{im} - \EE G_m| > \frac{M_n^2\eta_1(n)^2}{\delta}\sqrt{\frac{\log(d)}{n\delta^3}}~\big|A \bigg)\leq \cO(d^{-1}),
	\eeq
	which further implies that
	\beq
	&\PP\bigg(\max_m|\frac{1}{\cN_j}\sum_{i\in \cN_j}G_{im} - \EE G_m|> \frac{M_n^2\eta_1(n)^2}{\delta}\sqrt{\frac{\log(d)}{n\delta^3}}\bigg)\\
	\leq& \PP\bigg(\max_m|\frac{1}{\cN_j}\sum_{i\in \cN_j}G_{im} - \EE G_m| > \frac{M_n^2\eta_1(n)^2}{\delta}\sqrt{\frac{\log(d)}{n\delta^3}}~\big|A \bigg)+ \PP(A^C) \\
	= &  o(1),
	\eeq
	where the last step follows from Assumption \ref{ass_estimators}.
	\beq
	\bignorm{\frac{1}{\cN_j}\sum_{i\in \cN_j} w(Y_i)\frac{\bZ_iY_i}{\delta}\int_{\frac{X_i - \bbeta^{*T}\bZ_i}{\delta}}^{\frac{X_i - \wh{\bbeta}^{(k)}\bZ_i}{\delta}}K''(t)(\frac{X_i - \wh{\bbeta}^{(k)}\bZ_i}{\delta} - t)dt}_\infty =\cO_{\PP}(\frac{M_n^2\eta_1(n)^2}{\delta}).
	\eeq
	Under the conditions of Theorem \ref{theorem_score}, the desired result holds. This completes the proof.
\end{proof}

\begin{lemma}\label{lemma_RE}
    Under the conditions of Theorem \ref{theorem_score},
	let $s' = \norm{\bomega^*}_1$ and $\xi > 0$ is some constant. If $s'\big(M_n\eta_1(n)\vee \sqrt{\frac{\log(d)}{n\delta^3}}\vee \delta^{\ell-1}\big) = o(1)$, then with probability tending to one it holds that $\kappa_D(s') \geq \kappa/\sqrt{2}$, where
	$$
	\kappa_D(s') = \min\bigg \{
	\frac{s'^{1/2} (\bv^T \nabla^2_{\bgamma\bgamma} R^n_{\delta}(\wh{\bbeta})\bv)^{1/2} } {\norm{\bv_{s'}}_1 } : \bv \in \RR^{d-1} \backslash\{0\}, \norm{\bv_{s'^{c}}}_1 \leq \xi \norm{\bv_{s'}}_1
	\bigg \}.
	$$
\end{lemma}

\begin{proof}
	This proof is similar to the proof of Lemma J.3 in \cite{ning2017general} so we only give a sketch here. Firstly we have
	\beq
	\kappa_D(s')^2 \geq \min\bigg \{
	\frac{\bv^T \nabla^2_{\bgamma\bgamma} R^n_{\delta}(\wh{\bbeta})\bv } {\norm{\bv}_2^2 } : \bv \in \RR^{d-1} \backslash\{0\}, \norm{\bv_{s'^{c}}}_1 \leq \xi \norm{\bv_{s'}}_1
	\bigg \}.
	\eeq
	Similar to the proof of Theorem \ref{theorem_score}, Lemmas \ref{lemma_bias}, \ref{lemma_variance}, \ref{lemma_sample_diff} and \ref{lemma_addi_conditional} together imply that
	\beq
	&\bigg|\frac{\bv^T (\nabla^2_{\bgamma\bgamma} R^n_{\delta}(\wh{\bbeta}) - \nabla^2_{\bgamma\bgamma} R(\bbeta^*))\bv } {\norm{\bv}_2^2 }\bigg|\\
	=&\bigg|\frac{\bv^T\big[\frac12 \big(  \nabla^2_{\bgamma\bgamma} R^{n_{(1)}}_{\delta}(\wh{\bbeta}^{(2)}) + \nabla^2_{\bgamma\bgamma} R^{n_{(2)}}_{\delta}(\wh{\bbeta}^{(1)}) \big)
		- \nabla^2_{\bgamma\bgamma} R(\bbeta^*)\big]\bv } {\norm{\bv}_2^2 }\bigg|\\
	\lesssim & s'(\xi + 1 )^2 \big [M_n\eta_1(n)\vee \sqrt{\frac{\log(d)}{n\delta^3}}\vee \delta^{\ell-1}
	\big ] = o_{\PP}(1),
	\eeq
	where the inequality is because $\norm{\bv}_1^2 \leq s'(\xi + 1 )^2 \norm{\bv}_2^2$.
	Therefore for $n$ large enough, we have
	$|\frac{\bv^T (\nabla^2_{\bgamma\bgamma} R^n_{\delta}(\tilde{\bbeta}) - \nabla^2_{\bgamma\bgamma} R(\bbeta^*))\bv } {\norm{\bv}_2^2 }| \leq \frac12\kappa^2$. This implies $\kappa_D(\tilde{s}^*)^2\geq \frac12\kappa^2$ with probability tending to 1. This completes the proof.
\end{proof}

\begin{lemma}\label{lemma_power_approximation}
	Under the same conditions as in Theorem \ref{theorem_power}, for $j = 1,2$, it holds that
	\[
	(n\delta)^{1/2}|\bv^{*T}(\nabla R_\delta^{n_{(j)}}(0,\bgamma^*) - \nabla R_\delta^{n_{(j)}}(\bbeta^*)) + \theta^*\bv^{*T}\nabla^2_{\cdot\theta}R(\bbeta^*)| = o_{\PP}(1).
	\]
\end{lemma}
\begin{proof}
	By definition
	\beq
	& |\bv^{*T}(\nabla R_\delta^{n_{(j)}}(0,\bgamma^*) - \nabla R_\delta^{n_{(j)}}(\bbeta^*)) + \theta^*\bv^{*T}\nabla^2_{\cdot\theta}R(\bbeta^*)| \\
	&\leq\norm{\bv^*}_1
	\bigg[
	\underbrace{\norm{\nabla R_\delta^{n_{(j)}}(0,\bgamma^*) - \nabla R_\delta^{n_{(j)}}(\bbeta^*) + \theta^*\nabla^2_{\cdot\theta}R^{n_{(j)}}_\delta(\bbeta^*)}_\infty}_{I_1} + \\
	&~~~~~~~~~\underbrace{\norm{\theta^*(\nabla^2_{\cdot\theta}R^{n_{(j)}}_\delta(\bbeta^*) - \nabla^2_{\cdot\theta}R_\delta(\bbeta^*))}_\infty}_{I_2} + \underbrace{\norm{\theta^*(\nabla^2_{\cdot\theta}R_\delta(\bbeta^*) - \nabla^2_{\cdot\theta}R(\bbeta^*))}_\infty}_{I_3}
	\bigg].
	\eeq
	For $I_1$, similar to Lemma \ref{lemma_second-order-samplediff}, we can write
	\beq
	&|\nabla R_\delta^{n_{(j)}}(0,\bgamma^*) - \nabla R_\delta^{n_{(j)}}(\bbeta^*) + \theta^*\nabla^2_{\cdot\theta}R^{n_{(j)}}_\delta(\bbeta^*)|\\
	=&\bigg|\frac{1}{\cN_j}\sum_{i\in\cN_j}w(y_i)\frac{y_i\bz_i}{\delta}\int_{\frac{x_i - \bbeta^{*T}\bz_i}{\delta}}^{\frac{x_i - \bbeta_0^{*T}\bz_i}{\delta}}K''(t)(\frac{x_i - \bbeta_0^{*T}\bz_i}{\delta} - t)dt\bigg|\\
	:=&|\frac{1}{\cN_j}\sum_{i\in\cN_j} G_i|,
	\eeq
	where $G_i = w(y_i)\frac{y_i\bz_i}{\delta}\int_{\frac{x_i - \bbeta^{*T}\bz_i}{\delta}}^{\frac{x_i - \bbeta_0^{*T}\bz_i}{\delta}}K''(t)(\frac{x_i - \bbeta_0^{*T}\bz_i}{\delta} - t)dt \in \RR^d$. Thus, similar to the derivation of Lemma \ref{lemma_second-order-samplediff}, for each $1\leq m\leq d$, we can show that
	$\EE[G_{im}]\lesssim M_n^2\theta^{*2}/\delta$, $|G_{im}|\lesssim M_n^3\theta^{*2}/\delta^3$ and $Var [G_{im}]\lesssim M_n^4\theta^{*4}/\delta^5$, and thus applying Bernstein inequality yields
	\[
	I_1 = \cO_{\PP}(\frac{M_n^2\theta^{*2}}{\delta}).
	\]
	Meanwhile, Lemma \ref{lemma_variance} and Lemma \ref{lemma_bias} imply that $I_2 = \cO_{\PP}(\theta^*\sqrt{\frac{\log(d)}{n\delta^3}})$ and $I_3 = \cO(\theta^*\delta^{\ell-1})$. Combing the above results, we obtain
	\beq
	|\bv^{*T}(\nabla R_\delta^{n_{(j)}}(0,\bgamma^*) - \nabla R_\delta^{n_{(j)}}(\bbeta^*)) + \theta^*\bv^{*T}\nabla^2_{\cdot\theta}R(\bbeta^*)| = \cO_{\PP}(\frac{\norm{\bv^*}_1\theta^*}{\delta}(M_n^2\theta^* \vee\sqrt{\frac{\log(d)}{n\delta}} \vee \delta^\ell)).
	\eeq
	This together with the condition in Theorem \ref{theorem_power} completes the proof.
\end{proof}

\begin{lemma}\label{theorem_projection}
	Let $s'= \norm{\bomega^{*}}_0$. Suppose Assumptions \ref{ass_smooth}-\ref{ass_projection_norm} hold, and $\lambda_{\min}(\nabla^2_{\bgamma,\bgamma}R(\bbeta^*))\geq c$ for some constant $c>0$.
	If we choose $\delta \asymp (\log(d)/n)^{1/(2\ell+1)}$ and tuning parameter $\lambda' \asymp \norm{\bomega^*}_1(M_n\eta_1(n) + (\log(d)/n)^{(\ell-1)/(2\ell+1)})$ and $s'(M_n\eta_1(n) + (\log(d)/n)^{(\ell-1)/(2\ell+1)}) = o(1)$, then it holds that
	\beq\label{eq_projection_rate}
	\norm{\wh{\bomega} - \bomega^*}_1 \lesssim \norm{\bv^*}_1 s'(M_n\eta_1(n) + (\log(d)/n)^{(\ell-1)/(2\ell+1)}).
	\eeq
\end{lemma}

\begin{proof}
	Here we focus on the rate for $\wh{\bomega}^{(1)}$ and the result will follow accordingly.
	Denote $\wh{\triangle} = \wh{\bomega}^{(1)} - \bomega^*$. By definition, we can show that $\norm{\wh{\triangle}_{\tilde{s}^{*c}}}_1 \leq \norm{\wh{\triangle}_{\tilde{s}^{*}}}_1$.
	
	With some algebra, we will get
	\beq
	\wh{\triangle}^T\nabla^2_{\bgamma\bgamma} R^{n_{(1)}}_{\delta}(\wh{\bbeta}^{(2)})\wh{\triangle}
	&= \wh{\triangle}^T(\nabla^2_{\bgamma\cdot}R^{n_{(1)}}_{\delta}(\wh{\bbeta}^{(2)}) \bv^*) - \wh{\triangle}^T(\nabla^2_{\bgamma\cdot} R^{n_{(1)}}_{\delta}(\wh{\bbeta}^{(2)}) \wh{\bv}^{(1)})\\
	&\leq\norm{\wh{\triangle}}_1\norm{\nabla^2_{\bgamma\cdot}R^{n_{(1)}}_{\delta}(\wh{\bbeta}^{(2)})\bv^*}_\infty + \norm{\wh{\triangle}}_1
	\norm{\nabla^2_{\bgamma\theta} R^{n_{(1)}}_{\delta}(\wh{\bbeta}^{(2)}) - \nabla^2_{\bgamma\bgamma} R^{n_{(1)}}_{\delta}(\wh{\bbeta}^{(2)})\wh{\bomega}^{(1)} }_\infty.
	\eeq
	By definition, the second term is bounded by $\lambda'\norm{\wh{\triangle}}_1$. For the first term, recall by definition $\nabla^2_{\bgamma\cdot}R(\bbeta^*)\bv^* = 0$, and thus similar to the proof of Theorem \ref{theorem_score}, Lemma \ref{lemma_bias}, \ref{lemma_variance} and \ref{lemma_sample_diff} imply that with probability approaching to 1
	\[
	\norm{\nabla^2_{\bgamma\cdot}R^{n_{(1)}}_{\delta}(\wh{\bbeta}^{(2)})\bv^*}_\infty \lesssim \norm{\bv^*}_1\bigg[M_n\eta_1(n) \vee \sqrt{\frac{\log(d)}{n\delta^3}}\vee \delta^{\ell-1}
	\bigg ],
	\]
	and thus with the choice of $\lambda'$ it holds that $\wh{\triangle}^T\nabla^2_{\bgamma\bgamma} R^{n_{(1)}}_{\delta}(\wh{\bbeta}^{(2)})\wh{\triangle} \lesssim \lambda'\norm{\wh{\triangle}}_1$. This together with Lemma \ref{lemma_RE} implies that $\norm{\wh{\triangle}}_1\lesssim s'\lambda'$ with high probability. This completes the proof.
\end{proof}

\begin{lemma}\label{lemma_consistent_bias_pilot}
	Under the same conditions of Theorem \ref{theorem_score}, if $U$ is a proper kernel of order $\ell$ satisfying the same condition as $K$ in Assumption \ref{ass_kernel}, and in addition $U$ is $\ell$ times continuously differentiable and $U^{(i)}$ degenerates at the boundary for $i = 0,\dotso,\ell-1$, then when  $\sqrt{\frac{\log(d)}{nh^{2\ell+1}}} + (M_n\eta_1(n) \vee h)^\zeta= o(1)$, it holds that
	\[
	|\wh{\mu} - \bv^{*T}\bb^*| \lesssim \norm{\bv^*}_1\bigg(\eta_2(n) + \sqrt{\frac{\log(d)}{nh^{2\ell+1}}} + (M_n\eta_1(n) \vee h)^\zeta\bigg).
	\]
\end{lemma}
\begin{proof}
	It suffices to proof the results for $\wh{\bv}^{(1)T}\wh{T}_{h,U}^{(\ell),n_{(1)}}(\wh{\bbeta}^{(2)})$. By definition,
	\beq \label{eq_bias_pilot_-1}
	|\bv^{*T}T^{(\ell)}(\bbeta^*) - \wh{\bv}^{(1)T}\wh{T}_{h,U}^{(\ell),n_{(1)}}(\wh{\bbeta}^{(2)})|\leq
	|(\bv^* - \wh{\bv})^T \wh{T}_{h,U}^{(\ell),n_{(1)}}(\wh{\bbeta}^{(2)}) | + |\bv^{*T}(T^{(\ell)}(\bbeta^*) - \wh{T}_{h,U}^{(\ell),n_{(1)}}(\wh{\bbeta}^{(2)}))|.
	\eeq
	We firstly look at $\norm{T^{(\ell)}(\bbeta^*) - \wh{T}_{h,U}^{(\ell),n_{(1)}}(\wh{\bbeta}^{(2)})}_\infty$. Direct calculation gives that
	\beq\label{eq_bias_pilot_0}
	&\norm{T^{(\ell)}(\bbeta^*) - \wh{T}_{h,U}^{(\ell),n_{(1)}}(\wh{\bbeta}^{(2)})}_\infty\\
	\leq&
	\norm{T^{(\ell)}(\bbeta^*) - T^{(\ell)}(\wh{\bbeta}^{(2)})}_\infty +
	\norm{T^{(\ell)}(\wh{\bbeta}^{(2)}) - \tilde{T}_{h,U}^{(\ell)}(\wh{\bbeta}^{(2)})}_\infty +
	\norm{\tilde{T}_{h,U}^{(\ell)}(\wh{\bbeta}^{(2)}) - \wh{T}_{h,U}^{(\ell),n_{(1)}}(\wh{\bbeta}^{(2)})}_\infty,
	\eeq
	where $\tilde{T}_{h,U}^{(\ell)}(\wh{\bbeta}^{(2)}) := \sum_{y \in \{-1,1\}}w(y)y\int \frac{ \bz}{h^{1+\ell}}\int U^{(\ell)}\bigg(\frac{\wh{\bbeta}^{(2)T}\bz - x}{h}\bigg)f(x|y,\bz)dx f(y,\bz)d\bz$. For the first term, we have
	\beq
	&\norm{T^{(\ell)}(\bbeta^*) - T^{(\ell)}(\wh{\bbeta}^{(2)})}_\infty\\
	=&\bignorm{\sum_{y \in \{-1,1\}}w(y)\int y\bz (f^{(\ell)}(\bbeta^{*T}\bz|y,\bz) - f^{(\ell)}(\wh{\bbeta}^{(2)T}\bz|y,\bz))f(y,\bz) d\bz}_\infty\\
	=&\cO_{\PP}((M_n\eta_1(n))^{\zeta} ),
	\eeq
	where the last step follows from Assumption \ref{ass_smooth}. For the second term, notice that we can rewrite
	\beq
	&\tilde{T}_{h,U}^{(\ell)}(\wh{\bbeta}^{(2)})\\
	=&\sum_{y \in \{-1,1\}}w(y)y\int \frac{ \bz}{h}\int U\bigg(\frac{ \wh{\bbeta}^{(2)T}\bz - x}{h}\bigg)f^{(\ell)}(x|y,\bz)dx f(y,\bz)d\bz\\
	=&\sum_{y \in \{-1,1\}}w(y)y\int \bz \int U(u)f^{(\ell)}(\wh{\bbeta}^{(2)T}\bz - hu|y,\bz)du f(y,\bz)d\bz,
	\eeq
	where the first step is by repeated integration by parts and the second step is by a change of variable. This implies that
	\beq
	&\norm{T^{(\ell)}(\wh{\bbeta}^{(2)}) - \tilde{T}_{h,U}^{(\ell)}(\wh{\bbeta}^{(2)})}_\infty\\
	=&\bignorm{\sum_{y \in \{-1,1\}}w(y)y\int \bz \int U(u)\bigg(
		f^{(\ell)}(\wh{\bbeta}^{(2)T}\bz|y,\bz) -
		f^{(\ell)}(\wh{\bbeta}^{(2)T}\bz-hu|y,\bz)
		\bigg)du f(y,\bz)d\bz}_\infty\\
	=&\cO(h^{\zeta}).
	\eeq
	Finally, similar to the proof of Theorem \ref{theorem_score}, we can show that the third term on the RHS of (\ref{eq_bias_pilot_0}) is $\cO_{\PP}(\sqrt{\frac{\log(d)}{nh^{2\ell + 1}}})$ based on Lemmas \ref{lemma_variance} and \ref{lemma_addi_conditional}. Putting all three terms together we have
	\beq
	\norm{T^{(\ell)}(\bbeta^*) - \wh{T}_{h,U}^{(\ell),n_{(1)}}(\wh{\bbeta}^{(2)})}_\infty = \cO_{\PP}\bigg(\sqrt{\frac{\log(d)}{nh^{2\ell+1}}} + (M_n\eta_1(n) \vee h)^\zeta\bigg),
	\eeq
	and thus $\norm{\wh{T}_{h,U}^{(\ell),n_{(1)}}(\wh{\bbeta}^{(2)})} = \cO_{\PP}(1)$ by the condition of this lemma. Plugging this back to (\ref{eq_bias_pilot_-1}) gives the desired result. This completes the proof.
\end{proof}

\begin{lemma}\label{prop_var_pilot}
	Under the same conditions of Theorem \ref{theorem_score}, if $L$ is a proper kernel of order $\ell$ satisfying Assumption \ref{ass_kernel}  and  $g^{\ell} + \sqrt{\frac{\log(d)}{ng}} + M_n\eta_1(n) = o(1)$, then we have \beq
	|\wh{\sigma}^2 - \bv^{*T}\bSigma^*\bv^* | \lesssim \norm{\bv^*}_1^2\bigg(
	\eta_2(n) + g^{\ell} + \sqrt{\frac{\log(d)}{ng}} + M_n\eta_1(n)
	\bigg).
	\eeq
	
\end{lemma}
\begin{proof}
	It suffices to show the convergence rate of $\tilde{\mu}_K\wh{\bv}^{(1)T}\wh{H}_{g,K}^{n_{(1)}}(\wh{\bbeta}^{(2)})\wh{\bv}^{(1)}$.
	Similar to the proof of Theorem \ref{theorem_score}, following Lemmas \ref{lemma_bias}, \ref{lemma_variance}, \ref{lemma_sample_diff} and \ref{lemma_addi_conditional}, we can show that
	\beq
	\norm{\wh{H}_{g,K}^{(1)}(\wh{\bbeta}^{(2)}) - H(\bbeta^*)}_{\max} = \cO_{\PP}\bigg(g^{\ell} + \sqrt{\frac{\log(d)}{ng}} + M_n\eta_1(n)\bigg).
	\eeq
	Since $\norm{H(\bbeta^*)}_{\max} =\cO(1)$, we know $\norm{\wh{H}_{g,K}^{(1)}(\wh{\bbeta}^{(2)})}_{\max} = \cO_{\PP}(1)$ when $g^{\ell} + \sqrt{\frac{\log(d)}{ng}} + M_n\eta_1(n) = o(1)$.
	Now by triangle inequality, we obtain
	\beq\label{eq_lemma_consistent_var_1}
	&|\wh{\bv}^{(1)T}\wh{H}_{g,K}^{(1)}(\wh{\bbeta}^{(2)})\wh{\bv}^{(1)} - \bv^{*T}H(\bbeta^*)\bv^*|\\ \leq& \norm{\wh{\bv}^{(1)} - \bv^*}^2_1\norm{\wh{H}_{g,K}^{(1)}(\wh{\bbeta}^{(2)})}_{\max} + 2\norm{\bv^{*T}\wh{H}_{g,K}^{(1)}(\wh{\bbeta}^{(2)})}_\infty \norm{\wh{\bv}^{(1)} - \bv^*}_1 + |\bv^{*T}(\wh{H}_{g,K}^{(1)}(\wh{\bbeta}^{(2)}) - H(\bbeta^*))\bv^*|.
	\eeq
	This implies that
	\beq
	|\wh{\sigma}^2 - \bv^{*T}\bSigma^*\bv^* | = \norm{\bv^*}_1^2\cO_{\PP}\bigg(
	\eta_2(n) + g^{\ell} + \sqrt{\frac{\log(d)}{ng}} + M_n\eta_1(n)
	\bigg).
	\eeq
	This completes the proof.
\end{proof}

\begin{lemma}[Lemma 6.1 of \citet{chernozhukov2018double}]\label{lemma_addi_conditional}
	Let $\{\bX_m\}$ and $\{\bY_m\}$ be sequences of random vectors. If $\norm{\bX_m} = \cO_{\PP}(A_m)$ conditional on $\{\bY_m\}$ for a sequence of positive constants $\{A_m\}$, then $\norm{\bX_m} = \cO_{\PP}(A_m)$  unconditionally.
\end{lemma}

\section{Theoretical Results and Discussions for the Bandwidth Selection}\label{app_lem_band}
\subsection{Key results}
In this section, we collect a few key results in the analysis of bandwidth selection procedures and defer the proofs to Section \ref{app_pf_band}.

\begin{lemma} \label{lemma_V_delta}
	Under Assumptions \ref{ass_smooth} - \ref{ass_estimators}, if $\sqrt{\frac{\log(n\vee d)}{n\delta}} \vee M_n\eta_1(n) = o(1)$, and $M_n^2\sqrt{\log(n\vee d)/n^{\epsilon_1}} = \cO(1)$, it holds that
	\beq
	|\wh{V}(\delta) - V(\delta)| \lesssim \frac{\psi_1(n,\delta)}{\delta},
	\eeq
	uniformly over all $\delta \in \Delta$, where
	\beq\label{eq_psi1}
	\psi_1(n,\delta) = \norm{\bv^*}_1^2\bigg(\eta_2(n) \vee \sqrt{\frac{\log(n\vee d)}{n\delta}} \vee M_n\eta_1(n) \bigg).
	\eeq
\end{lemma}

\begin{lemma}\label{prop_sb}
	Suppose Assumptions \ref{ass_smooth}- \ref{ass_sb_additional} hold. Choose $J$ as a proper kernel function of order $r$ satisfying the conditions of $K$ in Assumption \ref{ass_kernel}. In addition, assume $J$ is $\ell$ times continuously differentiable and $J^{(i)}$ degenerates at the boundary for $i = 1,\dotso,\ell-1$. If $M_n\leq C\sqrt{nb/\log(nd)}$ for some constant $C$ and $\frac{\log(n\vee d)}{nb^{2\ell+1}}=o(1)$, then it holds that
	\beq\label{eq_double_smoothing_rate}
	|\wh{SB}(\delta) - SB(\delta)|
	\lesssim \delta^{2\ell} \psi_2(n,\delta)
	\eeq
	uniformly over all $\delta\in\Delta$, where
	\beq\label{eq_psi2}
	\psi_2(n,\delta) = \norm{\bv^*}_1^2\bigg(\sqrt{\frac{\log(n\vee d)}{nb^{2\ell+1}}} \vee (\delta\vee b)^r \vee M_n \eta_1(n)(1 \vee \frac{M_n\eta_1(n)}{\delta^{\ell}}) \vee \eta_2(n)\bigg).
	\eeq
\end{lemma}

	In view of (\ref{eq_double_smoothing_rate}) and (\ref{eq_psi2}), the error terms $\sqrt{\frac{\log(n\vee d)}{nb^{2\ell+1}}}$ and $(\delta\vee b)^r$ correspond to the variance and bias when estimating $A(\bbeta^*,\delta)$ in (\ref{eq_adaptive_sq_ori}) with the extra smoothing step. This rate is similar to the case of estimating the $\ell$th derivative of a density function using a kernel of order $r$ in the context of kernel density estimation under the smoothness Assumption \ref{ass_sb_additional}. The terms $M_n \eta_1(n)(1 \vee \frac{M_n\eta_1(n)}{\delta^{\ell}})$ and $\eta_2(n)$ come from the plug-in error of $\wh{\bbeta}$ and $\wh{\bv}$, respectively.

Now we are ready to present the following theorem on the rate of $\hat\delta$.

\begin{theorem}\label{theorem_adapt_main}
	Under the conditions in Lemmas \ref{lemma_V_delta} and \ref{prop_sb}, if $\psi_1(n,\delta),\psi_2(n,\delta) = o(1)$, then with probability tending to 1,
	\[
	\frac{\wh{\delta} - \delta^*}{\delta^*} \lesssim \psi_1(n,\delta^*) \vee \psi_2(n,\delta^*).
	\]
	Thus, $\wh{\delta}/\delta^*\rightarrow 1$ in probability.
\end{theorem}
\begin{remark}
Notice that in Theorem \ref{theorem_adapt_main}, the MSE-optimal bandwidth satisfies $\delta^*\asymp n^{-1/(2\ell + 1)}$. To examine the order of $\psi_1(n,\delta^*) \vee \psi_2(n,\delta^*)$, we consider $d\gtrsim n$ (high-dimensional case) and optimize the bandwidth $b$ in $\psi_2(n,\delta^*)$, leading to $b\asymp (\log d/n)^{1/(2\ell + 2r + 1)}$.
Furthermore, take $\eta_1(n), \eta_2(n)$ as the rate from \citet{feng2022nonregular} and Lemma \ref{theorem_projection}, and assume that $\norm{\bv^*}_1, M_n = \cO(1)$, then $\psi_1(n,\delta^*) \vee \psi_2(n,\delta^*)$ can be simplified to
\beq\label{eq_adapt_rate}
\Big(\frac{\log d}{n}\Big)^{r/(2\ell+2r+1)}  \vee s^{(4\ell+1)/(2\ell + 1)}\Big(\frac{\log ^2d}{n}\Big)^{\ell/(2\ell+1)} \vee s'\Big(\frac{\log d}{n}\Big)^{(\ell-1)/(2\ell + 1)}.
\eeq
There are in general two scenarios: firstly, if $r$ is small relative to $\ell$ while $s,s'$ do not grow too fast, then $(\log d/n)^{r/(2\ell + 2r + 1)}$ from the extra smoothing step for estimating $SB(\delta)$ will be the dominant term.  However, if $r$ is relatively large such that $r/(2\ell + 2r + 1) \geq (\ell - 1)/(2\ell + 1)$, then the last two terms in (\ref{eq_adapt_rate}) due to $\wh\bbeta$ and $\wh \bomega$ dominate. In this case, even if $f(x|y,\bz)$ has a large amount of extra smoothness $r$, the convergence rate of $\hat\delta$ cannot be further improved, as the rate is dominated by the error from $\wh\bbeta$ and $\wh \bomega$.
\end{remark}

Recall that our goal is to show the asymptotic normality of $\hat{U}_{n}(\hat\delta)$. We need the following additional assumption on the kernel function.

\begin{assumption}\label{ass_kernell_2}
Assume $K^{\prime\prime}(\cdot),K^{\prime\prime\prime}(\cdot)< C_1$ for some constant $C_1$, $K^\prime$ degenerates at the boundaries, and
\[
\int |K(u)|du,\int |u| |K(u)|du,
\]
\[
\int |u|^p |K^\prime(u)| du, \forall p\in \{0,1,2,3\},
\]
\[
 \int |u|^q |K^{\prime\prime}(u)| du, \forall q\in \{2,3,4,5\},
\]
\[
\int |u|^t |K^{\prime\prime \prime}(u)|du, \forall t\in \{1,2,3\},
\]
are all bounded by some constant.
\end{assumption}


Similar to Section \ref{sec_details}, we need to use the cross-fitting approach to estimate $\delta$. Specifically, we split the data into 3 equal folds. Following the approach described in Section \ref{sec_details}, we construct $\wh{S}^{(1)}_\delta(\theta,\wh{\bgamma}^{(2)})$ and $\wh{S}^{(2)}_\delta(\theta,\wh{\bgamma}^{(1)})$ using the first and second folds. We further estimate the variance and bias in the same way as in Section \ref{adaptivity}, and denote the estimators by $\hat V^{(1,2)}(\delta)$ and $\hat B^{(1,2)}(\delta)$, where the superscript $(1,2)$ refers to the fact that we only use  the data in the first and second folds. Then we estimate $\delta$ by
$$
\wh\delta^{(1,2)} = \argmin_{\delta} \Big[\frac{1}{n}\wh{V}^{(1,2)}(\delta) + \frac{n-1}{n}(\wh{B}^{(1,2)}(\delta))^2\Big].
$$
Finally, we plug-in $\wh\delta^{(1,2)}$ into the third fold of the data to construct  the bias corrected smoothed decorrelated score statistic
$$
\wh{U}_n^{(3)}(\wh\delta^{(1,2)}) = \sqrt{n_3\wh\delta^{(1,2)}}\Big(\frac{\wh{S}^{(3)}_{\wh\delta^{(1,2)}}(0,\wh{\bgamma}^{(1,2)}) - (\wh\delta^{(1,2)})^{\ell}\wh{\mu}^{(1,2)}}{\wh{\sigma}^{(1,2)}}\Big),
$$
where $n_3$ is the sample size for the third fold, $\wh{\mu}^{(1,2)}$ and $\wh{\sigma}^{(1,2)}$ are the cross-fitted estimators in (\ref{eq_cross_mu}), $\wh{\bgamma}^{(1,2)}=(\wh{\bgamma}^{(1)}+\wh{\bgamma}^{(2)})/2$, and
$$
\wh{S}^{(3)}_\delta(\theta,\bgamma)=\nabla_{\theta} R^{(3)}_\delta(\theta,\bgamma) - \wh{\bomega}^{(1,2)T}\nabla_{\bgamma} R^{(3)}_\delta(\theta, \bgamma),
$$
in which we use the data in the third fold to compute the gradient $\nabla_{\theta} R^{(3)}_\delta(\theta,\bgamma)$ and $\nabla_{\bgamma} R^{(3)}_\delta(\theta, \bgamma)$. Similarly, we can construct $\wh{U}_n^{(1)}(\wh\delta^{(2,3)})$ and $\wh{U}_n^{(2)}(\wh\delta^{(1,3)})$. The following theorem indeed implies that $\wh{U}_n^{(3)}(\wh\delta^{(1,2)})\stackrel{d}{\rightarrow} N(0,1)$, and $\wh{U}_n^{(1)}(\wh\delta^{(2,3)})$, $\wh{U}_n^{(2)}(\wh\delta^{(1,3)})$ and $\wh{U}_n^{(3)}(\wh\delta^{(1,2)})$ are asymptotically independent. Thus, the final test statistic  given by $[\wh{U}_n^{(1)}(\wh\delta^{(2,3)})+\wh{U}_n^{(2)}(\wh\delta^{(1,3)})+\wh{U}_n^{(3)}(\wh\delta^{(1,2)})]/\sqrt{3}$ is asymptotically $N(0,1)$.

\begin{theorem}\label{thm1_11}
Under the conditions in Lemmas \ref{lemma_V_delta} and \ref{prop_sb}, Theorem \ref{theorem_score} and Assumption \ref{ass_kernell_2}, we further assume $\max_{1\leq i\leq n}|X_i| \leq M_n$ and $\|\bbeta^*\|_1$ is bounded by a constant.
If
\begin{equation}\label{thmasp0}
    \frac{M_n^3}{\delta^{* 2}}\sqrt{\frac{\log d}{n\delta^*}}= \mathcal{O}(1),
\end{equation}
and
\begin{equation}\label{thmasp2}
    (n \delta^*)^{1 / 2}\left\|\boldsymbol{v}^{*}\right\|_{1}
C_{n,\delta^*}
\left(
C_{n,\delta^*} \vee \frac{M_{n}^{2} \eta_{1}(n)^{2}}{\delta^*}
\vee \sqrt{\frac{\log d}{n\delta^*}}
\vee \delta^{* (\ell-1)}\eta_1(n) \vee
\delta^{* \ell} \eta_2(n)
\right) = o(1),
\end{equation}
where $\psi_{1}\left(n, \delta^{*}\right) \vee \psi_{2}\left(n, \delta^{*}\right)\coloneqq C_{n,\delta^*}=o(1)$,  then under $H_{0}: \theta^{*}=0$, it holds that
\[
\wh{U}_n^{(3)}(\wh\delta^{(1,2)}) \stackrel{d}{\rightarrow} N(0,1).
\]
\end{theorem}

This is a more complete version of Theorem \ref{theorem_adapt_main_short} in the main paper. In the next subsection, we provide the proof of Theorem \ref{thm1_11}.

\begin{remark}
As remarked after Theorem \ref{theorem_adapt_main}, take $\eta_1(n), \eta_2(n)$ as the rate from \cite{feng2022nonregular} and Lemma \ref{theorem_projection},  and assume that $\left\|\boldsymbol{v}^{*}\right\|_{1}, M_{n}=\mathcal{O}(1)$, then we have
\begin{equation}\label{rate1}
    C_{n,\delta^*} \asymp \left(\frac{\log d}{n}\right)^{r /(2 \ell+2 r+1)} \vee s^{(4 \ell+1) /(2 \ell+1)}\left(\frac{\log ^{2} d}{n}\right)^{\ell /(2 \ell+1)} \vee s^{\prime}\left(\frac{\log d}{n}\right)^{(\ell-1) /(2 \ell+1)}.
\end{equation}
Since the optimal bandwidth has $\delta^* \asymp n^{-1/(2\ell+1)}$, (\ref{thmasp0}) is satisfied if we have $\ell \geq 3$ and $\log d\ll n^{\frac{2(\ell-2)}{2\ell+1}}$. Given (\ref{rate1}), with some calculation we can show that (\ref{thmasp2}) reduces to
\[
(n \delta^*)^{1 / 2}
C_{n,\delta^*}^2 = o(1).
\]
Followed by the remark after Theorem \ref{theorem_adapt_main}, if the first term in (\ref{rate1}) dominates, to satisfy the condition above we need $r > \frac{\ell(2\ell+1)}{2(\ell+1)}$ (ignoring the $\log d$ factor for simplicity). If the last two terms in (\ref{rate1}) dominate  and assume $s,s'=O(1)$, then $\ell \geq 3$ is sufficient for the condition above.
\end{remark}

\subsection{Proofs}\label{app_pf_band}	

\subsubsection{Proof of Lemma \ref{lemma_V_delta}}\label{pf_lemma_V_delta}

	It suffices to show the rate for $\wh{\bv}^{(1)T}\wh{\bGamma}^{(1)}(\delta)\wh{\bv}^{(1)}$.
	Let $\tilde{\bGamma}(\delta) = \EE\bigg[
	\nabla \bar{R}_\delta^1(\wh{\bbeta}^{(2)})\nabla \bar{R}_\delta^1(\wh{\bbeta}^{(2)})^T
	\bigg]$, and $\bGamma(\delta)=\EE[\nabla \bar{R}^1_\delta(\bbeta^*)\nabla \bar{R}^1_\delta(\bbeta^*)^T]$. We firstly look at
	\beq\label{eq_var_adapt_2}
	\norm{\wh{\bGamma}^{(1)}(\delta) - \bGamma(\delta)}_{\max} \leq
	\norm{\wh{\bGamma}^{(1)}(\delta) - \tilde{\bGamma}(\delta)}_{\max} + \norm{\tilde{\bGamma}(\delta) - \bGamma(\delta)}_{\max}.
	\eeq
	Now we bound the first term. For any $\delta \in \Delta$, we can always find a $\wt\delta \in \wt\Delta$ such that $|\delta - \wt\delta| \lesssim n^{-\rho}$, where $\rho$ is some constant that can be arbitrarily large and the set $\wt\Delta$ that is a subset of $\Delta$ has cardinality of order $\cO(n^{\rho})$. For such pair $(\delta,\wt\delta)$, we can show that
	\beq \label{eq_var_adapt_1}
	&\sqrt{n\delta^3}\norm{\wh{\bGamma}^{(1)}(\delta) - \tilde{\bGamma}(\delta)}_{\max}\\ \lesssim& \sqrt{n\wt\delta^3}\bigg (\norm{\wh{\bGamma}^{(1)}(\wt\delta) - \tilde{\bGamma}(\wt\delta)}_{\max} + \norm{\wh{\bGamma}^{(1)}(\delta) - \wh{\bGamma}^{(1)}(\tilde\delta)}_{\max} + \norm{\tilde{\bGamma}(\wt\delta) - \tilde{\bGamma}(\delta)}_{\max}\bigg)\\
	\lesssim& \sqrt{n\wt\delta^3}\norm{\wh{\bGamma}^{(1)}(\wt\delta) - \tilde{\bGamma}(\wt\delta)}_{\max} + \cO(n^{-1/2}),
	\eeq
	where the first inequality is by taking $\rho$ large enough.
	To see last step of above, recall that under Assumptions \ref{ass_kernel}, \ref{ass_proportion} and \ref{ass_moment}
	\beq
	\norm{\wh{\bGamma}^{(1)}(\delta) - \wh{\bGamma}^{(1)}(\wt\delta)}_{\max}=& \max_{j,k}\bigg |\frac{1}{|\cN_1|}\sum_{i\in \cN_1}
	w(y_i)^2\bz_{ij}\bz_{ik} \bigg[\frac{K^2((x_i - \wh\bbeta^{(2)T}\bz_i)/\delta)}{\delta^2}
	- \frac{K^2((x_i - \wh\bbeta^{(2)T}\bz_i)/\wt\delta)}{\wt\delta^2}
	\bigg]\bigg|\\
	\lesssim& M_n^2(\frac{1}{\delta^2} - \frac{1}{\wt\delta^2}).
	\eeq
Thus, taking $\rho$ large enough will ensure $\sqrt{n\wt\delta^3}(1+o(1))\norm{\wh{\bGamma}^{(1)}(\delta) - \wh{\bGamma}^{(1)}(\wt\delta)}_{\max} = \cO(n^{-1/2})$. With a similar derivation we will also obtain $\sqrt{n\wt\delta^3}(1+o(1))\norm{\wt{\bGamma}(\delta) - \wt{\bGamma}(\wt\delta)}_{\max} = \cO(n^{-1/2})$, and thus (\ref{eq_var_adapt_1}) holds.
	
	Now we start to bound $\max_{\wt\delta\in\wt\Delta}\sqrt{n\wt\delta^3}\norm{\wh{\bGamma}^{(1)}(\wt\delta) - \tilde{\bGamma}(\wt\delta)}_{\max}$.
    Since
	$$|\wh{\bGamma}^{(1)}_{ijk}(\wt\delta)| = |(\nabla \bar{R}_{\wt\delta}^i(\wh{\bbeta}^{(2)})\nabla \bar{R}_{\wt\delta}^i(\wh{\bbeta}^{(2)})^T)_{jk}| \lesssim \frac{M_n^2}{\wt\delta^2}$$ and $$Var[(\nabla \bar{R}_{\wt\delta}^i(\wh{\bbeta}^{(2)})\nabla \bar{R}_{\wt\delta}^i(\wh{\bbeta}^{(2)})^T)_{jk}] \lesssim \frac{1}{\wt\delta^3},$$
	by applying Bernstein inequality similar to the proof of Lemma \ref{lemma_variance} , conditioned on $\wh\bbeta^{(2)}$ we can show that
	\beq
	&\PP_{\wh\bbeta^{(2)}}(\max_{\wt\delta\in \wt\Delta}\sqrt{n\wt\delta^3}\norm{\wh{\bGamma}^{(1)}(\wt\delta) - \tilde{\bGamma}(\wt\delta)}_{\max} > t)\\
	\leq&\sum_{\wt\delta\in\wt\Delta}\sum_{j = 1}^d\sum_{k = 1}^d\PP_{\wh\bbeta^{(2)}}(\sqrt{n\wt\delta^3}|\wh{\bGamma}^{(1)}(\wt\delta) - \tilde{\bGamma}(\wt\delta)|_{jk} > t)\\
	\lesssim&n^{\rho}d^2\exp\bigg(-\frac{\frac{1}{2}t^2/\wt\delta^3}{\frac{c_0}{\wt\delta^3} + \frac{c1M_n^2t}{3\wt\delta^3\sqrt{n\wt\delta}} } \bigg),
	\eeq
	where $c_0,c_1>0$ are some constants. By taking $t = c_2\sqrt{\log(nd)}$ for some constant $c_2$ large enough, we obtain that conditioned on $\wh{\bbeta}^{(2)}$, with probability greater than $1 - \cO((nd)^{-1})$
	\beq
	\norm{\wh{\bGamma}^{(1)}(\wt\delta) - \tilde{\bGamma}(\wt\delta)}_{\max}\lesssim \sqrt{\frac{\log(nd)}{n\wt\delta^3}}
	\eeq
	uniformly over $\wt\delta\in\wt\Delta$.
	This combines with Lemma \ref{lemma_addi_conditional} further implies that with probability approaching to 1
	\beq
	\norm{\wh{\bGamma}^{(1)}(\delta) - \tilde{\bGamma}(\delta)}_{\max} \lesssim \sqrt{\frac{\log(nd)}{n\delta^3}}
	\eeq
	uniformly over $\delta\in\Delta$.
	
	For the second term on the RHS of (\ref{eq_var_adapt_2}), similar to Lemma \ref{lemma_sample_diff}, we can show that
	\beq
	\norm{\tilde{\bGamma}(\delta) - \bGamma(\delta)}_{\max} \lesssim \frac{M_n\eta_1(n)}{\delta}.
	\eeq
	This implies that
	\beq
	\norm{\wh{\bGamma}(\delta) - \bGamma(\delta)}_{\max} = \frac1\delta \cO_{\PP}(\sqrt{\frac{\log(nd)}{n\delta}} \vee M_n\eta_1(n)).
	\eeq
	Meanwhile, since
	\beq
	\norm{\delta \cdot \bGamma(\delta)}_{\max} =& \max_{i,j}\bigg|\sum_{y}w(y)^2\int \frac{z_iz_j}{\delta} \int K^2(\frac{y(x - \bbeta^{*T}\bz)}{\delta})f(x|y,\bz)dxf(y,\bz)d\bz
	\bigg|\\
	=&\cO(1),
	\eeq
	under the condition that $\sqrt{\frac{\log(nd)}{n\delta}} \vee M_n\eta_1(n) = o(1)$, we know $\norm{\delta \wh{\bGamma}(\delta)}_\infty = \cO_{\PP}(1)$.
	Finally, by triangle inequality
	\beq
	&|\wh{\bv}^{(1)T}\wh{\bGamma}^{(1)}(\delta)\wh{\bv}^{(1)} - \bv^{*T}{\bGamma}(\delta)\bv^*|\\
	\leq& \norm{\wh{\bv}^{(1)} - \bv^*}_1\norm{\wh{\bGamma}^{(1)}(\delta)(\wh{\bv}^{(1)} - \bv)}_\infty + 2\norm{\bv^{*T}\wh{\bGamma}^{(1)}(\delta)}_\infty \norm{\wh{\bv}^{(1)} - \bv^*}_1 + |\bv^{*T}(\wh{\bGamma}^{(1)}(\delta) - {\bGamma}(\delta))\bv^*|,
	\eeq
	which further implies that with probability approaching 1
	\[
	|\wh{\bv}^{(1)T}\wh{\bGamma}^{(1)}(\delta)\wh{\bv}^{(1)} - \bv^{*T}{\bGamma}(\delta)\bv^*| \lesssim  \frac{\norm{\bv^*}_1^2}{\delta}\bigg(\eta_2(n) + \sqrt{\frac{\log(nd)}{n\delta}} + M_n\eta_1(n) \bigg)
	\]
	uniformly over $\delta\in\Delta$. This completes the proof.

\subsubsection{Proof of Lemma \ref{prop_sb}}\label{pf_prop_sb}

	It suffices to bound $\bigg |
	\bigg (\wh{\bv}^{(1)T} \frac{1}{|\cN_1|}\sum_{i \in \cN_1}A_i(\wh{\bbeta}^{(2)},\delta)\bigg )^2 - SB(\delta)
	\bigg |$. Define
	\beq
	\wt A(\bbeta,\delta) =& \EE_{\bbeta}(A_1(\bbeta,\delta))\\
	=&\sum_{y}w(y)y\int_{\bz}\bz\int_t J(t)\int_uK(u)[f(u\delta+tb+\bbeta^{T}\bz|y,\bz) - f(tb+ \bbeta^T\bz|y,\bz)]dudtf(y,\bz)d\bz.
	\eeq By definition,
	\beq
	&\bigg |
	\bigg (\wh{\bv}^{(1)T} \frac{1}{|\cN_1|}\sum_{i \in \cN_1}A_i(\wh{\bbeta}^{(2)},\delta)\bigg )^2 - SB(\delta)
	\bigg |\\
	=& \bigg|
	\bigg (\wh{\bv}^{(1)T} \frac{1}{|\cN_1|}\sum_{i \in \cN_1}A_i(\wh{\bbeta}^{(2)},\delta)\bigg )^2 - \big(\bv^{*T}A(\bbeta^*,\delta)\big)^2
	\bigg|\\
	\leq& \bigg|
	\bigg (\wh{\bv}^{(1)T} \frac{1}{|\cN_1|}\sum_{i \in \cN_1}A_i(\wh{\bbeta}^{(2)},\delta)\bigg )^2 - \big(\wh{\bv}^{(1)T}A(\wh{\bbeta}^{(2)},\delta)\big)^2
	\bigg| + \bigg|
	\big(\wh{\bv}^{(1)T}A(\wh{\bbeta}^{(2)},\delta)\big)^2
	- \big(\bv^{*T}A(\bbeta^*,\delta)\big)^2
	\bigg|\\
	\leq&\underbrace{\bigg|
		\bigg(\wh{\bv}^{(1)T}\big(\frac{1}{|\cN_1|}\sum_{i\in \cN_1}A_i(\wh{\bbeta}^{(2)},\delta) - \wt A(\hat{\bbeta}^{(2)},\delta)\big)\bigg)^2
		\bigg|}_{I_1}\\
	&~~~~~~~~ + 2\underbrace{\bigg|
		\bigg(
		\wh{\bv}^{(1)T}\wt A(\hat{\bbeta}^{(2)},\delta)\cdot \wh{\bv}^{(1)T}\big(\frac{1}{|\cN_1|}\sum_{i\in \cN_1}A_i(\wh{\bbeta}^{(2)},\delta) - \wt A(\hat{\bbeta}^{(2)},\delta)\big)
		\bigg|}_{I_2}\\
	&~~~~~~~~ + \underbrace{\bigg|
		\big(\wh{\bv}^{(1)T}\wt A(\hat{\bbeta}^{(2)},\delta)\big)^2 - \big(\wh{\bv}^{(1)T}A(\wh{\bbeta}^{(2)},\delta)\big)^2
		\bigg|}_{I_3} + \underbrace{\bigg|
		\big(\wh{\bv}^{(1)T}A(\wh{\bbeta}^{(2)},\delta)\big)^2
		- \big(\bv^{*T}A(\bbeta^*,\delta)\big)^2
		\bigg|}_{I_4}.
	\eeq
	Now we start to bound each term.
	\begin{itemize}
		\item For $I_1$,
		we firstly study $\delta^{-\ell}\norm{\frac{1}{|\cN_1|}\sum_{i\in \cN_1}A_i(\wh{\bbeta}^{(2)},\delta) - \wt A(\hat{\bbeta}^{(2)},\delta)}_\infty$ condition on $\wh{\bbeta}^{(2)}$.
		Similar to the proof of Lemma \ref{prop_var_pilot}, it suffices to bound
		\beq
		\wt\delta^{-\ell}\norm{\frac{1}{|\cN_1|}\sum_{i\in \cN_1}A_i(\wh{\bbeta}^{(2)},\wt\delta) - \wt A(\hat{\bbeta}^{(2)},\wt\delta)}_\infty
		\eeq
		uniformly over $\wt\delta \in \wt\Delta$, where $\wt\Delta$ is a subset of $\Delta$ with cardinality of order $\cO(n^\rho)$ for some constant $\rho >0$, such that for each $\delta\in\Delta$, there exist $\wt\delta\in\wt\Delta$ with $|\delta - \wt\delta| \lesssim n^{-\rho}$.
		
		Let $T_{ij}(\wt\delta) = \wt\delta^{-\ell}\big( A_i(\wh{\bbeta}^{(2)},\wt\delta) - \wt A(\hat{\bbeta}^{(2)},\wt\delta) \big)_j$, we know $\EE_{\wh{\bbeta}^{(2)}} T_{ij}(\wt\delta) = 0$. Moreover for all $j$,
		
		\beq
		\wt\delta^{-\ell}|(A_{i}(\wh{\bbeta}^{(2)},\wt\delta))_j| =& \wt\delta^{-\ell}\bigg |
		w(y_i)\frac{z_{ij}y_i}{b}\int K(u) \big[
		J(\frac{x_i - \wh{\bbeta}^{(2)T}\bz_i - u\wt\delta}{b}) - J(\frac{x_i - \wh{\bbeta}^{(2)T}\bz_i}{b})
		\big]
		\biggr\} du
		\bigg | \\
		=& \wt\delta^{-\ell}\bigg |
		w(y_i)\frac{z_{ij}y_i}{b}\int K(u) \frac{\wt\delta^\ell}{b^{\ell}}\frac{u^\ell}{\ell!}
		J^{\ell}(\frac{x_i - \wh{\bbeta}^{(2)T}\bz_i - \tau u\wt\delta}{b})
		\big]
		\biggr\} du
		\bigg | \\
		\lesssim & \frac{M_n}{b^{\ell+1}},
		\eeq
		where $\tau \in [0,1]$ and the second equality is because $K$ is of order $\ell$ and mean value theorem. This implies that $|T_{ij}(\wt\delta)|\lesssim \frac{M_n}{b^{\ell+1}}$. Now we look at $\EE_{\wh{\bbeta}^{(2)}}[T^2_{ij}(\wt\delta)]$.
		
		Notice that for all $j$
		\beq\label{eq_sb_1}
	\wt A(\hat{\bbeta}^{(2)},\wt\delta)_j =& \sum_{y}w(y)\int_{\bz}\bz_j y\int_t J(t)\int_uK(u)[f(u\wt\delta+tb+\wh{\bbeta}^{(2)T}\bz|y,\bz) - f(tb+\wh{\bbeta}^{(2)T}\bz|y,\bz)]\\
& dudtf(y,\bz)d\bz\\
		=&\sum_{y}w(y)\int_{\bz}\bz_j y\int_t J(t)\int_uK(u)\frac{(u\wt\delta)^{\ell}}{\ell!} f^{(\ell)}(\tau u\wt\delta+tb+\wh{\bbeta}^{(2)T}\bz|y,\bz)dudtf(y,\bz) d\bz\\
		\lesssim&\wt\delta^{\ell},
		\eeq
		where the second equality is because $K$ is of order $\ell$ and $\tau \in [0,1]$.
		For the second moment, we have
		\beq
		&\EE_{\wh{\bbeta}^{(2)}}[(\wt\delta^{-\ell}A_{1}(\wh{\bbeta}^{(2)},\wt\delta))_j^2]\\
		=&\wt\delta^{-2\ell}\EE_{\wh{\bbeta}^{(2)}} \bigg[
		w(y)^2 \frac{z_j^2}{b^2}\bigg(\int K(u) \big[
		J(\frac{x - \wh{\bbeta}^{(2)T}\bz - u\wt\delta}{b}) - J(\frac{x - \wh{\bbeta}^{(2)T}\bz}{b})
		\big]
		\biggr\} du\bigg)^2
		\bigg] \\
		=&\wt\delta^{-2\ell}\EE_{\wh{\bbeta}^{(2)}} \bigg[
		w(y)^2 \frac{z_j^2}{b^2}\bigg(
		\int K(u) \frac{\wt\delta^\ell}{b^{\ell}}\frac{u^\ell}{\ell!}
		J^{\ell}(\frac{x - \wh{\bbeta}^{(2)T}\bz - \tau u\wt\delta}{b})du
		\bigg)^2
		\bigg] \\
		=&\frac{1}{b^{2\ell+1}} \bigg |   \sum_{y}\int_z w(y)^2z_j^2\int_t \bigg(
		\int K(u) \frac{u^\ell}{\ell!}
		J^{\ell}(t - \frac{\tau u\wt\delta}{b})d u
		\bigg)^2f(tb + \wh{\bbeta}^{(2)T}\bz|y,\bz)dt f(y,\bz) d\bz \bigg | \\
		\lesssim & \frac{1}{b^{2\ell+1}},
		\eeq
		where the last equality is due to a change of variable $t = \frac{x - \wh{\bbeta}^{(2)T}\bz}{b}$.
		Since $b\rightarrow0$, we conclude that $\EE_{\wh{\bbeta}^{(2)}}[T^2_{ij}(\wt\delta)]\lesssim \frac{1}{b^{2\ell+1}}.$ Now applying Bernstein inequality similar to Lemma \ref{prop_var_pilot} with $M_n\sqrt{\frac{\log(d)}{nb}} = \cO(1)$, second moment of order $\frac{1}{b^{2\ell+1}}$ and each term is bounded by $\frac{M_n}{b^{\ell+1}}$, we will obtain that with probability greater than $1 - \cO((nd)^{-1})$, conditioned on $\wh{\bbeta}^{(2)}$
		\beq
		\norm{\frac{1}{|\cN_1|}\sum_{i\in \cN_1}A_i(\wh{\bbeta}^{(2)},\wt\delta) - \EE A_i(\wh{\bbeta}^{(2)},\wt\delta)}_\infty\lesssim \frac{\wt\delta^{\ell}}{b^{\ell}}\sqrt{\frac{\log(nd)}{nb}},
		\eeq
		which together with Lemma \ref{lemma_addi_conditional} implies that with probability approaching 1
		\beq\label{eq_sb_2}
		\norm{\frac{1}{|\cN_1|}\sum_{i\in \cN_1}A_i(\wh{\bbeta}^{(2)},\delta) - \EE_{\wh{\bbeta}^{(2)}} A_i(\wh{\bbeta}^{(2)},\delta)}_\infty\lesssim \frac{\delta^{\ell}}{b^{\ell}}\sqrt{\frac{\log(nd)}{nb}}
		\eeq
		holds uniformly for all $\delta\in\Delta$.
		
		Therefore we conclude that
		\beq
		I_1 = \cO_{\PP}(\norm{\bv^*}_1^2\frac{\delta^{2\ell}}{b^{2\ell}}\frac{\log(nd)}{nb}).
		\eeq
		
		\item For $I_2$, (\ref{eq_sb_1}) implies that
		$\norm{\EE_{\wh{\bbeta}^{(2)}} A_1(\wh{\bbeta}^{(2)},\delta)}_{\infty} \lesssim \delta^{\ell}$ regardless the choice of $\wh{\bbeta}^{(2)}$. This combined with (\ref{eq_sb_2}) further implies that
		\beq
		I_2 = \cO_{\PP}(\norm{\bv^*}_1^2 \delta^{\ell}\sqrt{\frac{\delta^{2\ell}}{b^{2\ell}}\frac{\log(nd)}{nb}}) = \cO_{\PP}(
		\norm{\bv^*}_1^2
		\frac{\delta^{2\ell}}{b^{\ell}}\sqrt{\frac{\log(nd)}{nb}}).\eeq
		
		\item For $I_3$, by definition for all $j$ we can write
		\beq
		&\big ( \wt A(\wh{\bbeta}^{(2)},\delta) - A(\wh{\bbeta}^{(2)},\delta)\big)_j\\
		=& \sum_{y}w(y)\int_{\bz}\bz_j y\int_t J(t)\int_uK(u)[(f(u\delta+tb+\wh{\bbeta}^{(2)T}\bz|y,\bz) - f(tb+\wh{\bbeta}^{(2)T}\bz|y,\bz))\\
		&~~~~~~~~ - (f(u\delta + \wh{\bbeta}^{(2)T}\bz|y,\bz) - f(\wh{\bbeta}^{(2)T}\bz|y,\bz)) ]dudtf(y,\bz)d\bz\\
		=&\sum_{y}w(y)\int_{\bz}\bz_j y\int_u\sum_{j = \ell}^{\ell + r -1}
		\frac{(u\delta)^j}{j!}K(u)\int_t J(t) [  f^{(j)}(tb+\wh{\bbeta}^{(2)T}\bz|y,\bz) - f^{(j)}(\wh{\bbeta}^{(2)T}\bz|y,\bz)]\\
& dtduf(y,\bz)d\bz + \cO(\delta^{\ell + r})\\
		=& \sum_{y}w(y)\int_{\bz}\bz_j y\int_u\sum_{j = \ell}^{\ell + r - 1}
		\frac{(u\delta)^j}{j!}K(u)\int_t J(t) \frac{(tb)^{\ell + r - j}}{(\ell + r - j)!}  f^{(\ell + r)}(\tau tb+\wh{\bbeta}^{(2)T}\bz|y,\bz) \\
& dtduf(y,\bz)d\bz + \cO(\delta^{\ell + r})\\
		=&\cO(\delta^{\ell}(\delta \vee b)^{r}),
		\eeq
		where $\tau \in [0,1]$.
		Here the first equality is by definition, the second equality uses the property that $\int K(u)u^i du = 0~\forall~i < \ell$, and the third equality uses the property that $\int J(u)u^i du = 0~\forall ~ i < r$ and the mean value theorem. This implies that
		\beq
		\norm{\wt A(\wh{\bbeta}^{(2)},\delta) - A(\wh{\bbeta}^{(2)},\delta)}_\infty = \cO(\delta^{\ell}(\delta \vee b)^{r}).
		\eeq
		Meanwhile, by (\ref{eq_sb_1}) and a similar derivation, we know both $\norm{\wt A(\wh{\bbeta}^{(2)},\delta)}_{\infty}$ and $\norm{A(\wh{\bbeta}^{(2)},\delta)}_{\infty}$ are of order $\delta^{\ell}$. These together imply that
		\beq
		I_3 = \cO(\norm{\bv^*}_1^2\delta^{2\ell}(\delta \vee b)^{r}).
		\eeq
		
		\item For $I_4$, again $\norm{A(\wh{\bbeta}^{(2)},\delta)}_\infty\asymp \norm{A(\bbeta^*,\delta)}_\infty = \cO(\delta^\ell)$. For the difference, we have
		\beq
		&\norm{A(\wh{\bbeta}^{(2)},\delta) - A(\bbeta^*,\delta)}_\infty\\
		=&\max_j
		\bigg |
		\int_u K(u) \bigg[
		\sum_y w(y)\int_{\bz} z_j y \big(
		f(u\delta+ \wh{\bbeta}^{(2)T}\bz|y,\bz) - f(\wh{\bbeta}^{(2)T}\bz|y,\bz)\\
		&~~~~~~~~~~~~~~~~~~~~~~~~~~~~ - f(u\delta+ \bbeta^{*T}\bz|y,\bz) + f(\bbeta^{*T}\bz|y,\bz)
		\big) f(y,\bz) d\bz
		\bigg] du \bigg |\\
		=&\max_j
		\bigg | \sum_{y} w(y)\int_{\bz}z_j y \int_{u} K(u)\bigg[
		\big(
		f(u\delta+ \wh{\bbeta}^{(2)T}\bz|y,\bz) - f(u\delta+ \bbeta^{*T}\bz|y,\bz)
		\big)\\
		&~~~~~~~~~~~~~~~~~~~~~~~~~~~~ - \big(
		f(\wh{\bbeta}^{(2)T}\bz|y,\bz) - f(\bbeta^{*T}\bz|y,\bz)
		\big)
		\bigg] f(y,\bz) d\bz du\bigg | \\
		=&\max_j
		\bigg | \sum_{y} w(y)\int_{\bz}z_j y \int_{u} K(u)\bigg[
		(\betadiff)^T\bz\big(
		f'(u\delta+ \bbeta^{*T}\bz|y,\bz)
		- f'(\bbeta^{*T}\bz|y,\bz)\big) \\
		&~~~~~~~~~~~~~~~~~~~~~~~~~~~~~  + ((\betadiff)^T\bz)^2
		\big (
		f''(u\delta + \tilde{\bbeta}^T\bz) - f''( \breve{\bbeta}^T\bz)
		\big )
		\bigg]duf(y,\bz)d\bz \bigg | \\
		\lesssim&M_n\norm{\betadiff}_1\delta^{\ell} + M_n^2\norm{\betadiff}_1^2,
		\eeq
		where the last step follows from an $\ell$ order Taylor expansion of $f'(u\delta+ \bbeta^{*T}\bz|y,\bz)
		- f'(\bbeta^{*T}\bz|y,\bz)$ for the first term and boundedness of $f''$ for the second term.
		This implies that
		\beq
		\norm{A(\wh{\bbeta}^{(2)},\delta) - A(\bbeta^*,\delta)}_\infty \lesssim M_n\eta_1(n)(\delta^{\ell} + M_n\eta_1(n))
		\eeq
		with probability approaching 1 by Assumption \ref{ass_estimators}. Therefore, we can obtain
		\beq
		I_4 \leq& \bigg|
		(\bv^{(1)T}A(\wh{\bbeta}^{(2)},\delta) - \bv^{*T}A(\bbeta^*,\delta))(\bv^{(1)T}A(\wh{\bbeta}^{(2)},\delta) + \bv^{*T}A(\bbeta^*,\delta))
		\bigg|\\
		\leq&\bigg|
		\bv^{(1)T}(A(\wh{\bbeta}^{(2)},\delta) - A(\bbeta^*,\delta)) + (\bv^{(1)T}- \bv^{*T})A(\bbeta^*,\delta)
		\bigg|\bigg|
		\bv^{(1)T}A(\wh{\bbeta}^{(2)},\delta) + \bv^{*T}A(\bbeta^*,\delta)
		\bigg|\\
		=&\norm{\bv^*}_1^2\delta^{2\ell}\cO_{\PP}(M_n\eta_1(n)(1 \vee M_n\eta_1(n)/\delta^{\ell}) + \eta_2(n)).
		\eeq
	\end{itemize}
	Combining all bounds for $I_1,\dotso,I_4$, we know that uniformly over all $\delta\in\Delta$
	\beq
	&\bigg |
	\bigg (\wh{\bv}^{(1)T} \frac{1}{|\cN_1|}\sum_{i \in \cN_1}A_i(\wh{\bbeta}^{(2)},\delta)\bigg )^2 - SB(\delta)
	\bigg |\\
	\lesssim&\norm{\bv^*}_1^2\delta^{2\ell} \bigg(\frac{\log(nd)}{nb^{2\ell+1}} + \sqrt{\frac{\log(nd)}{nb^{2\ell+1}}} + (\delta\vee b)^r + M_n \eta_1(n)(1 \vee \frac{M_n\eta_1(n)}{\delta^{\ell}}) + \eta_2(n)\bigg).
	\eeq
	This completes the proof.

\subsubsection{Proof of Theorem \ref{theorem_adapt_main}}\label{pf_theorem_adapt}
	By Lemma \ref{lemma_V_delta} and Lemma \ref{prop_sb}, uniformly over all $\delta\in\Delta$
	\beq
	&|\wh{M}(\delta) - M(\delta)|\\
	=&\norm{\bv^*}_1^2\frac{1}{n\delta}\cO_{\PP}\bigg(\eta_2(n) + \sqrt{\frac{\log(nd)}{n\delta}} + M_n\eta_1(n) \bigg) \\
	&~~~~~~ + \norm{\bv^*}_1^2\delta^{2\ell} \cO_{\PP}\bigg(\frac{\log(nd)}{nb^{2\ell+1}} + \sqrt{\frac{\log(nd)}{nb^{2\ell+1}}} + (\delta\vee b)^r + M_n \eta_1(n)(1 \vee \frac{M_n\eta_1(n)}{\delta^{\ell}}) + \eta_2(n)\bigg)\\
	\lesssim& \frac{1}{n\delta}\psi_1(n) + \delta^{2\ell}\psi_2(n).\label{eq_pf_theorem_adapt_0}
	\eeq
	Recall that $M(\delta) = \frac1n V(\delta) + \frac{n-1}{n}SB(\delta),$ where
	\beq
	V(\delta) = \frac{1}{\delta}(\bv^{*T}\bSigma^*\bv^*)(1+o(1))~~~~~\text{and}~~ {SB(\delta)} = \delta^{2\ell}(\bv^{*T}\bb^*)^2(1+o(1)).\label{eq_pf_theorem_adapt_1}
	\eeq
	This implies that $\delta^*=(\frac{1}{n}\frac{\bv^{*T}\bSigma^*\bv^*}{2\ell(\bv^{*T}\bb^*)^2})^{1/(2\ell + 1)}(1+o(1))$. By Assumption \ref{ass_projection_norm} we know $M(\delta) \gtrsim \frac{1}{n\delta} + \delta^{2\ell}$, which implies
	\[
	\frac{|\wh M(\delta) - M(\delta)|}{M(\delta)} \lesssim \psi_1(n) \vee \psi_2(n)= o_{\PP}(1),
	\]
	uniformly over all $\delta\in\Delta$. That means for any $\epsilon>0$, uniformly over all $\delta\in\Delta$, the event $(1-\epsilon)M(\delta)\leq \wh M(\delta)\leq (1+\epsilon)M(\delta)$ holds with probability tending to 1. Under this event, we have that
	\beq
	M(\hat\delta)&=\wh M(\hat\delta)+[M(\hat\delta)-\hat M(\hat\delta)]\\
	&\leq \wh M(\delta^*)+[M(\hat\delta)-\hat M(\hat\delta)]\\
	&\leq (1+\epsilon) M(\delta^*)+\epsilon M(\hat\delta),
	\eeq
	where the first inequality holds because $\hat\delta$ is the minimizer of $\wh M(\delta)$ and the second inequality follows from $\wh M(\delta^*)\leq (1+\epsilon)M(\delta^*)$ and $\wh M(\wh\delta)\geq (1-\epsilon)M(\wh\delta)$. As a result, we obtain $M(\hat\delta)\leq \frac{1+\epsilon}{1-\epsilon}M(\delta^*)$, from which we can claim that $\hat\delta\leq \frac{3}{2}\delta^*$. To see this, let us consider the complement case $\hat\delta> \frac{3}{2}\delta^*$. When it holds, we know from (\ref{eq_pf_theorem_adapt_1}) that
	$$
	M(\hat\delta)\geq (\frac{3}{2})^{2\ell}SB(\delta^*)(1+o(1))=(\frac{3}{2})^{2\ell}\frac{1}{2\ell+1}M(\delta^*)(1+o(1))\geq 1.01M(\delta^*)(1+o(1)),
	$$
	where the second step follows from $M(\delta^*)=(2\ell+1)SB(\delta^*)(1+o(1))$ and last step holds by the condition $\ell\geq 2$. However, the above inequality contradicts with $M(\hat\delta)\leq \frac{1+\epsilon}{1-\epsilon}M(\delta^*)$ for some sufficiently small $\epsilon$. This justifies the statement $\hat\delta\leq \frac{3}{2}\delta^*$. Following a similar argument, we can show that $\hat\delta\geq \frac{1}{2}\delta^*$. Thus, with probability tending to 1, we have $|\hat\delta-\delta^*|/\delta^*\leq 1/2$.

	By definition, we have $M'(\delta^*) = 0$ and $\wh{M}'(\wh{\delta}) = 0$. This implies that
	\beq
	\wh{M}'(\wh{\delta}) - M'(\wh{\delta}) = -M''(\tilde{\delta})(\wh{\delta} - \delta^*),
	\eeq
	for some intermediate value $\tilde{\delta}$, which gives
	\beq
	(\wh{\delta} - \delta^*) = \frac{\wh{M}'(\wh{\delta}) - M'(\wh{\delta}) }{-M''(\tilde{\delta})}.\label{eq_pf_theorem_adapt_2}
	\eeq
	For any ${\delta} \asymp \delta^*$, after some algebra, we can show that
	\beq
	M''(\delta) \asymp \frac{1}{\delta^{3}n}.
	\eeq
	Similar to the proof of (\ref{eq_pf_theorem_adapt_0}), we can show that $(\wh{M}' - M')(\delta)$ has the same order as $(\wh{M} - M)(\delta)$, except for an additional factor $\frac{1}{\delta}$. Since $|\hat\delta-\delta^*|/\delta^*\leq 1/2$ holds in probability, it implies $\wh\delta$ and $\tilde\delta$ are of the same order of $\delta^*$. Thus, \eqref{eq_pf_theorem_adapt_2} implies
	\beq
	\frac{\wh{\delta} - \delta^*}{\delta^*} \lesssim& \psi_1(n) \vee \psi_2(n).
	\eeq
	This completes the proof.
	
\subsubsection{Proof of Theorem \ref{thm1_11}}	

We start from the following two lemmas.

\begin{lemma} \label{lemma1_11}
Denote $s(\delta) = \sqrt{\delta}\boldsymbol{v}^{T}\nabla R_{\delta}^{n}(\boldsymbol{\beta})$ for some fix $\boldsymbol{\beta}$ and $\boldsymbol{v}$, under the conditions of Theorem~\ref{thm1_11} we have
\[
\sqrt{n}\left|s(\hat{\delta})-s(\delta^*)
-s^\prime(\delta^*)(\hat{\delta}-\delta^*)
\right|=o_{\mathbb{P}}(1).
\]
\end{lemma}	
\begin{proof}
Note that
\[
s(\delta)= \frac{1}{n}\sum\limits_{i=1}^n w(y_i)y_i\boldsymbol{z}_i^T\boldsymbol{v}\frac{1}{\sqrt{\delta}}
K\left(\frac{x_i-\boldsymbol{\beta}^T\boldsymbol{z_i}}{\delta}\right),
\]
\[
s^{\prime}(\delta)=\frac{1}{n}\sum\limits_{i=1}^n w(y_i)y_i\boldsymbol{z}_i^T\boldsymbol{v}\left[
-\frac{1}{2}\delta^{-\frac{3}{2}}K\left(\frac{x_i-\boldsymbol{\beta}^T\boldsymbol{z_i}}{\delta}\right)-
(x_i-\boldsymbol{\beta}^T\boldsymbol{z_i})\delta^{-\frac{5}{2}}K^{\prime}\left(\frac{x_i-\boldsymbol{\beta}^T\boldsymbol{z_i}}{\delta}\right)
\right].
\]
With some algebra we have
\begin{align}
    &\sqrt{n}\left|s(\hat{\delta})-s(\delta^*)
-s^\prime(\delta^*)(\hat{\delta}-\delta^*)
\right|\nonumber\\
=& \sqrt{n}\Bigg|
\frac{1}{n}\sum\limits_{i=1}^n w(y_i)y_i\boldsymbol{z}_i^T\boldsymbol{v}
\int_{\delta^*}^{\hat{\delta}} \Bigg[(x_i-\boldsymbol{\beta}^T\boldsymbol{z_i})^2K^{\prime \prime}\left(\frac{x_i-\boldsymbol{\beta}^T\boldsymbol{z_i}}{t}\right)t^{-\frac{9}{2}}\nonumber\\
+ &3(x_i-\boldsymbol{\beta}^T\boldsymbol{z_i})K^{\prime}\left(\frac{x_i-\boldsymbol{\beta}^T\boldsymbol{z_i}}{t}\right)t^{-\frac{7}{2}}+
\frac{3}{4}K\left(\frac{x_i-\boldsymbol{\beta}^T\boldsymbol{z_i}}{t}\right)t^{-\frac{5}{2}}\Bigg](\hat{\delta}-t)dt
\Bigg| \label{lemma0.1}
\end{align}

Now consider the first term in the integral in (\ref{lemma0.1}).
Denote
$$
G_i = w(y_i)\boldsymbol{z}_iy_i\int_{\delta^*}^{\hat{\delta}} (x_i-\boldsymbol{\beta}^T\boldsymbol{z_i})^2K^{\prime \prime}\left(\frac{x_i-\boldsymbol{\beta}^T\boldsymbol{z_i}}{t}\right)t^{-\frac{9}{2}}(\hat{\delta}-t)dt$$
and $G_{im}, 1 \leq m \leq d$ as its coordinate. Consider event $A \coloneqq \{\frac{\widehat{\delta}-\delta^{*}}{\delta^{*}} \lesssim C_{n,\delta*}\}$, where $C_{n,\delta*}=\psi_{1}\left(n, \delta^{*}\right) \vee \psi_{2}\left(n, \delta^{*}\right)$. For each $m$ we have
{\small
\begin{align*}
    &\sum_{y} w(y) \int_{z} z_my\int_x
    (x-\boldsymbol{\beta}^T\boldsymbol{z})^2
    \int_{\delta^*}^{\hat{\delta}} K^{\prime \prime}\left(\frac{x-\boldsymbol{\beta}^T\boldsymbol{z}}{t}\right)t^{-\frac{9}{2}}(\hat{\delta}-t)dt f(x \mid y,\boldsymbol{z})dx f(y, \boldsymbol{z})d\boldsymbol{z}\\
    x-\boldsymbol{\beta}^T\boldsymbol{z}=u\delta^* \atop =& (\delta^*)^3\sum_{y} w(y) \int_{z} z_my
    \int_u
    u^2
    \int_{\delta^*}^{\hat{\delta}} K^{\prime \prime}\left(\frac{u\delta^*}{t}\right)t^{-\frac{9}{2}}(\hat{\delta}-t)dt f(u\delta^*+\boldsymbol{\beta}^T\boldsymbol{z} \mid y,\boldsymbol{z})du f(y, \boldsymbol{z})d\boldsymbol{z}.
\end{align*}
}
Note that $ \frac{\widehat{\delta}-\delta^{*}}{\delta^{*}}
\coloneqq C_{n,\delta^*}=o(1)$, hence
\begin{align*}
     &\int_u u^2
    \int_{\delta^*}^{\hat{\delta}} K^{\prime \prime}\left(\frac{u\delta^*}{t}\right)t^{-\frac{9}{2}}(\hat{\delta}-t)dt f(u\delta^*+\boldsymbol{\beta}^T\boldsymbol{z} \mid y,\boldsymbol{z})du\\
    &\lesssim \delta^{* -\frac{9}{2}}(\hat{\delta}-\delta^*)^2 \int_u u^2 K^{\prime \prime}\left(u(1+o(1))\right)\left(
    f(\boldsymbol{\beta}^T\boldsymbol{z} \mid y,\boldsymbol{z})+
    u\delta^* f^\prime(\tau u\delta^* + \boldsymbol{\beta}^T\boldsymbol{z} \mid y,\boldsymbol{z})  \right)du,
\end{align*}
where $\tau \in [0,1]$. Since $\int u^2 |K^{\prime\prime}(u)|du, \int |u^3| |K^{\prime\prime}(u)|du$ are bounded by a constant,
 the last formula is bounded by $C\delta^{* \frac{1}{2}}C_{n,\delta*}^2$ for some constant $C$.
This implies that $\mathbb{E}\left[G_{i m} \mid A\right] \lesssim \delta^{* \frac{1}{2}}C_{n,\delta*}^2$.

Then we look at its variance. Similarly, note that $\int u^4 K^{\prime\prime}(u)^2 du, \int |u^5| K^{\prime\prime}(u)^2 du$ are bounded by a constant, we have
\begin{align*}
    & \mathbb{E}[G_{im}^2\mid A] \\
    =&
    \sum_{y} w(y)^2 \int_{z} z_m^2 \int_x (x-\boldsymbol{\beta}^T\boldsymbol{z})^4 \left(
    \int_{\delta^*}^{\hat{\delta}} K^{\prime \prime}\left(\frac{x-\boldsymbol{\beta}^T\boldsymbol{z}}{t}\right)t^{-\frac{9}{2}}(\hat{\delta}-t)dt
    \right)^2
    f(x \mid y,\boldsymbol{z})dx f(y, \boldsymbol{z})d\boldsymbol{z} \\
    \lesssim & (\delta^*)^5(\delta^{* -\frac{5}{2}}C_{n,\delta*}^2)^2=C_{n,\delta*}^4.
\end{align*}
Also we know for each $m=1, \ldots, d$,
\[
\left|z_{m}y\int_{\delta^*}^{\hat{\delta}} (x-\boldsymbol{\beta}^T\boldsymbol{z})^2K^{\prime \prime}\left(\frac{x-\boldsymbol{\beta}^T\boldsymbol{z}}{t}\right)t^{-\frac{9}{2}}(\hat{\delta}-t)dt\right| \lesssim M_n^3 \delta^{* -\frac{5}{2}}C_{n,\delta*}^2.
\]
Therefore, applying Bernstein inequality with $M_n^3 \delta^{* -\frac{5}{2}} \sqrt{\frac{\log d}{n}} = \mathcal{O}(1)$, we have for some constant $C$,
\[
\mathbb{P}\left(\max _{m}\left|\frac{1}{n} \sum_{i=1}^n G_{i m}-\mathbb{E} G_{im}\right|> C\sqrt{\frac{\log d}{n}}C_{n,\delta*}^2 \mid A\right) \leq \mathcal{O}\left(d^{-1}\right),
\]
which further implies that
\begin{align*}
    &\mathbb{P}\left(\max _{m}\left|\frac{1}{n} \sum_{i=1}^n G_{i m}-\mathbb{E} G_{im}\right|> C\sqrt{\frac{\log d}{n}}C_{n,\delta*}^2\right)\\
    \leq& \mathbb{P}\left(\max _{m}\left|\frac{1}{n} \sum_{i=1}^n G_{i m}-\mathbb{E} G_{im}\right|> C\sqrt{\frac{\log d}{n}}C_{n,\delta*}^2 \mid A\right)+ \mathbb{P}(A^c)= o(1).
\end{align*}
Hence
\begin{align}\label{G}
    \left\|\frac{1}{n}\sum\limits_{i=1}^nG_{im}\right\|_\infty = \mathcal{O}_{\mathbb{P}}(\delta^{* \frac{1}{2}}C_{n,\delta*}^2).
\end{align}

Similarly, denote
\[
H_i = w(y_i)\boldsymbol{z}_iy_i\int_{\delta^*}^{\hat{\delta}} (x_i-\boldsymbol{\beta}^T\boldsymbol{z_i})K^{ \prime}\left(\frac{x_i-\boldsymbol{\beta}^T\boldsymbol{z_i}}{t}\right)t^{-\frac{7}{2}}(\hat{\delta}-t)dt,
\]
\[
J_i = w(y_i)\boldsymbol{z}_iy_i\int_{\delta^*}^{\hat{\delta}} K\left(\frac{x_i-\boldsymbol{\beta}^T\boldsymbol{z_i}}{t}\right)t^{-\frac{5}{2}}(\hat{\delta}-t)dt
\]
and $H_{im}, J_{im}, 1 \leq m \leq d$ as their coordinate. Note that for $H_{im}$, we have
\[
\int |u| |K^{\prime}(u)|du,
\int u^2 |K^{\prime}(u)|du, \int u^2 K^{\prime}(u)^2 du
\int |u^3| K^{\prime}(u)^2 du
\]
are bounded, for $J_{im}$, we have $\int |K(u)|du,\int |u| |K(u)|du,  \int  K(u)^2 du, \int |u| K(u)^2 du$ are bounded,
it can be shown that
$\mathbb{E}\left[H_{i m} \mid A\right] \lesssim \delta^{* \frac{1}{2}}C_{n,\delta*}^2$, $\mathbb{E}\left[J_{i m} \mid A\right] \lesssim \delta^{* \frac{1}{2}}C_{n,\delta*}^2$, $\mathbb{E}[H_{im}^2\mid A] \lesssim C_{n,\delta*}^4 $, $\mathbb{E}[J_{im}^2\mid A] \lesssim C_{n,\delta*}^4 $, $|H_{im}|\lesssim M_n^2\delta^{* -\frac{3}{2}}C_{n,\delta*}^2$ and  $|J_{im}|\lesssim M_n\delta^{* -\frac{1}{2}}C_{n,\delta*}^2$.
Therefore, applying Bernstein inequality with $M_n^2 \delta^{* -\frac{3}{2}} \sqrt{\frac{\log d}{n}} = \mathcal{O}(1)$, $M_n \delta^{* -\frac{1}{2}} \sqrt{\frac{\log d}{n}} = \mathcal{O}(1)$, we have for some constant $C$,
\[
\mathbb{P}\left(\max _{m}\left|\frac{1}{n} \sum_{i=1}^n H_{i m}-\mathbb{E} H_{im}\right|> C\sqrt{\frac{\log d}{n}}C_{n,\delta*}^2 \mid A\right) \leq \mathcal{O}\left(d^{-1}\right),
\]
\[
\mathbb{P}\left(\max _{m}\left|\frac{1}{n} \sum_{i=1}^n J_{i m}-\mathbb{E} J_{im}\right|> C\sqrt{\frac{\log d}{n}}C_{n,\delta*}^2 \mid A\right) \leq \mathcal{O}\left(d^{-1}\right),
\]
respectively. Hence
\begin{align}\label{H&J}
    \left\|\frac{1}{n}\sum\limits_{i=1}^nH_{im}\right\|_\infty, \left\|\frac{1}{n}\sum\limits_{i=1}^nJ_{im}\right\|_\infty = \mathcal{O}_{\mathbb{P}}(\delta^{* \frac{1}{2}}C_{n,\delta*}^2).
\end{align}
Plugging (\ref{G}) and (\ref{H&J}) back to (\ref{lemma0.1}) we have
\[
\sqrt{n}\left|s(\hat{\delta})-s(\delta^*)
-s^\prime(\delta^*)(\hat{\delta}-\delta^*)
\right| = \mathcal{O}_{\mathbb{P}}(\sqrt{n}\|\boldsymbol{v}\|_1\delta^{* \frac{1}{2}}C_{n,\delta*}^2).
\]
Under the conditions of Theorem~\ref{thm1_11}, the desired result holds.
\end{proof}

\begin{lemma} \label{lemma2_11} Under the conditions of Theorem~\ref{thm1_11}, we have
\[
\sqrt{n\delta^*}(\hat{\delta}-\delta^*)\left|\boldsymbol{v}^{* T}\left(\nabla_{\cdot \delta} R_{\delta^*}^{n_{(j)}}\left(\widehat{\boldsymbol{\beta}}^{(k)}\right)-\nabla_{\cdot \delta} R_{\delta^*}^{n_{(j)}}\left(\boldsymbol{\beta}^{*}\right)-\nabla^{2}_{\cdot \delta} R_{\delta^*}^{n_{(j)}}\left(\boldsymbol{\beta}^{*}\right)\left(\widehat{\boldsymbol{\beta}}^{(k)}-\boldsymbol{\beta}^{*}\right)\right)\right|=o_{\mathbb{P}}(1)
\]
for $(j, k) \in\{(1,2),(2,1)\}$.
\end{lemma}
\begin{proof}
Note that
\[
\nabla \bar{R}_{\delta}^{i}\left(\boldsymbol{\beta}\right)=w\left(y_{i}\right) \frac{y_{i} \boldsymbol{z}_{i}}{\delta} K\left(\frac{y_{i}\left(x_{i}-\boldsymbol{\beta}^{ T} \boldsymbol{z}_{i}\right)}{\delta}\right),
\]
\[
\nabla_{\cdot \delta} \bar{R}_{\delta}^{i}\left(\boldsymbol{\beta}\right)=
\nabla_\delta(\nabla \bar{R}_{\delta}^{i}\left(\boldsymbol{\beta}\right))
=-w(y_i)\frac{y_i\boldsymbol{z_i}}{\delta^2} K\left(\frac{x_{i}-\boldsymbol{\beta}^{ T} \boldsymbol{z}_{i}}{\delta}\right)-w(y_i)\frac{y_i\boldsymbol{z}_i}{\delta}K^\prime\left(\frac{x_{i}-\boldsymbol{\beta}^{ T} \boldsymbol{z}_{i}}{\delta}\right)\frac{x_{i}-\boldsymbol{\beta}^{ T} \boldsymbol{z}_{i}}{\delta^2}.
\]
With some algebra we have
\begin{align}
    & \left|\boldsymbol{v}^{* T}\left(\nabla_{\cdot \delta} R_{\delta^*}^{n_{(j)}}\left(\widehat{\boldsymbol{\beta}}^{(k)}\right)-\nabla_{\cdot \delta} R_{\delta^*}^{n_{(j)}}\left(\boldsymbol{\beta}^{*}\right)-\nabla^{2}_{\cdot \delta} R_{\delta^*}^{n_{(j)}}\left(\boldsymbol{\beta}^{*}\right)\left(\widehat{\boldsymbol{\beta}}^{(k)}-\boldsymbol{\beta}^{*}\right)\right)\right|\nonumber\\
    \leq &\left|
    \frac{1}{\mathcal{N}_{j}} \sum_{i \in \mathcal{N}_{j}}
    3w(y_i)\frac{y_i\boldsymbol{z}_i^T\boldsymbol{v}^*}{\delta^2} \int_{\frac{x_{i}-\boldsymbol{\beta}^{* T} \boldsymbol{z}_{i}}{\delta}}^{\frac{x_{i}-\widehat{\boldsymbol{\beta}}^{(k) T} \boldsymbol{z}_{i}}{\delta}} K^{\prime \prime}(t)\left(\frac{x_{i}-\widehat{\boldsymbol{\beta}}^{(k) T} \boldsymbol{z}_{i}}{\delta}-t\right) dt
    \right|\nonumber\\
    +&\left|
    \frac{1}{\mathcal{N}_{j}} \sum_{i \in \mathcal{N}_{j}}
    w(y_i)\frac{y_i\boldsymbol{z}_i^T\boldsymbol{v}^*}{\delta^2}
    \int_{\frac{x_{i}-\boldsymbol{\beta}^{* T} \boldsymbol{z}_{i}}{\delta}}^{\frac{x_{i}-\widehat{\boldsymbol{\beta}}^{(k) T} \boldsymbol{z}_{i}}{\delta}} K^{\prime\prime \prime}(t)t\left(\frac{x_{i}-\widehat{\boldsymbol{\beta}}^{(k) T} \boldsymbol{z}_{i}}{\delta}-t\right) dt
    \right|\nonumber\\
    \leq &
    \left\|\boldsymbol{v}^{*}\right\|_{1}
    \left\|\frac{1}{\mathcal{N}_{j}} \sum_{i \in \mathcal{N}_{j}} 3w(y_i)\frac{y_{i}\boldsymbol{z}_{i} }{\delta^2} \int_{\frac{x_{i}-\boldsymbol{\beta}^{* T} \boldsymbol{z}_{i}}{\delta}}^{\frac{x_{i}-\widehat{\boldsymbol{\beta}}^{(k) T} \boldsymbol{z}_{i}}{\delta}} K^{\prime \prime}(t)\left(\frac{x_{i}-\widehat{\boldsymbol{\beta}}^{(k) T} \boldsymbol{z}_{i}}{\delta}-t\right) d t\right\|_{\infty}\nonumber \\
    +&
     \left\|\boldsymbol{v}^{*}\right\|_{1}
     \left\|\frac{1}{\mathcal{N}_{j}} \sum_{i \in \mathcal{N}_{j}} w(y_i)\frac{y_{i}\boldsymbol{z}_{i} }{\delta^2} \int_{\frac{x_{i}-\boldsymbol{\beta}^{* T} \boldsymbol{z}_{i}}{\delta}}^{\frac{x_{i}-\widehat{\boldsymbol{\beta}}^{(k) T} \boldsymbol{z}_{i}}{\delta}} K^{\prime\prime \prime}(t)t\left(\frac{x_{i}-\widehat{\boldsymbol{\beta}}^{(k) T} \boldsymbol{z}_{i}}{\delta}-t\right) d t\right\|_{\infty} \label{lemma0.2}
\end{align}
Denote
\[
G_i = w(y_i)\frac{ y_{i}\boldsymbol{z}_{i}}{\delta^2} \int_{\frac{x_{i}-\boldsymbol{\beta}^{* T} \boldsymbol{z}_{i}}{\delta}}^{\frac{x_{i}-\widehat{\boldsymbol{\beta}}^{(k) T} \boldsymbol{z}_{i}}{\delta}} K^{\prime \prime}(t)\left(\frac{x_{i}-\widehat{\boldsymbol{\beta}}^{(k) T} \boldsymbol{z}_{i}}{\delta}-t\right) d t,
\]
\[
F_i = w(y_i)\frac{y_{i}\boldsymbol{z}_{i}}{\delta^2} \int_{\frac{x_{i}-\boldsymbol{\beta}^{* T} \boldsymbol{z}_{i}}{\delta}}^{\frac{x_{i}-\widehat{\boldsymbol{\beta}}^{(k) T} \boldsymbol{z}_{i}}{\delta}} K^{\prime\prime \prime}(t)t\left(\frac{x_{i}-\widehat{\boldsymbol{\beta}}^{(k) T} \boldsymbol{z}_{i}}{\delta}-t\right) d t,
\]
and $G_{im}, F_{im}$, $1 \leq m \leq d$ as its coordinates. Consider the event $A:=\left\{\left\|\widehat{\boldsymbol{\beta}}^{(k)}-\boldsymbol{\beta}^{*}\right\|_{1} \lesssim C \eta_{1}(n)\right\}$ for some constant $C$. Denote $\Delta=\frac{\left(\widehat{\boldsymbol{\beta}}^{(k)}-\boldsymbol{\beta}\right)^{T} \boldsymbol{z}}{\delta}$, then following the similar proof in Lemma \ref{lemma_second-order-samplediff}, we can show that $\mathbb{E}\left[G_{i m} \mid A\right] \lesssim |\Delta|^2 \asymp  \frac{M_{n}^{2} \eta_{1}(n)^{2}}{\delta^2}$, and $\mathbb{E}\left[G_{i m}^2 \mid A\right] \lesssim  \frac{M_{n}^{4} \eta_{1}(n)^{4}}{\delta^{7}}$. Also we know for each $i \in \mathcal{N}_{j}, m=1, \ldots, d$, $|G_{im}| \lesssim \frac{M_{n}}{\delta^2}\left(\frac{M_{n} \eta_{1}(n)}{\delta}\right)^{2}$. Therefore, by Bernstein inequality with $M_{n} \sqrt{\frac{\log (d)}{n \delta}}=\mathcal{O}(1)$, we can obtain that
\begin{align*}
    \mathbb{P}\left(\max _{m}\left|\frac{1}{\mathcal{N}_{j}} \sum_{i \in \mathcal{N}_{j}} G_{i m}-\mathbb{E} G_{m}\right|>\frac{M_{n}^{2} \eta_{1}(n)^{2}}{\delta^2} \sqrt{\frac{\log (d)}{n \delta^{3}}} \mid A\right) \leq \mathcal{O}\left(d^{-1}\right).
\end{align*}
Then since the event $A$ holds with probability tending to 1, we have
\begin{equation*}
    \mathbb{P}\left(\max _{m}\left|\frac{1}{\mathcal{N}_{j}} \sum_{i \in \mathcal{N}_{j}} G_{i m}-\mathbb{E} G_{m}\right|>\frac{M_{n}^{2} \eta_{1}(n)^{2}}{\delta^2} \sqrt{\frac{\log (d)}{n \delta^{3}}} \right)=o(1),
\end{equation*}
and therefore the first term in (\ref{lemma0.2}) follows
\begin{equation}\label{Gim}
     \left\|\frac{1}{\mathcal{N}_{j}} \sum_{i \in \mathcal{N}_{j}} 3w(y_i)\frac{y_{i}\boldsymbol{z}_{i} }{\delta^2} \int_{\frac{x_{i}-\boldsymbol{\beta}^{* T} \boldsymbol{z}_{i}}{\delta}}^{\frac{x_{i}-\widehat{\boldsymbol{\beta}}^{(k) T} \boldsymbol{z}_{i}}{\delta}} K^{\prime \prime}(t)\left(\frac{x_{i}-\widehat{\boldsymbol{\beta}}^{(k) T} \boldsymbol{z}_{i}}{\delta}-t\right) d t\right\|_{\infty}=
     \mathcal{O}_{\mathbb{P}}\left(\frac{M_{n}^{2} \eta_{1}(n)^{2}}{\delta^2}\right).
\end{equation}

Now we analyze $F_i$. Note that
{\small
\begin{align}
    &\sum_{y} w(y) \int_{\boldsymbol{z}} \frac{z_{m} y}{\delta^2} \int_{x} \int_{\frac{x-\boldsymbol{\beta}^{* T}\boldsymbol{z}}{\delta}}^{\frac{x-\hat{\boldsymbol{\beta}}^{(k) T}\boldsymbol{z}}{\delta}} K^{\prime\prime \prime}(t)t\left(\frac{x-\widehat{\boldsymbol{\beta}}^{(k) T} \boldsymbol{z}}{\delta}-t\right) d t f(x \mid y, \boldsymbol{z}) d x f(y, \boldsymbol{z}) d \boldsymbol{z}\nonumber\\
   (u=\left(x-\boldsymbol{\beta}^{* T} \boldsymbol{z}\right) / \delta)\atop=
   & \sum_{y} w(y) \int_{\boldsymbol{z}}\frac{z_{m} y}{\delta} \int_{u} \int_{u}^{u+\Delta}
   K^{\prime\prime \prime}(t)t(u+\Delta-t)dt f\left(u \delta+\boldsymbol{\beta}^{* T} \boldsymbol{z} \mid y, \boldsymbol{z}\right) d u f(y, \boldsymbol{z}) d \boldsymbol{z}\nonumber\\
   = & \sum_{y} w(y) \int_{\boldsymbol{z}}\frac{z_{m} y}{\delta} \int_{u} \int_{u}^{u+\Delta}
   K^{\prime\prime \prime}(t)t(u+\Delta-t)dt f\left(\boldsymbol{\beta}^{* T} \boldsymbol{z} \mid y, \boldsymbol{z}\right) d u f(y, \boldsymbol{z}) d \boldsymbol{z}\nonumber\\
   +&
   \sum_{y} w(y) \int_{\boldsymbol{z}}z_{m} y \int_{u}u \int_{u}^{u+\Delta}
   K^{\prime\prime \prime}(t)t(u+\Delta-t)dt f^\prime\left(\tau u\delta+ \boldsymbol{\beta}^{* T} \boldsymbol{z} \mid y, \boldsymbol{z}\right) d u f(y, \boldsymbol{z}) d \boldsymbol{z}, \label{lemma0.4}
\end{align}
}
where $\tau \in[0,1]$.
Note that
\begin{align*}
     \int_{u}^{u+\Delta}
   K^{\prime\prime \prime}(t)t(u+\Delta-t)dt
   =-u\Delta K^{\prime\prime}(u)+(u+\Delta)K^\prime(u+\Delta)+(\Delta-u)K^\prime(u)-2K(u+\Delta)+2K(u).
\end{align*}
Since $K^\prime$ degenerates at the boundaries and $\int K(u) du =1$, we have
\begin{align*}
    \int_{u} \int_{u}^{u+\Delta}
   K^{\prime\prime \prime}(t)t(u+\Delta-t)dt f\left(\boldsymbol{\beta}^{* T} \boldsymbol{z} \mid y, \boldsymbol{z}\right) d u = 0.
\end{align*}
Therefore, we only need to consider the second term of (\ref{lemma0.4}).
Note that $|f^\prime(x \mid y, \boldsymbol{z})|$ is bounded, we have for some constant $C, 0 < \tau <1$,
\begin{align*}
    &\int_{u}u \int_{u}^{u+\Delta}
   K^{\prime\prime \prime}(t)t(u+\Delta-t)dt f^\prime\left(\tau u\delta+ \boldsymbol{\beta}^{* T} \boldsymbol{z} \mid y, \boldsymbol{z}\right) d u\\
   \lesssim & C  \int_t  |t K^{\prime\prime \prime}(t)| \int_{t-\Delta}^t
   |u(u+\Delta-t)|du dt \\
   \leq &  C(1-\tau)\Delta^2 \int_t \left| t K^{\prime\prime \prime}(t) (t-\tau \Delta) \right|dt.
\end{align*}
Note that $\int |t| |K^{\prime\prime \prime}(t)|dt, \int t^2 |K^{\prime\prime \prime}(t)|dt$ are bounded, hence
$\mathbb{E}\left[F_{i m} \mid A\right] \lesssim  \Delta^2  \lesssim \frac{M_n^2\eta_1(n)^2}{\delta^2}$.
For the second moment, note that $\int u^2 K^{\prime\prime \prime}(u)^2 du, \int |u^3| K^{\prime\prime \prime}(u)^2du$ are bounded,
we have
\begin{align*}
&\mathbb{E}\left[\left(w(y)\frac{yz_m}{\delta^2} \int^{\frac{x-\hat\bbeta^{(k)T}\bz}{\delta}}_{\frac{x-\bbeta^{*T}\bz}{\delta}}
K^{\prime\prime \prime}(t)t\left(\frac{x-\widehat{\boldsymbol{\beta}}^{(k) T} \boldsymbol{z}}{\delta}-t\right) d t\right)^2 \mid A\right] \\
    =&\sum_{y} w(y)^2 \int_{\boldsymbol{z}}\frac{z_{m}^2 }{\delta^3} \int_{u} \left(\int_{u}^{u+\Delta}
   K^{\prime\prime \prime}(t)t(u+\Delta-t)dt\right)^2 f\left(u \delta+\boldsymbol{\beta}^{* T} \boldsymbol{z} \mid y, \boldsymbol{z}\right) d u f(y, \boldsymbol{z}) d \boldsymbol{z}\\
   \lesssim &\sum_{y} w(y)^2 \int_{\boldsymbol{z}}\frac{z_{m}^2 }{\delta^3}
   \Delta^4
   f(y, \boldsymbol{z}) d \boldsymbol{z} \lesssim \frac{M_n^4\eta_1(n)^4}{\delta^7}.
\end{align*}
Note that $|F_{i m}| \lesssim \frac{M_{n}}{\delta^{2}}\cdot\frac{M_n}{\delta}\left(\frac{M_{n} \eta_{1}(n)}{\delta}\right)^{2} = \frac{M_n^4\eta_1(n)^2}{\delta^5}$. Therefore, by Bernstein inequality with $\frac{M_n^2}{\delta} \sqrt{\frac{\log (d)}{n \delta}}=\mathcal{O}(1)$, we can obtain that
\begin{equation*}
    \mathbb{P}\left(\max _{m}\left|\frac{1}{\mathcal{N}_{j}} \sum_{i \in \mathcal{N}_{j}} F_{i m}-\mathbb{E} F_{m}\right|>\frac{M_{n}^{2} \eta_{1}(n)^{2}}{\delta^3} \sqrt{\frac{\log (d)}{n \delta}} \mid A\right) \leq \mathcal{O}\left(d^{-1}\right),
\end{equation*}
which further implies that
\begin{equation}\label{Fim}
     \left\|\frac{1}{\mathcal{N}_{j}} \sum_{i \in \mathcal{N}_{j}} w(y_i)\frac{y_{i}\boldsymbol{z}_{i} }{\delta^2} \int_{\frac{x_{i}-\boldsymbol{\beta}^{* T} \boldsymbol{z}_{i}}{\delta}}^{\frac{x_{i}-\widehat{\boldsymbol{\beta}}^{(k) T} \boldsymbol{z}_{i}}{\delta}} K^{\prime\prime \prime}(t)t\left(\frac{x_{i}-\widehat{\boldsymbol{\beta}}^{(k) T} \boldsymbol{z}_{i}}{\delta}-t\right) d t\right\|_{\infty}=
     \mathcal{O}_{\mathbb{P}}\left(\frac{M_{n}^{2} \eta_{1}(n)^{2}}{\delta^2}\right).
\end{equation}
Plugging (\ref{Gim}) and (\ref{Fim}) back to (\ref{lemma0.2}) we have
\begin{align*}
    \left|\boldsymbol{v}^{* T}\left(\nabla_{\cdot \delta} R_{\delta^*}^{n_{(j)}}\left(\widehat{\boldsymbol{\beta}}^{(k)}\right)-\nabla_{\cdot \delta} R_{\delta^*}^{n_{(j)}}\left(\boldsymbol{\beta}^{*}\right)-\nabla^{2}_{\cdot \delta} R_{\delta^*}^{n_{(j)}}\left(\boldsymbol{\beta}^{*}\right)\left(\widehat{\boldsymbol{\beta}}^{(k)}-\boldsymbol{\beta}^{*}\right)\right)\right|=
     \mathcal{O}_{\mathbb{P}}\left(\frac{\|\boldsymbol{v}^*\|_1M_{n}^{2} \eta_{1}(n)^{2}}{\delta^{* 2}}\right),
\end{align*}
and
\begin{align*}
    &\sqrt{n\delta^*}(\hat{\delta}-\delta^*)\left|\boldsymbol{v}^{* T}\left(\nabla_{\cdot \delta} R_{\delta^*}^{n_{(j)}}\left(\widehat{\boldsymbol{\beta}}^{(k)}\right)-\nabla_{\cdot \delta} R_{\delta^*}^{n_{(j)}}\left(\boldsymbol{\beta}^{*}\right)-\nabla^{2}_{\cdot \delta} R_{\delta^*}^{n_{(j)}}\left(\boldsymbol{\beta}^{*}\right)\left(\widehat{\boldsymbol{\beta}}^{(k)}-\boldsymbol{\beta}^{*}\right)\right)\right|\\
    =&\mathcal{O}_{\mathbb{P}}\left(\frac{\sqrt{n}\|\boldsymbol{v^*}\|_1M_{n}^{2} \eta_{1}(n)^{2}C_{n,\delta^*}}{\delta^{* \frac{1}{2}}}\right).
\end{align*}
Under the conditions of Theorem~\ref{thm1_11}, the desired result holds.
\end{proof}

Now, we are ready to prove Theorem \ref{thm1_11}.

\begin{proof}
Recall that $\widehat{U}_{n}(\delta)=\sqrt{n \delta}\left(\frac{\widehat{S}_{\delta}(0, \widehat{\gamma})-\delta^{\ell} \widehat{\mu}}{\widehat{\sigma}}\right)$. Denote $f(\delta) = \sqrt{n \delta}\left(\widehat{S}_{\delta}(0, \widehat{\gamma})-\delta^{\ell} \widehat{\mu}\right)$. Note that the optimal bandwidth $\delta^*$ has the same order as the fixed oracle $\delta$, hence it suffices to show that $ f(\hat{\delta})-f(\delta^*)=o_{\mathbb{P}}(1)$. Recall that to make $\hat\delta$ and the data in the score function independent, we apply the same cross-fitting technique. In the following proof, we omit the superscript for the $i$th fold for simplicity.

Write $f(\delta) = \sqrt{n \delta}\widehat{S}_{\delta}(0, \widehat{\gamma})
- \sqrt{n \delta}\delta^{\ell} \widehat{\mu}\coloneqq f_1(\delta)+f_2(\delta)$. In Lemma~\ref{lemma1_11} we showed that
\[
\left|
f_1(\hat{\delta})-f_1(\delta^*)-f_1^\prime(\delta^*)(\hat{\delta}-\delta^*)
\right|=o_{\mathbb{P}}(1).
\]
Note that
\[
\left|
f_2(\hat{\delta})-f_2(\delta^*)-f_2^\prime(\delta^*)(\hat{\delta}-\delta^*)
\right|= \mathcal{O}\left( \sqrt{n}\Tilde{\delta}^{\ell-\frac{3}{2}}\delta^{* 2}C_{n,\delta^*}^2 \hat{\mu} \right),
\]
where $\Tilde{\delta}$ is between $\delta^*$ and $\hat{\delta}$. Note that in Lemma \ref{lemma_consistent_bias_pilot} we have shown that $|\hat{\mu}-\mu^*|=o_{\mathbb{P}}(1)$, since $\mu^*$ is bounded away from 0 and $\infty$, we have $\sqrt{n}\Tilde{\delta}^{\ell-\frac{3}{2}}\delta^{* 2}C_{n,\delta^*}^2 \hat{\mu}
\lesssim \sqrt{n}\delta^{* \ell+\frac{1}{2}}C_{n,\delta^*}^2 =
o_{\mathbb{P}}(1)$ under the conditions of Theorem~\ref{thm1_11}.

For simplicity we write $\nabla_{\cdot \delta} R_{\delta}^{n}(\boldsymbol{\beta})=\nabla_\delta(\nabla R_{\delta}^{n}(\boldsymbol{\beta}))$, we can write
\begin{align*}
    &f(\hat{\delta})-f(\delta^*)\\
    =&f^\prime(\delta^*)(\hat{\delta}-\delta^*)+o_{\mathbb{P}}(1)\\
    =
    &\frac{\hat{\delta}-\delta^*}{2\delta^*}\sqrt{n\delta^*}\left(\widehat{S}_{\delta^*}(0, \widehat{\gamma})-\delta^{* \ell} \widehat{\mu}\right)
    +
    \sqrt{n\delta^*}\left(\hat{\boldsymbol{v}}^{T}\nabla_{\cdot \delta} R_{\delta^*}^{n}(0, \hat{\gamma})
    -\ell\delta^{* \ell-1}\hat{\mu}
    \right)
    (\hat{\delta}-\delta^*)
    +o_{\mathbb{P}}(1).
\end{align*}
In Theorem \ref{theorem_score}, we showed that  $\sqrt{n\delta}\left(\widehat{S}_{\delta}(0, \widehat{\gamma})-\delta^\ell u^*\right)=O_{\mathbb{P}}(1)$ for the oracle $\delta^*$. Then under the assumption that $|\hat{\mu}-\mu^*|=o_{\mathbb{P}}(1)$, we have $  \frac{\hat{\delta}-\delta^*}{2\delta^*}\sqrt{n\delta^*}\left(\widehat{S}_{\delta^*}(0, \widehat{\gamma})-\delta^{* \ell} \widehat{\mu}\right) = o_{\mathbb{P}}(1)$.

Denote
\begin{align*}
I_{1}&=\mathbb{E}(\boldsymbol{v}^{* T}\nabla_{\cdot \delta} R_{\delta^*}^{n}(\boldsymbol{\beta}^*))
    -\ell\delta^{* \ell-1}\hat{\mu},\\
I_{2}&=\boldsymbol{v}^{* T}\nabla_{\cdot \delta} R_{\delta^*}^{n}(\boldsymbol{\beta}^*)
    -\mathbb{E}(\boldsymbol{v}^{* T}\nabla_{\cdot \delta} R_{\delta^*}^{n}(\boldsymbol{\beta}^*)),\\
I_{3}&=\hat{\boldsymbol{v}}^{T}\nabla_{\cdot \delta} R_{\delta^*}^{n}(0, \hat{\gamma})
    -\boldsymbol{v}^{* T}\nabla_{\cdot \delta} R_{\delta^*}^{n}(\boldsymbol{\beta}^*).
\end{align*}
Then we can write
\begin{align*}
        \sqrt{n\delta^*}\left(\hat{\boldsymbol{v}}^{T}\nabla_{\cdot \delta} R_{\delta^*}^{n}(0, \hat{\gamma})
    -\ell\delta^{* \ell-1}\hat{\mu}
    \right)
    (\hat{\delta}-\delta^*) =  \sqrt{n\delta^*} (\hat{\delta}-\delta^*)\left(
    I_{1}+I_{2}+I_{3}
    \right).
\end{align*}
Note that $\mathbb{E}(\boldsymbol{v}^{* T}\nabla_{\cdot \delta} R_{\delta^*}^{n}(\boldsymbol{\beta}^*))=\boldsymbol{v}^{* T}\nabla_{\cdot \delta}R_{\delta^*}(\boldsymbol{\beta}^*)$ and
\[
\boldsymbol{v}^{* T} \nabla R_{\delta}\left(\boldsymbol{\beta}^{*}\right)=\sum_{y \in\{-1,1\}} w(y) y \int\left(\boldsymbol{v}^{* T} \boldsymbol{z}\right) \int K(u) f\left(u \delta+\boldsymbol{\beta}^{* T} \boldsymbol{z} \mid y, \boldsymbol{z}\right) d u f(y, \boldsymbol{z}) d \boldsymbol{z},
\]
hence
\begin{align}
    \boldsymbol{v}^{* T} \nabla_{\cdot \delta} R_{\delta}\left(\boldsymbol{\beta}^{*}\right)
&=\sum_{y \in\{-1,1\}} w(y) y \int\left(\boldsymbol{v}^{* T} \boldsymbol{z}\right) \int K(u) u f^\prime\left(u \delta+\boldsymbol{\beta}^{* T} \boldsymbol{z} \mid y, \boldsymbol{z}\right) d u f(y, \boldsymbol{z}) d \boldsymbol{z} \nonumber\\
&=\sum_{y \in\{-1,1\}} w(y) y \int\left(\boldsymbol{v}^{* T} \boldsymbol{z}\right) \int K(u) \frac{u^{\ell}\delta^{\ell-1}}{(\ell-1) !}\left(f^{(\ell)}\left(\boldsymbol{\beta}^{* T} \boldsymbol{z} \mid y, \boldsymbol{z}\right)+\mathcal{O}\left((u \delta)^{\zeta}\right)\right) d u f(y, \boldsymbol{z}) d \boldsymbol{z} \nonumber\\
&=\ell\delta^{\ell-1} \boldsymbol{v}^{* T} \boldsymbol{b}^{*}(1+o(1))
=\ell\delta^{\ell-1}\mu^*(1+o(1)).\label{pf1}
\end{align}
Under the assumption that $|\hat{\mu}-\mu^*|=o_{\mathbb{P}}(1)$, $\delta^* \asymp n^{-1 /(2 \ell+1)}$ and $\frac{\widehat{\delta}-\delta^{*}}{\delta^{*}}=o_{\mathbb{P}}(1)$, we have
\[
 \sqrt{n\delta^*} (\hat{\delta}-\delta^*)I_{1}= \sqrt{n\delta^*} (\hat{\delta}-\delta^*)\ell\delta^{* (\ell-1)}(\mu^*-\hat{\mu}+o(\mu^*))
 \asymp n^{\frac{1}{2 \ell+1}}(\hat{\delta}-\delta^*)
 =o_{\mathbb{P}}(1).
\]

Next we consider $I_{2}$. First write $\nabla_{\cdot \delta} R_{\delta}^{n}(\boldsymbol{\beta}^*)=\frac{1}{n}\sum_{i=1}^n
\nabla_{\cdot \delta} \bar{R}_{\delta}^{i}\left(\boldsymbol{\beta}^{*}\right)$. Denote
\[
T_{ij}=\left(\nabla_{\cdot \delta} \bar{R}_{\delta}^{i}(\boldsymbol{\beta}^*)-\nabla_{\cdot \delta}R_{\delta}(\boldsymbol{\beta}^*)\right)_{j},
\]
then $\mathbb{E}(T_{ij})=0$. Note that
\[
\nabla \bar{R}_{\delta}^{i}\left(\boldsymbol{\beta}^{*}\right)=w\left(y_{i}\right) \frac{y_{i} \boldsymbol{z}_{i}}{\delta} K\left(\frac{y_{i}\left(x_{i}-\boldsymbol{\beta}^{* T} \boldsymbol{z}_{i}\right)}{\delta}\right),
\]
hence
\[
\nabla_{\cdot \delta} \bar{R}_{\delta}^{i}\left(\boldsymbol{\beta}^{*}\right)
=-w(y_i)\frac{y_i\boldsymbol{z_i}}{\delta^2} K\left(\frac{x_{i}-\boldsymbol{\beta}^{* T} \boldsymbol{z}_{i}}{\delta}\right)-w(y_i)\frac{y_i\boldsymbol{z}_i}{\delta}K^\prime\left(\frac{x_{i}-\boldsymbol{\beta}^{* T} \boldsymbol{z}_{i}}{\delta}\right)\frac{x_{i}-\boldsymbol{\beta}^{* T} \boldsymbol{z}_{i}}{\delta^2},
\]
and
\begin{align*}
    & \mathbb{E}\left((\nabla_{\cdot \delta} \bar{R}_{\delta}^{i}(\boldsymbol{\beta}^*))_j^2\right)\\
    =&\sum\limits_{y}\int \left(
    w(y)\frac{yz_j}{\delta^2} K\left(\frac{x-\boldsymbol{\beta}^{* T} \boldsymbol{z}}{\delta}\right)+w(y)\frac{yz_j}{\delta}K^\prime\left(\frac{x-\boldsymbol{\beta}^{* T} \boldsymbol{z}}{\delta}\right)\frac{y\left(x-\boldsymbol{\beta}^{* T} \boldsymbol{z}\right)}{\delta^2}
    \right)^2f(x\mid y,\boldsymbol{z})dx f(y,\boldsymbol{z})d\boldsymbol{z}\\
    =&\sum\limits_{y}\int
    \left(
    w(y)^2\frac{z_j^2}{\delta^3}K^2(u)+
    w(y)^2\frac{z_j^2u^2}{\delta^3}K^{\prime 2}(u)+
    2w(y)^2y\frac{z_j^2u}{\delta^3}K(u)K^\prime(u)
    \right)
    f(u\delta + \boldsymbol{\beta}^{* T}\boldsymbol{z}\mid y,\boldsymbol{z})du f(y,\boldsymbol{z})d\boldsymbol{z}\\
    = & \sum\limits_{y}\int
    \left(
    w(y)^2\frac{z_j^2}{\delta^3}K^2(u)+
    w(y)^2\frac{z_j^2u^2}{\delta^3}K^{\prime 2}(u)+
    2w(y)^2y\frac{z_j^2u}{\delta^3}K(u)K^\prime(u)
    \right)\\
    &~~~~~~~\cdot
   \left(
    f( \boldsymbol{\beta}^{* T}\boldsymbol{z}\mid y,\boldsymbol{z})
    +
    u \delta f^\prime(\tau u\delta+\boldsymbol{\beta}^{* T}\boldsymbol{z}\mid y,\boldsymbol{z})
   \right)du f(y,\boldsymbol{z})d\boldsymbol{z},
\end{align*}
where $0<\tau<1$.
Note that
\[
\int K^2(u)du, \int |u|K^2(u)du, \int u^2 K^\prime(u)^2 du, \int |u^3| K^\prime(u)^2 du
\]
are bounded, hence $\mathbb{E}\left((\nabla_{\cdot \delta} \bar{R}_{\delta}^{i}(\boldsymbol{\beta}^*))_j^2\right) = \mathcal{O}(\frac{1}{\delta^3})$.
From (\ref{pf1}) we know that
\begin{equation}\label{ap2}
    \left(\nabla_{\cdot \delta}R_{\delta}(\boldsymbol{\beta}^*)\right)_j^2
=\mathcal{O}(\delta^{2\ell-2}).
\end{equation}
Therefore $\mathbb{E}T_{i j }^2=\mathcal{O}\left(\frac{1}{\delta^{3}}\right)$. Since $\left|T_{i j }\right| \lesssim \frac{M_{n}^{2}}{\delta^{3}}$, by Bernstein inequality with $M_{n}^{2} \sqrt{\frac{\log d}{n\delta^3}}=\mathcal{O}(1)$, we can show that with probability greater than $1-o(1)$
\begin{equation}\label{ap3}
    \|\nabla_{\cdot \delta} R_{\delta}^{n}(\boldsymbol{\beta}^*)-\nabla_{\cdot \delta}R_{\delta}(\boldsymbol{\beta}^*)\|_{\max} \lesssim \sqrt{\frac{\log d}{n\delta^3}}.
\end{equation}
Therefore with probability greater than $1-o(1)$,
\begin{align*}
    |\sqrt{n\delta^*} (\hat{\delta}-\delta^*)I_{2}|
    =&|\sqrt{n\delta^*} (\hat{\delta}-\delta^*)\left(\boldsymbol{v}^{* T}\nabla_{\cdot \delta} R_{\delta^*}^{n}(\boldsymbol{\beta}^*)
    -\mathbb{E}(\boldsymbol{v}^{* T}\nabla_{\cdot \delta} R_{\delta^*}^{n}(\boldsymbol{\beta}^*))\right)|\\
    \leq& |\sqrt{n\delta^*} (\hat{\delta}-\delta^*)|\|\boldsymbol{v}^*\|_1
     \|\nabla_{\cdot \delta} R_{\delta^*}^{n}(\boldsymbol{\beta}^*)-\nabla_{\cdot \delta}R_{\delta^*}(\boldsymbol{\beta}^*)\|_{\max} \\
     \lesssim&
     \sqrt{\log d}\|\boldsymbol{v}^*\|_1
     \left|\frac{\hat{\delta}-\delta^*}{\delta^{*}}\right|=o(1)
\end{align*}
under the conditions in Theorem~\ref{thm1_11}.

Now we analyze $I_3$. Write
\begin{align*}
    \sqrt{n\delta^*} (\hat{\delta}-\delta^*)I_{3} =& \sqrt{n\delta^*} \left(\hat{\boldsymbol{v}}^{ T}\nabla_{\cdot \delta} R_{\delta^*}^{n}(\hat{\boldsymbol{\beta}}_0)
    -\boldsymbol{v}^{* T}\nabla_{\cdot \delta} R_{\delta^*}^{n}(\boldsymbol{\beta}^*)\right)(\hat{\delta}-\delta^*)\\
     \leq &\sqrt{n\delta^*}(\hat{\delta}-\delta^*)|\boldsymbol{v}^{* T}\left(\nabla_{\cdot \delta} R_{\delta^*}^{n}(\hat{\boldsymbol{\beta}}_0)-\nabla_{\cdot \delta} R_{\delta^*}^{n}(\boldsymbol{\beta}^*)\right)|\\
    &~~+\sqrt{n\delta^*}(\hat{\delta}-\delta^*)|(\hat{\boldsymbol{v}}^{ T}-\boldsymbol{v}^{* T})\nabla_{\cdot \delta} R_{\delta^*}^{n}(\hat{\boldsymbol{\beta}}_0)|\\
    \coloneqq &I_{31}+I_{32}.
\end{align*}
By Lemma~\ref{lemma2_11} we have
\begin{align*}
   I_{31}&= \sqrt{n\delta^*}(\hat{\delta}-\delta^*)|\boldsymbol{v}^{* T}\nabla^2_{\cdot \delta} R_{\delta^*}^{n}(\boldsymbol{\beta}^*)(\hat{\boldsymbol{\beta}}_0-\boldsymbol{\beta}^*)|
+o_{\mathbb{P}}(1)\\
&\leq \sqrt{n\delta^*}(\hat{\delta}-\delta^*)\|\boldsymbol{v}^{*}\|_1\|\nabla^2_{\cdot \delta} R_{\delta^*}^{n}(\boldsymbol{\beta}^*)\|_{\max}\|\hat{\boldsymbol{\beta}}_0-\boldsymbol{\beta}^*\|_1+o_{\mathbb{P}}(1).
\end{align*}
Denote
\[
T_{ijk}=\left(\nabla^2_{\cdot \delta} \bar{R}_{\delta}^{i}(\boldsymbol{\beta}^*)-\nabla^2_{\cdot \delta}R_{\delta}(\boldsymbol{\beta}^*)\right)_{jk}.
\]
Note that
\begin{align*}
    \nabla^2_{\cdot \delta} \bar{R}_{\delta}^{i}\left(\boldsymbol{\beta}^{*}\right)
=2w(y_i)\frac{y_i\boldsymbol{z_i}\boldsymbol{z_i}^T}{\delta^3}K^\prime\left(\frac{x_{i}-\boldsymbol{\beta}^{* T} \boldsymbol{z}_{i}}{\delta}\right)+
w(y_i)K^{\prime\prime}\left(\frac{x_{i}-\boldsymbol{\beta}^{* T} \boldsymbol{z}_{i}}{\delta}\right)\frac{y_i(x_i-\boldsymbol{\beta}^{* T}\boldsymbol{z}_i)}{\delta^4}\boldsymbol{z}_i\boldsymbol{z}_i^T,
\end{align*}
and we have
\begin{align*}
    \nabla^2_{\cdot \delta}R_{\delta}(\boldsymbol{\beta}^*)_{jk}=&\sum_{y \in\{-1,1\}} w(y) y \int  z_jz_k \int K(u) u f^{\prime\prime}\left(u \delta+\boldsymbol{\beta}^{* T} \boldsymbol{z} \mid y, \boldsymbol{z}\right) d u f(y, \boldsymbol{z}) d \boldsymbol{z} \\
    = & \sum_{y \in\{-1,1\}} w(y) y \int  z_jz_k \int K(u) \frac{u^{\ell-1}\delta^{\ell-2}}{(\ell-2) !}\left(f^{(\ell)}\left(\boldsymbol{\beta}^{* T} \boldsymbol{z} \mid y, \boldsymbol{z}\right)+\mathcal{O}\left((u \delta)^{\zeta}\right)\right) d u f(y, \boldsymbol{z}) d \boldsymbol{z}\\
    = & \mathcal{O}(\delta^{\ell-2}),
\end{align*}
hence
\begin{align*}
    \left(\nabla^2_{\cdot \delta}R_{\delta}(\boldsymbol{\beta}^*)\right)_{jk}^2=\mathcal{O}(\delta^{2\ell-4}).
\end{align*}
Note that
\[
\int K^\prime(u)^2 du. \int |u|K^\prime(u)^2 du, \int u^2 K^{\prime\prime}(u)^2 du, \int |u^3| K^{\prime\prime}(u)^2 du
\]
are bounded, hence
\begin{align*}
    &     \mathbb{E}[ \left(\nabla^2_{\cdot \delta}\Bar{R}_{\delta}^i(\boldsymbol{\beta}^*)\right)_{jk}^2]\\
    =&\sum\limits_{y}\int \left(
    2w(y)\frac{y z_jz_k}{\delta^3}K^\prime\left(\frac{x-\boldsymbol{\beta}^{* T} \boldsymbol{z}}{\delta}\right)+
w(y)K^{\prime\prime}\left(\frac{x-\boldsymbol{\beta}^{* T} \boldsymbol{z}}{\delta}\right)\frac{y(x-\boldsymbol{\beta}^{* T}\boldsymbol{z})}{\delta^4}z_jz_k
    \right)^2f(x\mid y,\boldsymbol{z})dx f(y,\boldsymbol{z})d\boldsymbol{z}\\
    =&\sum\limits_{y}\int
    w(y)^2\frac{z_j^2z_k^2}{\delta^5}\left(
    4K^\prime(u)^2+
   u^2K^{\prime\prime}(u)^2+
    4uK^{\prime\prime}(u)K^\prime(u)
    \right)
    f(u\delta + \boldsymbol{\beta}^{* T}\boldsymbol{z}\mid y,\boldsymbol{z})du f(y,\boldsymbol{z})d\boldsymbol{z}\\
    = & \sum\limits_{y}\int
    w(y)^2\frac{z_j^2z_k^2}{\delta^5}\left(
    4K^\prime(u)^2+
   u^2K^{\prime\prime}(u)^2+
    4uK^{\prime\prime}(u)K^\prime(u)
    \right)\\
    &~~~\cdot
   \left(
    f( \boldsymbol{\beta}^{* T}\boldsymbol{z}\mid y,\boldsymbol{z})
    +
    u \delta f^\prime(\tau u\delta+\boldsymbol{\beta}^{* T}\boldsymbol{z}\mid y,\boldsymbol{z})
   \right)du f(y,\boldsymbol{z})d\boldsymbol{z}= \mathcal{O}(\frac{1}{\delta^5}).
\end{align*}
Therefore $\mathbb{E}T_{i jk }^2=\mathcal{O}\left(\frac{1}{\delta^{5}}\right)$.
Since $\left|T_{i jk }\right| \lesssim \frac{M_{n}^{3}}{\delta^{4}}$, applying Bernstein inequality with $\frac{M_{n}^{3}}{\delta}\sqrt{\frac{\log d}{n\delta}}=\mathcal{O}(1)$, we can show that with probability greater than $1-o(1)$
\[
\|\nabla^2_{\cdot \delta} R_{\delta}^{n}(\boldsymbol{\beta}^*)-\nabla^2_{\cdot \delta}R_{\delta}(\boldsymbol{\beta}^*)\|_{\max} \lesssim \sqrt{\frac{\log d}{n\delta^{5}}} .
\]
Therefore,
\begin{align*}
    \|\nabla^2_{\cdot \delta} R_{\delta}^{n}(\boldsymbol{\beta}^*)\|_{\max}
    \leq  &\|\nabla^2_{\cdot \delta} R_{\delta}^{n}(\boldsymbol{\beta}^*)-\nabla^2_{\cdot \delta}R_{\delta}(\boldsymbol{\beta}^*)\|_{\max} + \max\left(\nabla^2_{\cdot \delta}R_{\delta}(\boldsymbol{\beta}^*)\right)_{jk} \\
    =& \mathcal{O}(\delta^{\ell-2})+\mathcal{O}_{\mathbb{P}}(\sqrt{\frac{\log d}{n\delta^{5}}}),
\end{align*}
and with probability greater than $1-o(1)$
\begin{align*}
        I_{31}\leq &\sqrt{n\delta^*}(\hat{\delta}-\delta^*)\|\boldsymbol{v}^{*}\|_1\|\nabla^2_{\cdot \delta} R_{\delta^*}^{n}(\boldsymbol{\beta}^*)\|_{\max}\|\hat{\boldsymbol{\beta}}_0-\boldsymbol{\beta}^*\|_1+o_{\mathbb{P}}(1)
        \\
    \lesssim
    &\frac{\sqrt{n}\|\boldsymbol{v}^*\|_1C_{n,\delta^*}\delta^{* \ell}}{\sqrt{\delta^*}}\eta_1(n)+o_{\mathbb{P}}(1)=o_{\mathbb{P}}(1)
\end{align*}
under the conditions in Theorem~\ref{thm1_11}, where $\|\widehat{\boldsymbol{\beta}}-\boldsymbol{\beta}^{*}\|_{1} \lesssim \eta_{1}(n)$ is from Assumption \ref{ass_estimators}.


For $I_{32}$, by Lemma~\ref{lemma2_11} we have
\begin{align*}
    I_{32} &\leq \sqrt{n\delta^*}(\hat{\delta}-\delta^*)\|\hat{\boldsymbol{v}}^{ }-\boldsymbol{v}^{* }\|_1\|\nabla_{\cdot \delta} R_{\delta^*}^{n}(\hat{\boldsymbol{\beta}}_0)\|_{\max}\\
    & \leq
    \sqrt{n\delta^*}(\hat{\delta}-\delta^*)\|\hat{\boldsymbol{v}}-\boldsymbol{v}^{* }\|_1\left(\|\nabla_{\cdot \delta} R_{\delta^*}^{n}( \boldsymbol{\beta}^{*})\|_{\max}
+
\|\nabla^2_{\cdot \delta} R_{\delta^*}^{n}(\boldsymbol{\beta}^*)\|_{\max}\|\hat{\boldsymbol{\beta}}_0-\boldsymbol{\beta}^*\|_1+o_{\mathbb{P}}(1)\right)
\end{align*}
 In the proof for $I_{31}$ we have showed that $\sqrt{n\delta^*}(\hat{\delta}-\delta^*)\|\boldsymbol{v}^{*}\|_1\|\nabla^2_{\cdot \delta} R_{\delta^*}^{n}(\boldsymbol{\beta}^*)\|_{\max}\|\hat{\boldsymbol{\beta}}_0-\boldsymbol{\beta}^*\|_1 = o_{\mathbb{P}}(1)$.
By (\ref{ap2}) and (\ref{ap3}) we have $\|\nabla_{\cdot \delta} R_{\delta^*}^{n}( \boldsymbol{\beta}^{*})\|_{\max} = \mathcal{O}_{\mathbb{P}}(\delta^{* \ell-1})$, and from Assumption \ref{ass_estimators} we know that $\|\hat{\boldsymbol{v}}-\boldsymbol{v}^{* }\|_1 \lesssim \|\boldsymbol{v}^*\|_1\eta_2(n)$, hence we have
\begin{align*}
I_{32} = \mathcal{O}_{\mathbb{P}}(\sqrt{n\delta^*}\eta_2(n)\|\boldsymbol{v}^*\|_1(\hat{\delta}-\delta^*)
\delta^{* \ell-1}
)=
\mathcal{O}_{\mathbb{P}}(\sqrt{n\delta^*}\eta_2(n)\|\boldsymbol{v}^*\|_1
C_{n,\delta^*}
\delta^{* \ell}
)=
o_{\mathbb{P}}(1)
\end{align*}
under the conditions in Theorem~\ref{thm1_11}.
Combine the result we obtained for $I_1, I_2$ and $I_3$, the desired result holds.
\end{proof}

\section{More Details on Conditions and Assumptions}

\subsection{About existence and uniqueness of $\beta^*$}\label{sec:betaexistunique}
Recall that the risk function $R(\bbeta)$ is
$$
R(\bbeta)=\EE\big[w(Y)L_{01}\{Y(X - \bbeta^{T}\bZ)\},
$$
where $L_{01}(u)=\frac12\{1-\sign u\}$ is the 0-1 loss and $w(Y)$ is a known weight. Unfortunately, the function $R(\bbeta)$ is not always convex. The following is a counter-example. 
Suppose the distribution of $X$ given $Y$ is $X \mid Y \sim N(Y,1)$. For simplicity, we assume there is no covariate $\bZ$ and therefore we use a common threshold $\gamma$ to dichotomize $X$. Let the weight $w(1)=1/\pi, w(-1)=1/(1-\pi)$, where $\pi=\PP(Y=1)$, then the risk function becomes
\begin{align*}
    (1.2) &=R(\gamma)= \mathbb{E}\left[w(Y) L_{01}\left\{Y\left(X-\gamma\right)\right\}\right]\\
    &= \int_{-\infty}^\gamma f_{X\mid Y}(x \mid y=1) dx +
    \int_{\gamma}^\infty f_{X\mid Y}(x \mid y=-1) dx.
\end{align*}
Hence
\[
\nabla R(\gamma)= f_{X \mid Y}(\gamma \mid y=1)-f_{X \mid Y}(\gamma \mid y=-1)=\frac{1}{\sqrt{2\pi}}\left(e^{-\frac{(\gamma-1)^2}{2}}-e^{-\frac{(\gamma+1)^2}{2}}\right),
\]
\[
\nabla^2 R(\gamma)=\frac{1}{\sqrt{2\pi}}\left[-(\gamma-1)e^{-\frac{(\gamma-1)^2}{2}}+(\gamma+1)e^{-\frac{(\gamma+1)^2}{2}}\right].
\]
Note that for  $2<\gamma<3$, $\nabla^2 R(\gamma)<0$, hence in this case the risk function is not convex.

Given the fact that the risk function is generally not convex, a more subtle question is when $\bbeta^* = \argmin_{\bbeta}R(\bbeta)$  exists and is unique.
Before we elaborate this point, we first note that the existence and uniqueness of $\boldsymbol{\beta}^*$ is a standard identifiability assumption, commonly used in the M-estimation literature. Otherwise, the estimand $\bbeta^*$ is not well defined. In the following, we will discuss the existence and uniqueness of $\boldsymbol{\beta}^*$ respectively.

{\bf Existence}.  The existence of $\bbeta^*$ can be verified under more specific modeling assumptions. For example, \cite{manski1985semiparametric} considered the following binary response model
\begin{equation}\label{binary_model}
Y=\operatorname{sign}\left(X-\boldsymbol{Z}^{T} \boldsymbol{\beta}^*+\epsilon\right),
\end{equation}
where $\operatorname{Median}(\epsilon \mid X, \boldsymbol{Z})=0$. It is shown that the true coefficient $\bbeta^*$ is a minimizer of the risk function $R(\bbeta)$ with equal weight, and therefore the minimizer of the risk exists. Note that this is the data generating model used in our simulation studies. In addition, we can also show that the minimizer of the risk also exists under the logistic regression and linear discriminant analysis, we refer to \cite{feng2022nonregular} for more details.

{\bf Uniqueness}. Generally speaking, showing the uniqueness of the minimizer under a non-convex loss is a very challenging problem. In the following, we leverage the symmetric property of the loss to show the uniqueness of $\bbeta^*$ under the binary response model (\ref{binary_model}).

In this model, we assume $\epsilon \mid X,\boldsymbol{Z} \sim N(0,1)$, and $(X, \boldsymbol{Z})$ is multivariate normal with covariance matrix $\Sigma$ and mean 0. Suppose $w(1)=w(-1)=1$.
Note that
\[
f(x \mid y=-1, \boldsymbol{z})=\frac{ F_{\epsilon \mid x, \boldsymbol{z}}\left(\boldsymbol{\beta}^{* T} \boldsymbol{z}-x \mid x, \boldsymbol{z}\right) f_{x \mid \boldsymbol{z}}(x \mid \boldsymbol{z})}{\mathbb{P}(Y=-1 \mid \boldsymbol{Z}=\boldsymbol{z})},
\]
and
\[
f(x \mid y=1, \boldsymbol{z})=\frac{ F_{\epsilon \mid x, \boldsymbol{z}}\left(x-\boldsymbol{\beta}^{* T} \boldsymbol{z} \mid x, \boldsymbol{z}\right) f_{x \mid \boldsymbol{z}}(x \mid \boldsymbol{z})}{\mathbb{P}(Y=1 \mid \boldsymbol{Z}=\boldsymbol{z})},
\]
where $F_{\epsilon \mid x, \boldsymbol{z}}$ is the cdf of $\epsilon$ given $X,\bZ$. Hence we have
\begin{align*}
    \nabla R(\boldsymbol{\beta})= &\sum_{y \in\{-1,1\}} w(y) \int y z f\left(\boldsymbol{\beta}^{T} \boldsymbol{z} \mid y, \boldsymbol{z}\right) f(y, \boldsymbol{z}) d \boldsymbol{z} \\
    = & \int  \boldsymbol{z}  \left(F_{\epsilon \mid x, \boldsymbol{z}}\left(\boldsymbol{\beta}^T\boldsymbol{z}-\boldsymbol{\beta}^{* T} \boldsymbol{z} \mid x, \boldsymbol{z}\right) -
    F_{\epsilon \mid x, \boldsymbol{z}}\left(\boldsymbol{\beta}^{* T}\boldsymbol{z}-\boldsymbol{\beta}^{ T} \boldsymbol{z} \mid x, \boldsymbol{z}\right)
    \right) f_{x \mid \boldsymbol{z}}(\boldsymbol{\beta}^T\boldsymbol{z} \mid \boldsymbol{z})f(\boldsymbol{z}) d\boldsymbol{z}\\
    = &  \frac{1}{(2\pi)^{1+d/2} |\Sigma|^{1/2}}\int \boldsymbol{z} \int_{(\boldsymbol{\beta}^{*} - \boldsymbol{\beta})^T\boldsymbol{z} }^{(\boldsymbol{\beta}^{} - \boldsymbol{\beta}^*)^{ T}\boldsymbol{z}}\exp\left(-\frac{t^2}{2}\right)dt
    \exp \left(-\frac{(\boldsymbol{\beta}^T\boldsymbol{z},\boldsymbol{z})^{T}
    \Sigma^{-1}
     (\boldsymbol{\beta}^T\boldsymbol{z},\boldsymbol{z}) }{2}\right) d \boldsymbol{z}.
\end{align*}
Then by the mean value theorem, we have for any $\bbeta\neq\bbeta^*$ there exists some $0 < c < 1$,
\begin{align*}
      &R(\boldsymbol{\beta}) -   R(\boldsymbol{\beta}^*) \\
      =&  \nabla R(c\boldsymbol{\beta}+ (1-c)\boldsymbol{\beta}^*) (\boldsymbol{\beta}- \boldsymbol{\beta}^*)\\
      = & \frac{1}{(2\pi)^{1+d/2} |\Sigma|^{1/2}}
      \int  (\boldsymbol{\beta}- \boldsymbol{\beta}^*)^T \boldsymbol{z}
      \int_{c(\boldsymbol{\beta}^{*} - \boldsymbol{\beta})^T\boldsymbol{z} }^{c(\boldsymbol{\beta} - \boldsymbol{\beta}^*)^{ T}\boldsymbol{z}}\exp\left(-\frac{t^2}{2}\right)dt
    \exp \left(-\frac{(\Tilde{\boldsymbol{\beta}}^T\boldsymbol{z},\boldsymbol{z})^{T}
    \Sigma^{-1}
     (\Tilde{\boldsymbol{\beta}}^T\boldsymbol{z},\boldsymbol{z}) }{2}\right) d \boldsymbol{z}\\
     =& \frac{2}{(2\pi)^{1+d/2} |\Sigma|^{1/2}}
      \int_{(\boldsymbol{\beta}- \boldsymbol{\beta}^*)^T \boldsymbol{z}>0}  (\boldsymbol{\beta}- \boldsymbol{\beta}^*)^T \boldsymbol{z}
      \int_{c(\boldsymbol{\beta}^{*} - \boldsymbol{\beta})^T\boldsymbol{z} }^{c(\boldsymbol{\beta} - \boldsymbol{\beta}^*)^{ T}\boldsymbol{z}}\exp\left(-\frac{t^2}{2}\right)dt
    \exp \left(-\frac{(\Tilde{\boldsymbol{\beta}}^T\boldsymbol{z},\boldsymbol{z})^{T}
    \Sigma^{-1}
     (\Tilde{\boldsymbol{\beta}}^T\boldsymbol{z},\boldsymbol{z}) }{2}\right) d \boldsymbol{z}\\
     >& 0,
\end{align*}
where $\Tilde{\boldsymbol{\beta}} = c\boldsymbol{\beta}+ (1-c)\boldsymbol{\beta}^*$. The key is the third equality, where we use the fact that the above function is symmetry in $\bz$. Since we have $R(\boldsymbol{\beta})>R(\boldsymbol{\beta}^*)$ for any $\bbeta\neq\bbeta^*$, hence $\boldsymbol{\beta}^*$ is unique.

\subsection{About Nikol'ski class}\label{sec:nikolski}

Recall that the Nikol'ski class is defined as the set of functions $f: \RR \rightarrow \RR$ whose derivatives $f^{(\ell)}$ of order $\ell=\lfloor\beta\rfloor$ exist and satisfy
\[
\left[\int\left(f^{(\ell)}(x+t)-f^{(\ell)}(x)\right)^{2} d x\right]^{1 / 2} \leq L|t|^{\beta-\ell}, \quad \forall t \in \RR.
\]
After going through the proof carefully, we realize that we can relax the Hölder class condition defined in (3.1) to a Nikol'ski class condition. That is, we assume $f^{(\ell)}(x \mid y, \boldsymbol{z})$ satisfies the following Nikol'ski class constraint:
\begin{equation} \label{nikol}
    \left[\int\left(f^{(\ell)}(\boldsymbol{\beta}^{* T}\boldsymbol{z}+\triangle_{\boldsymbol{z}} \mid y, \boldsymbol{z})-f^{(\ell)}(\boldsymbol{\beta}^{* T}\boldsymbol{z} \mid y, \boldsymbol{z})\right)^{2} d \boldsymbol{z}\right]^{1 / 2} \leq L \triangle^{\zeta},
\end{equation}
for any $|\triangle_{\boldsymbol{z}}| \leq \triangle$, where  $\triangle_{\boldsymbol{z}}$ may depend on $\boldsymbol{z}$ and $L$ is a constant. The condition (\ref{nikol}) looks a bit unconventional, because we need to allow $\triangle_{\boldsymbol{z}}$ to depend on $\boldsymbol{z}$. The reason is as follows. In the proof of Lemma \ref{lemma_consistent_bias_pilot} on the consistency of the bias estimate, a key step is to bound $\|T^{(\ell)}\left(\boldsymbol{\beta}^{*}\right)-T^{(\ell)}(\widehat{\boldsymbol{\beta}}^{(2)})\|_{\infty}$, which is
\begin{align*}
    &\left\|T^{(\ell)}\left(\boldsymbol{\beta}^{*}\right)-T^{(\ell)}\left(\widehat{\boldsymbol{\beta}}^{(2)}\right)\right\|_{\infty} \\
=&\left\|\sum_{y \in\{-1,1\}} w(y) \int y \boldsymbol{z}\left(f^{(\ell)}\left(\boldsymbol{\beta}^{* T} \boldsymbol{z} \mid y, \boldsymbol{z}\right)-f^{(\ell)}\left(\widehat{\boldsymbol{\beta}}^{(2) T} \boldsymbol{z} \mid y, \boldsymbol{z}\right)\right) f(y, \boldsymbol{z}) d \boldsymbol{z}\right\|_{\infty} \\
\leq &
\max_j\sum_{y \in\{-1,1\}} w(y)
\sqrt{\int z_j^2 f^2(y,\boldsymbol{z})d\boldsymbol{z}}
\sqrt{\int \left(f^{(\ell)}\left(\boldsymbol{\beta}^{* T} \boldsymbol{z} \mid y, \boldsymbol{z}\right)-f^{(\ell)}\left(\widehat{\boldsymbol{\beta}}^{(2) T} \boldsymbol{z} \mid y, \boldsymbol{z}\right)\right)^2 d\boldsymbol{z}}
\\
\leq &
\max_j
\sum_{y \in\{-1,1\}} w(y)
\sqrt{\int z_j^2 f^2(y,\boldsymbol{z})d\boldsymbol{z}}\cdot
L|M_{n} \eta_{1}(n)|^\zeta
\\
=& \mathcal{O}_{\mathbb{P}}\left( \left(M_{n} \eta_{1}(n)\right)^{\zeta} \right),
\end{align*}
where the first inequity is from Cauchy inequality, and the second one is from the Nikol'ski condition (\ref{nikol}) with $\triangle_{\boldsymbol{z}}=(\hat\bbeta^{(2)}-\bbeta^*)^T\bz$ and $\triangle=CM_{n} \eta_{1}(n)$ as $\|\hat\bbeta^{(2)}-\bbeta^*\|_1\lesssim \eta_1(n)$ and $\|\bz\|_\infty\leq M_n$. In this case, we have to allow the increment  $\triangle_{\boldsymbol{z}}$ to depend on $\bz$ to make the proof work. Since in (\ref{nikol}) the integration is over $\bz$, the right hand of (\ref{nikol}) cannot depend on $\bz$ and therefore we introduce an upper bound $\triangle$ with $|\triangle_{\boldsymbol{z}}| \leq \triangle$ in (\ref{nikol}).

In conclusion, our theory does work under a more relaxed Nikol'ski class condition (\ref{nikol}). However, since (\ref{nikol}) deviates from the standard  Nikol'ski condition in a very subtle way, to avoid any confusion, we decide to keep the Hölder class in the main paper and include the above discussions on the Nikol'ski condition in the supplementary materials.

\subsection{Examples of kernel functions}\label{sec:kernel}

In Table~\ref{tb:kernel} below, we provide a list of commonly-seen second-order, fourth-order and sixth-order kernel functions.

\begin{table}[!htbp]
	\centering
\caption{Commonly-seen second-order, fourth-order and sixth-order kernel functions $K(\cdot)$.}\label{tb:kernel}
	\begin{tabular}{c|l|l}
        \hline
    & Kernel & Equation $K(\cdot)$ \\
    \hline
\multirow{5}{*}{Second-Order} & Uniform & $K_0(t)=\frac12 I(|t|\leq 1)$\\
    & Epanechnikov & $K_1(t)=\frac34 (1-t^2) I(|t|\leq 1)$\\
    & Biweight & $K_2(t)=\frac{15}{16}(1-t^2)^2 I(|t|\leq 1)$\\
    & Triweight & $K_3(t)=\frac{35}{32}(1-t^2)^3 I(|t|\leq 1)$\\
    & Gaussian & $K_\phi(t)=\frac1{\sqrt{2\pi}} \exp(-\frac{t^2}{2})$\\
    \hline
\multirow{4}{*}{Fourth-Order} & Epanechnikov & $K_{4,1}(t)=\frac{15}{8} \left(1-\frac{7}{3}t^2\right) K_1(t)$\\
    & Biweight & $K_{4,2}(t) = \frac{7}{4}(1-3t^2) K_2(t)$\\
    & Triweight & $K_{4,3}(t) = \frac{27}{16}\left(1-\frac{11}{3}t^2\right) K_3(t)$\\
    & Gaussian & $K_{4,\phi}(t) = \frac12(3-t^2)K_\phi(t)$\\
    \hline
\multirow{4}{*}{Sixth-Order} & Epanechnikov & $K_{6,1}(t)=\frac{175}{64}\left(1-6t^2+\frac{33}{5}t^4\right)K_1(t)$\\
    & Biweight & $K_{6,2}(t)=\frac{315}{128}\left( 1-\frac{22}{3}t^2+\frac{143}{15}t^4 \right)K_2(t)$\\
    & Triweight & $K_{6,2}(t)=\frac{297}{128}\left(1-\frac{26}{3}t^2+13t^4 \right)K_3(t)$\\
    & Gaussian & $K_{6,\phi}(t)=\frac18 (15-10t^2+t^4)K_\phi(t)$\\
		\hline
	\end{tabular}
\end{table}

\subsection{Verifying Assumptions under Binary Response Models}\label{app_verify_assumption}

In the section, we articulate the assumptions used in Section \ref{theory} under the binary response model. The following lemma is our main result.

\begin{lemma}\label{lem_verify_assumptions}
Consider the following binary response model   $Y=\operatorname{sign}\left(X-\boldsymbol{\beta}^{* T} \boldsymbol{Z}+\epsilon\right)$, where we assume $(X, \boldsymbol{Z})$ has mean $0$ and $\epsilon \sim N\left(0, \sigma^{2}\right)$. Assume that the smallest and largest eigenvalues of $\bSigma_Z:=\Cov(\bZ)$ are bounded away from 0 and infinity by some constants. Under either of the following two sets of conditions: for some constants $C_1$ and $C_2$, (1) $\|\bbeta^*\|_1\leq C_1$, $X \mid \boldsymbol{Z}\sim N(\bc^T\bZ, \sigma_X^2)$, with $\|\bc\|_1\leq C_1$ and $Z_j$ has a bounded support uniformly over $j$, (2) $\|\bbeta^*\|_2\leq C_2$ and $X \sim N(0, \sigma_X^2)$ independent of $\bZ\sim N(0,\bSigma_Z)$, our Assumptions \ref{ass_smooth}-\ref{ass_projection_norm} hold.
\end{lemma}

For the convergence rate of $\hat\bbeta$ in Assumption \ref{ass_estimators}, it is established by \cite{feng2022nonregular}. Specifically, Assumption 3.1 and 3.2 in \cite{feng2022nonregular} are the same as Assumption 1 and Assumption 4 in this paper. Verification of Assumption 3.3 and 3.4 for binary response model is shown in supplementary material for \cite{feng2022nonregular} section S.4 and S.3 respectively. Finally, we note that the convergence rate of $\hat\bv$ in Assumption \ref{ass_estimators} is established in Lemma \ref{theorem_projection} in the supplementary material under Assumptions \ref{ass_smooth}-\ref{ass_projection_norm}, and $\lambda_{\min}(\nabla^2_{\bgamma,\bgamma}R(\bbeta^*))\geq c$ for some constant $c>0$. The latter condition is an intermediate step in the proof of Lemma \ref{lem_verify_assumptions} and therefore also holds.

\begin{proof}
Denote $f_{X \mid \boldsymbol{Z}}(x \mid \boldsymbol{z})$ as the Gaussian p.d.f of $X$ given $\boldsymbol{Z}$, $\Phi(\cdot)$ as the c.d.f of the standard Gaussian distribution. Note that
\[
\mathbb{P}(Y=1 \mid X=x, \boldsymbol{Z}= \boldsymbol{z})
= \Phi\left(\frac{x-\boldsymbol{\beta}^{* T}\boldsymbol{z}}{\sigma}\right),
\]
hence
\begin{equation}\label{x|yz}
    f_{X \mid Y, \boldsymbol{Z}}(x \mid y=1, \boldsymbol{z})
=  \frac{\Phi\left(\frac{x-\boldsymbol{\beta}^{* T}\boldsymbol{z}}{\sigma}\right)f_{X \mid \boldsymbol{Z}}(x \mid \boldsymbol{z})}{\mathbb{P}(Y=1 \mid \boldsymbol{Z}=\boldsymbol{z})}.
\end{equation}
As Gaussian p.d.f is infinitely differentiable, for any integer $l \geq 1$, with some calculation we have
\[
f^{(\ell)}_{X \mid Y, \boldsymbol{Z}}(x \mid y=1, \boldsymbol{z})
=
\frac{\sum_{k=0}^{\ell}\left(\begin{array}{c}\ell \\ k\end{array}\right) \Phi^{(k)}\left(\frac{x-\boldsymbol{\beta}^{* T}\boldsymbol{z}}{\sigma}\right) f_{X \mid \boldsymbol{Z}}^{(\ell-k)}(x \mid \boldsymbol{z})}{\mathbb{P}(Y=1 \mid \boldsymbol{Z}=\boldsymbol{z})}.
\]

Note that the derivatives of Gaussian p.d.f are continuous and bounded, hence for any fixed integer $\ell \geq 2$, they are uniformly bounded. Therefore,  $\left|\Phi^{(i)}\left(\frac{x-\boldsymbol{\beta}^{* T}\boldsymbol{z}}{\sigma}\right)\right|,\left|f_{X \mid \boldsymbol{Z}}^{(i)}(x \mid \boldsymbol{z})\right|$ are bounded above by some constant for almost all $\boldsymbol{z}$ for $i=0, \ldots, l$. In addition, we can show that $\mathbb{P}(Y=1)=0.5$ and
\[
\mathbb{P}(Y=1 \mid \boldsymbol{Z}=\boldsymbol{z})
=\int_X \Phi\left(\frac{x-\boldsymbol{\beta}^{* T}\boldsymbol{z}}{\sigma}\right)f_{X \mid \boldsymbol{Z}}(x \mid \boldsymbol{z}) dx,
\]

Let us first consider case (1). When $\|\bbeta^*\|_1, \|\bc\|_1\leq C_1$ and $Z_j$ has a bounded support, we know that $|\bbeta^{*T}\bz|, |\bc^T\bz|\leq C$ for some constant $C$. Thus,
$$
\mathbb{P}(Y=1 \mid \boldsymbol{Z}=\boldsymbol{z})\geq \int_X \Phi\left(\frac{x-C}{\sigma}\right)f_{X \mid \boldsymbol{Z}}(x \mid \boldsymbol{z}) dx\geq \frac{1}{2} \PP(X>C|\bZ=\bz),
$$
which is lower bounded by a positive constant for any $\bz$. This implies that
$f(x \mid y, \boldsymbol{z})$ is $\ell$ smooth, and Assumptions \ref{ass_smooth} and \ref{ass_proportion} are satisfied. Since $Z_j$ is bounded, Assumption \ref{ass_moment} is also true.

To show that $f(x \mid y, \bz)$ is $\ell$ smooth under case (2), consider  $X\sim N(0, 1)$ without loss of generality. Denote
 $\phi(\cdot)$ as the p.d.f of the standard Gaussian distribution.
Then
\[
f^{(\ell)}_{X \mid Y, \boldsymbol{Z}}(x \mid y=1, \boldsymbol{z})
=
\frac{\sum_{k=0}^{\ell}\left(\begin{array}{c}\ell \\ k\end{array}\right) \Phi^{(k)}\left(x-\boldsymbol{\beta}^{* T}\boldsymbol{z}\right)\phi^{(\ell-k)}(x) }{
\int \Phi\left(x-\boldsymbol{\beta}^{* T}\boldsymbol{z}\right)\phi(x) dx
}.
\]
Note that $\Phi\left(x-\boldsymbol{\beta}^{* T}\boldsymbol{z}\right)=\frac{1}{\sqrt{2\pi}}\int_{-\infty}^{x-\boldsymbol{\beta}^{* T}\boldsymbol{z}}e^{-t^2/2}dt$, $\Phi^{(k)}\left(x-\boldsymbol{\beta}^{* T}\boldsymbol{z}\right) = g(x-\boldsymbol{\beta}^{* T}\boldsymbol{z})\exp(-\frac{(x-\boldsymbol{\beta}^{* T}\boldsymbol{z})^2}{2})$, where
$g(\cdot)$ is some polynomial function. Hence for any integer $0 \leq k \leq \ell$,
\begin{align*}
    \frac{\partial (\Phi^{(k)}\left(x-\boldsymbol{\beta}^{* T}\boldsymbol{z}\right)) }{\partial(\boldsymbol{\beta}^{* T}\boldsymbol{z})}=&
\frac{\partial g(x-\boldsymbol{\beta}^{* T}\boldsymbol{z})}{\partial(\boldsymbol{\beta}^{* T}\boldsymbol{z})}\exp(-\frac{(x-\boldsymbol{\beta}^{* T}\boldsymbol{z})^2}{2})+
g(x-\boldsymbol{\beta}^{* T}\boldsymbol{z})(x-\boldsymbol{\beta}^{* T}\boldsymbol{z})\exp(-\frac{(x-\boldsymbol{\beta}^{* T}\boldsymbol{z})^2}{2})\\
:=& h(x-\boldsymbol{\beta}^{* T}\boldsymbol{z})\exp(-\frac{(x-\boldsymbol{\beta}^{* T}\boldsymbol{z})^2}{2}),
\end{align*}
where $h(\cdot)$ is also some polynomial function.
Therefore,
\begin{align}\label{lim1}
    \lim_{|\boldsymbol{\beta}^{* T}\boldsymbol{z}|\rightarrow \infty}
    \frac{ \Phi^{(k)}\left(x-\boldsymbol{\beta}^{* T}\boldsymbol{z}\right)\phi^{(\ell-k)}(x) }{
\int \Phi\left(x-\boldsymbol{\beta}^{* T}\boldsymbol{z}\right)\phi(x) dx
}= \lim_{|\boldsymbol{\beta}^{* T}\boldsymbol{z}|\rightarrow \infty}
\frac{
h(x-\boldsymbol{\beta}^{* T}\boldsymbol{z})\exp(-\frac{(x-\boldsymbol{\beta}^{* T}\boldsymbol{z})^2}{2})\phi^{(\ell-k)}(x)
}
{-\int \exp{(-\frac{(x-\boldsymbol{\beta}^{* T}\boldsymbol{z})^2+x^2}{2})}dx}.
\end{align}
Note that
\begin{align*}
    &\int \exp{(-\frac{(x-\boldsymbol{\beta}^{* T}\boldsymbol{z})^2+x^2}{2})}dx\\
    =&
\frac{1}{\sqrt{2}}\exp{(-\frac{1}{4}(\boldsymbol{\beta}^{* T}\boldsymbol{z})^2)}\int
\exp{-\frac{(y-\frac{1}{\sqrt{2}}\boldsymbol{\beta}^{* T}\boldsymbol{z})^2}{2}}dy=
C\exp{(-\frac{1}{4}(\boldsymbol{\beta}^{* T}\boldsymbol{z})^2)}
\end{align*}
for some constant $C$. Hence,
\begin{align*}
    (\ref{lim1})= C_1 \lim_{|\boldsymbol{\beta}^{* T}\boldsymbol{z}|\rightarrow \infty}
    h(x-\boldsymbol{\beta}^{* T}\boldsymbol{z})\phi^{(\ell-k)}(x) \exp(-\frac{(\boldsymbol{\beta}^{* T}\boldsymbol{z}-2x)^2}{4})\exp(\frac{x^2}{2})=0.
\end{align*}
This implies that $\lim\limits_{|\boldsymbol{\beta}^{* T}\boldsymbol{z}|\rightarrow \infty}f^{(\ell)}_{X \mid Y, \boldsymbol{Z}}(x \mid y=1, \boldsymbol{z})=0$, and we already showed that when $Z_j$ has a bounded support, $f^{(\ell)}_{X \mid Y, \boldsymbol{Z}}(x \mid y=1, \boldsymbol{z})$ is bounded. Therefore, Assumption \ref{ass_smooth} is satisfied under case (2).

Finally, we focus on Assumption \ref{ass_projection_norm}. Recall that
\begin{align*}
    \boldsymbol{\Sigma}^{*}&:=\sum_{y \in\{-1,1\}} w(y)^{2} \int \boldsymbol{z} \boldsymbol{z}^{T} \int K(u)^{2} d u f\left(\boldsymbol{\beta}^{* T} \boldsymbol{z} \mid y, \boldsymbol{z}\right) f(y, \boldsymbol{z}) d \boldsymbol{z}\\
    &= \int K(u)^{2} d u
    \int_{\boldsymbol{Z}} \boldsymbol{z} \boldsymbol{z}^{T} f_{X \mid \boldsymbol{Z}}(\boldsymbol{\beta}^{* T} \boldsymbol{z}  \mid \boldsymbol{z})
    f(\boldsymbol{z}) d \boldsymbol{z}.
\end{align*}
Note that $0<\int K(u)^{2} d u < \infty$ and $X\mid \boldsymbol{Z}$ is Gaussian,  for some constants $0 < c_1 < C_1 < \infty$,  we have
\[
c_1<f_{X \mid \boldsymbol{Z}}(\boldsymbol{\beta}^{* T} \boldsymbol{z}  \mid \boldsymbol{z}) < C_1,
\]
under case (1).

Since the smallest and largest eigenvalues of $\bSigma_Z=\Cov(\bZ)$ are bounded away from 0 and infinity by some constants, if we can further show $\|\bv^*\|_2$ is also bounded from above and below, then we have $\sigma^{*}=\sqrt{\boldsymbol{v}^{* T} \boldsymbol{\Sigma}^{*} \boldsymbol{v}^{*}}$ is bounded away from $0$ and $\infty$.

Recall that $\boldsymbol{v}^{*}=\left(1,-\boldsymbol{\omega}^{* T}\right)^{T}$. So, clearly $\|\bv\|_2\geq 1$. In the following, we focus on the upper bound of $\bomega^*$. Recall that $\boldsymbol{\omega}^* = \left(\nabla_{\boldsymbol{\gamma}, \boldsymbol{\gamma}}^{2} R\left(\boldsymbol{\beta}^{*}\right)\right)^{-1} \nabla_{\boldsymbol{\gamma}, \theta}^{2} R\left(\boldsymbol{\beta}^{*}\right)$. We have
 \[
\frac{ \|\boldsymbol{\omega}^*\|_2^2}{\|\nabla_{\boldsymbol{\gamma}, \theta}^{2} R\left(\boldsymbol{\beta}^{*}\right)\|_2^2}\leq \lambda_{\max}^2\left(\left(\nabla_{\boldsymbol{\gamma}, \boldsymbol{\gamma}}^{2} R\left(\boldsymbol{\beta}^{*}\right)\right)^{-1}\right)\leq \lambda_{\max}^2\left((\nabla^{2} R(\boldsymbol{\beta}^*) )^{-1}\right).
 \]
Note that
\begin{align*}
\nabla^{2} R(\boldsymbol{\beta}^*) &=\sum_{y=\pm 1} w(y) \int_{\boldsymbol{Z}} \boldsymbol{z} \boldsymbol{z}^{T} y f^{\prime}\left(\boldsymbol{\beta}^{* T} \boldsymbol{z} \mid \boldsymbol{z}, y\right) f(\boldsymbol{z}, y) d \boldsymbol{z} \\
&=2 \int_{\boldsymbol{Z}} \boldsymbol{z} \boldsymbol{z}^{T}f_{\epsilon \mid x, \boldsymbol{z}}(0\mid x=\boldsymbol{\beta}^{* T}  \boldsymbol{z}, \boldsymbol{z} )f_{x \mid \boldsymbol{z}}(\boldsymbol{\beta}^{* T}  \boldsymbol{z})
f(\boldsymbol{z}) d \boldsymbol{z}\\
&=2f_{\epsilon}(0) \int_{\boldsymbol{Z}} \boldsymbol{z} \boldsymbol{z}^{T}f_{X \mid \boldsymbol{Z}}(\boldsymbol{\beta}^{* T}  \boldsymbol{z})
f(\boldsymbol{z}) d \boldsymbol{z},
\end{align*}
hence together with the eigenvalue condition on  $\operatorname{Cov}(\boldsymbol{Z})$, both $\lambda_{\max}^2\left((\nabla^{2} R(\boldsymbol{\beta}^*) )^{-1}\right)$ and $\nabla_{\boldsymbol{\gamma}, \theta}^{2} R\left(\boldsymbol{\beta}^{*}\right)$ are bounded above by some constant.
Therefore,
\[
\|\boldsymbol{\omega}^*\|_2^2\leq \|\nabla_{\boldsymbol{\gamma}, \theta}^{2} R\left(\boldsymbol{\beta}^{*}\right)\|_2^2 \lambda_{\max}^2\left((\nabla^{2} R(\boldsymbol{\beta}^*) )^{-1}\right)\leq C
\]
for some constant $C$.

For case (2), note that
\begin{align*}
   &\int_{\boldsymbol{Z}} \boldsymbol{z} \boldsymbol{z}^{T} f_{X \mid \boldsymbol{Z}}(\boldsymbol{\beta}^{* T} \boldsymbol{z}  \mid \boldsymbol{z})
    f(\boldsymbol{z}) d \boldsymbol{z}\\
    =&
     \int_{\boldsymbol{Z}} \boldsymbol{z} \boldsymbol{z}^{T}
     \frac{1}{(2\pi)^{1/2}}\frac{1}{|2\pi\bSigma_{\bZ}|^{1/2}}
     \exp\left(-\frac{(\boldsymbol{\beta}^{* T} \boldsymbol{z})^2  }{2}\right)
     \exp \left(-\frac{\boldsymbol{z}^{T}\bSigma_{\bZ}^{-1} \boldsymbol{z}}{2}\right)
      d \boldsymbol{z}\\
      =&
      \int_{\boldsymbol{Z}} \boldsymbol{z} \boldsymbol{z}^{T}
     \frac{1}{(2\pi)^{1/2}}\frac{1}{|2\pi\bSigma_{\bZ}|^{1/2}}
     \exp\left(
     -\frac{\boldsymbol{z}^{T}\left(\bSigma_{\bZ}^{-1}+\boldsymbol{\beta}^{*} \boldsymbol{\beta}^{*T} \right) \boldsymbol{z}}{2}
     \right)d\bz\\
     =&  \frac{1}{(2\pi)^{1/2}}\frac{\left(\bSigma_{\bZ}^{-1}+\boldsymbol{\beta}^{*} \boldsymbol{\beta}^{*T } \right)^{-1}}{|\bSigma_{\bZ}|^{1/2}|\bSigma_{\bZ}^{-1}+\boldsymbol{\beta}^{*}\boldsymbol{\beta}^{*T }|^{1/2}},
\end{align*}
where in the last step we use the integral of a normal distribution with variance $\left(\bSigma_{\bZ}^{-1}+\boldsymbol{\beta}^{*}\boldsymbol{\beta}^{* T} \right)^{-1}$. Using the assumption that  $\|\bbeta^*\|_2$ is bounded, we can similarly prove that $\sigma^{*}=\sqrt{\boldsymbol{v}^{* T} \boldsymbol{\Sigma}^{*} \boldsymbol{v}^{*}}$ is bounded away from $0$ and $\infty$.

For the bias $\mu^*$ in Assumption \ref{ass_projection_norm}, write
\begin{align*}
|\mu^{*}|&:=|\boldsymbol{v}^{* T} \boldsymbol{b}^{*}| =\Big|\boldsymbol{v}^{* T}\left(\int K(u) \frac{u^{\ell}}{\ell !} d u\right) \sum_{y \in\{-1,1\}} w(y) \int y \boldsymbol{z} f^{(\ell)}\left(\boldsymbol{\beta}^{* T} \boldsymbol{z} \mid y, \boldsymbol{z}\right) f(y, \boldsymbol{z}) d \boldsymbol{z}\Big| \\
&=2\Big|\boldsymbol{v}^{* T}
\left(\int K(u) \frac{u^{\ell}}{\ell !} d u\right) \cdot \int_{\boldsymbol{Z}}\boldsymbol{z}\Big(
f^{(\ell)}\left(\boldsymbol{\beta}^{* T} \boldsymbol{z} \mid y=1,\boldsymbol{z}
\right)\mathbb{P}(Y=1 \mid \boldsymbol{Z}=\boldsymbol{z})\\
&~~~~-
f^{(\ell)}\left(\boldsymbol{\beta}^{* T} \boldsymbol{z} \mid y=-1, \boldsymbol{z}
\right)\mathbb{P}(Y=-1 \mid \boldsymbol{Z}=\boldsymbol{z})
\Big)f(\boldsymbol{z}) d \boldsymbol{z}\Big|\\
&\lesssim \Big\{\int (\bv^{*T}\bz)^2f(\bz)d\bz\Big\}^{1/2},
\end{align*}
where we use the fact that $f^{(\ell)}(x \mid y, \boldsymbol{z})$ is bounded, and the last step is from the Cauchy-Schwarz inequality. Finally, from our previous proof, $\int (\bv^{*T}\bz)^2f(\bz)d\bz$ is bounded by a constant. This completes the proof.
\end{proof}

\subsection{About the magnitude of $w^*$ and $v^*$}\label{sec:wandv}

{
In general, we make assumptions on the magnitude of $\boldsymbol{\omega}^{*}$, such as $\boldsymbol{\omega}^{*}$ is sparse or approximately sparse, in order to achieve the consistency of the estimator $\hat\bomega$ and, subsequently, the validity of the debias/decorrelated inference for high-dimensional models.
Note that $\boldsymbol{v}^*=(1,-\boldsymbol{\omega}^{*T})^T$, thus the discussions on the magnitude of $\boldsymbol{\omega}^*$ and $\boldsymbol{v}^*$ are equivalent.
}

Usually, the sparsity of $\boldsymbol{\omega}^{*}$ can be verified under some additional independence assumption among covariates. In the following, we demonstrate this point under the binary response model.
Consider the binary response model $Y=\operatorname{sign}\left(X-\boldsymbol{\beta}^{* T} \boldsymbol{Z}+u \right)$. Assume  $w(y) \equiv 1$, $u$ is independent of $X,\bZ$ and $(X,\bZ)\sim N(0,\bSigma)$. Denote $f_{x \mid \boldsymbol{z}}(\cdot)$ as the p.d.f of $X$ given $\boldsymbol{Z}$ and $f_{u}(\cdot)$ as the p.d.f of $u$.
Note that
\begin{align*}
\nabla^{2} R(\boldsymbol{\beta}^*) &=2 \int_{\boldsymbol{Z}} \boldsymbol{z} \boldsymbol{z}^{T}f_{u}(0)f_{x \mid \boldsymbol{z}}(\boldsymbol{\beta}^{* T}  \boldsymbol{z})
f(\boldsymbol{z}) d \boldsymbol{z}.
\end{align*}
Our first assumption is that, given a subset of $\bZ$ denoted by $\bZ_{S_1}$, $X$ is independent of the rest of the covariates.  Since $(X,\bZ)\sim N(0,\bSigma)$, this assumption simply says the first row of the precision matrix $\bSigma^{-1}$ is sparse. Recall that we denote the support set of $\bbeta^*$ by $S$. Thus, the function $f_{x \mid \boldsymbol{z}}(\boldsymbol{\beta}^{* T}  \boldsymbol{z})$ only depends on $\bZ_{S_2}$, where $S_2=S_1\cup S$.

Our second assumption is that $\bZ$ is blockwise independent. To be precise, we assume there exists a small set $S_3$ containing $Z_1$ and $\bZ_{S_2}$ defined above such that $\bZ_{S_3}\perp \bZ_{S_3^c}$. Intuitively, we can find $S_3$ by merging all the blocks that contain variables in $Z_1$ and $\bZ_{S_2}$.

Then by rearranging $\bZ$ we have
\begin{align}
&\nabla^{2} R(\boldsymbol{\beta}^*) \nonumber\\
= &2f_{u}(0) \int_{\boldsymbol{Z}}
\left(\begin{array}{ll}
 \bz_{S_3}\bz_{S_3}^T & \bz_{S_3} \bz_{S_3^c}^T \\
 \bz_{S_3^c}\bz_{S_3}^T
& \bz_{S_3^c}\bz_{S_3^c}^T
\end{array}\right)
\cdot f_{x |z}(\boldsymbol{\beta}^{* T}  \boldsymbol{z})
f(\bz_{ S_3})f(\bz_{S_3^c}) d \bz_{ S_3} d \bz_{S_3^c}\nonumber\\
= &2f_{u}(0) \int_{\boldsymbol{Z}}
\left(\begin{array}{ll}
 \bz_{S_3}\bz_{S_3}^T & 0 \\
 0
& \bz_{S_3^c}\bz_{S_3^c}^T
\end{array}\right)
\cdot f_{x |z}(\boldsymbol{\beta}^{* T}  \boldsymbol{z})
f(\bz_{ S_3})f(\bz_{S_3^c}) d \bz_{ S_3} d \bz_{S_3^c},\label{eq_verify_sparse}
\end{align}
where the off-diagonal block is 0 because $\EE(\bZ_{S_3^c})=0$. Since $\nabla^{2} R(\boldsymbol{\beta}^*)$ is a block-wise diagonal matrix, $\left(\nabla^{2} R(\boldsymbol{\beta}^*)\right)^{-1}$ is also block-wise diagonal.

 Note that  previously we defined $\theta$ as the coefficient of $\bZ_1$, $\bgamma$ as the  coefficient of the rest components of $\bZ$, and we have  $\bomega^* =(\nabla_{\bgamma,\bgamma}^{2} R(\boldsymbol{\beta}^*) )^{-1}\nabla_{\theta,\bgamma}^{2} R(\boldsymbol{\beta}^*) $.
 From the block matrix inversion we know that
\begin{align}
       \nabla^{2} R(\boldsymbol{\beta}^*)^{-1} =&
\left(\begin{array}{ll}
\nabla_{\theta,\theta}^{2} R(\boldsymbol{\beta}^*)  & \nabla_{\theta,\bgamma}^{2} R(\boldsymbol{\beta}^*)  \\
\nabla_{\bgamma,\theta}^{2} R(\boldsymbol{\beta}^*)  & \nabla_{\bgamma,\bgamma}^{2} R(\boldsymbol{\beta}^*)
\end{array}\right)^{-1}\\
=&
\left(\begin{array}{ll}
~~~~~~~~~~~~~* & -\frac{1}{k} \nabla_{\bgamma,\bgamma}^{2} R(\boldsymbol{\beta}^*)^{-1}\nabla_{\theta,\bgamma}^{2} R(\boldsymbol{\beta}^*)  \\
 -\frac{1}{k}\nabla_{\bgamma,\theta}^{2} R(\boldsymbol{\beta}^*)\nabla_{\bgamma,\bgamma}^{2} R(\boldsymbol{\beta}^*)^{-1}    &
~~~~~~~~~~~~~**
\end{array}\right) \\
=&
\left(\begin{array}{ll}
~~~* & -\frac{1}{k} \bomega^*  \\
 -\frac{1}{k} \bomega^{*T}  & ~~~**
\end{array}\right),
\end{align}
 where $k = \nabla_{\theta,\theta}^{2} R(\boldsymbol{\beta}^*)- \nabla_{\bgamma,\theta}^{2} R(\boldsymbol{\beta}^*)  \nabla_{\bgamma,\bgamma}^{2} R(\boldsymbol{\beta}^*)^{-1} \nabla_{\theta,\bgamma}^{2} R(\boldsymbol{\beta}^*)$. Comparing this with (\ref{eq_verify_sparse}), we can conclude that $\bomega^*$ is sparse and its support set is contained in $S_3$.

 Finally, we note that under the assumption that the smallest and largest eigenvalues of $\bSigma_Z:=\Cov(\bZ)$ are bounded away from 0 and infinity by some constants, we can show that $\|\bomega^*\|_2\leq C$ for some constant $C$ (see the proof of Lemma \ref{lem_verify_assumptions} above). If $|S_3|\leq s_3$, Cauchy-Schwarz inequality yields
$\|\bomega^*\|_1\lesssim s_3^{1/2}$. Thus, $\|\bomega^*\|_1$ at most is of order $s_3^{1/2}$.

\subsection{About sparse eigenvalue assumptions}\label{sec:sparseeigenvalue}

To establish the rate of $\hat\bbeta$, we need the following so-called sparse eigenvalue assumption.
\begin{assumption}
Define the largest and smallest sparse eigenvalues as
\begin{align}
&\rho_{\max} = \sup\bigg \{
\bv^T\nabla^2R_\delta^n(\bbeta)\bv: \norm{\bv}_2 = 1,\norm{\bv}_0 \leq Cs, \bbeta \in \Omega,\norm{\bbeta}_0\leq Cs \bigg\},\\
&\rho_{\min} = \inf\bigg \{
\bv^T\nabla^2R_\delta^n(\bbeta)\bv: \norm{\bv}_2 = 1,\norm{\bv}_0 \leq Cs, \bbeta \in \Omega,\norm{\bbeta}_0\leq Cs
\bigg \},
\end{align}
where $C$ is constant, $s$ is the sparsity level of $\bbeta^*$, and $\Omega$ is a suitable convex set containing the true parameter $\bbeta^*$ (e.g. $\Omega=\{\bbeta\in\RR^d: \|\bbeta-\bbeta^*\|_2\leq R\}$ for some $R>0$).
We assume that $c\leq \rho_{\min}\leq \rho_{\max}\leq 1/c$ for some constant $0<c<1$.
\end{assumption}
We note that we need to introduce the set $\Omega$ in the above definition, since as shown above the population risk $R(\bbeta)$ is globally non-convex. Thus, we have to focus on the eigenvalue structure in a local neighborhood around the truth. Theorem 1 in \cite{feng2022nonregular} showed that under the current Assumptions 1-4 and the above sparse eigenvalue assumption, the estimator $\hat\bbeta$ attains the minimax-optimal rate with  $\eta_1(n)=\sqrt{s}(\frac{s\log(d)}{n})^{\ell/(2\ell+1)}$. The same paper further verifies that the sparse eigenvalue assumption indeed holds, when the data are generated from the binary response model under some mild assumptions (e.g. the smallest and largest eigenvalues of $\Cov(\bZ)$ are bounded away from 0 and infinity by constants).

\subsection{About $\|\beta^*\|_2$}\label{sec:aboutbeta}
It is seen from Lemma \ref{lem_verify_assumptions} that, when the data are generated by a binary response model, we need some conditions on the norm of $\bbeta^*$ (either $\|\bbeta^*\|_1\leq C$ or $\|\bbeta^*\|_2\leq C$ for some constant $C$) in order to verify Assumptions 1-6 used in our main results. So, for the binary response model, we need to normalize $\bbeta^*$ according to either $L_1$ or $L_2$ norm.

To see the reason, let us consider the binary response model $Y=\operatorname{sign}\left(X-\boldsymbol{\beta}^{T} \boldsymbol{Z}+u\right)$, where for simplicity we assume $X \sim N(0,1), \boldsymbol{Z} \sim N(0, \mathbf{I})$ independent of $X$, $u$ is independent of $X,\bZ$ and $w(y) \equiv 1 .$ Assume its density  $f_{u}(\cdot)$ is bounded by a constant. The median of $u$ is 0.
By the definition of $R(\boldsymbol{\beta})$, it is easily shown that for any $\|\boldsymbol{v}\|_{2}=1$,
$$
\boldsymbol{v}^{T} \nabla^{2} R\left(\boldsymbol{\beta}^{*}\right) \boldsymbol{v}=2 \int\left(\boldsymbol{z}^{T} \boldsymbol{v}\right)^{2} f_{u}(0) \phi\left(\boldsymbol{\beta}^{* T} \boldsymbol{z}\right) \phi\left(z_{1}\right) \ldots \phi\left(z_{n}\right) d \boldsymbol{z},
$$
where $\phi$ is the pdf of $N(0,1)$. Since $f_{u}(\cdot)$ is upper bounded by a constant $C$, we have
$$
\begin{aligned}
\boldsymbol{v}^{T} \nabla^{2} R\left(\boldsymbol{\beta}^{*}\right) \boldsymbol{v} & \leq 2 C \int\left(\boldsymbol{z}^{T} \boldsymbol{v}\right)^{2} \frac{1}{(2 \pi)^{1 / 2}} \frac{1}{(2 \pi)^{d / 2}} \exp \left(-\frac{\boldsymbol{z}^{T}\left(\mathbf{I}+\boldsymbol{\beta}^{*} \boldsymbol{\beta}^{* T}\right) \boldsymbol{z}}{2}\right) d \boldsymbol{z} \\
&=2 C \frac{\boldsymbol{v}^{T}\left(\mathbf{I}+\boldsymbol{\beta}^{*} \boldsymbol{\beta}^{* T}\right)^{-1} \boldsymbol{v}}{(2 \pi)^{1 / 2}\left|\mathbf{I}+\boldsymbol{\beta}^{*} \boldsymbol{\beta}^{* T}\right|^{1 / 2}},
\end{aligned}
$$
where in the last step we use the integral of a normal distribution with variance $\left(\mathbf{I}+\boldsymbol{\beta}^{*} \boldsymbol{\beta}^{* T}\right)^{-1}$. Since the matrix $\mathbf{I}+\boldsymbol{\beta}^{*} \boldsymbol{\beta}^{* T}$ has $d-1$ eigenvalues 1 and 1 eigenvalue $1+\left\|\boldsymbol{\beta}^{*}\right\|_{2}^{2}$, we obtain that
$$
\boldsymbol{v}^{T} \nabla^{2} R\left(\boldsymbol{\beta}^{*}\right) \boldsymbol{v} \leq \frac{2 C}{(2 \pi)^{1 / 2}\left(1+\left\|\boldsymbol{\beta}^{*}\right\|_{2}^{2}\right)^{1 / 2}}.
$$
Thus, when $\left\|\boldsymbol{\beta}^{*}\right\|_{2}$ is diverging, $\boldsymbol{v}^{T} \nabla^{2} R\left(\boldsymbol{\beta}^{*}\right) \boldsymbol{v}$ is of order $1 /\left\|\boldsymbol{\beta}^{*}\right\|_{2}\rightarrow 0$. In other words, the population Hessian matrix $\nabla^{2} R(\boldsymbol{\beta}^{*})$ shrinks to 0 in the operator norm. This leads to slower convergence rate of the estimator $\hat\bbeta$ (see \cite{feng2022nonregular}) and also the asymptotic variance of the score test $\sigma^{*2}$ may diverge to infinity. For this reason, we need to normalize $\bbeta^*$ in the simulation.

\section{An Alternative Variance Estimator}\label{section_variance}
Recall that the numerator of the score statistic $\hat{U}_n$ can be written as
\[
\widehat{S}_{{\delta}}(0, \widehat{\gamma})-{\delta}^{\ell} \widehat{\mu}=\widehat{\boldsymbol{v}}^{T} \nabla R_{{\delta}}^{n}\left(\widehat{\boldsymbol{\beta}}_{0}\right)-{\delta}^{\ell} \gamma_{K, \ell} \widehat{\boldsymbol{v}}^{T} \widehat{T}_{h, U}^{\ell}(\widehat{\boldsymbol{\beta}}):= \hat{\boldsymbol{v}}^T D_{{\delta}, h}^{n}(\widehat{\boldsymbol{\beta}}, \widehat{\boldsymbol{v}})
\]
where
\[
 D_{\delta, h}^{n}(\boldsymbol{\beta}, \boldsymbol{v})
= \frac{1}{n} \sum_{i=1}^{n} w\left(y_{i}\right) y_{i} \frac{ \boldsymbol{z}_{i}}{\delta} M\left(x_{i}, \boldsymbol{z}_{i} ; \boldsymbol{\beta}\right)
:=  \frac{1}{n} \sum_{i=1}^{n} \Bar{D}_{\delta, h}^{i}(\boldsymbol{\beta}),
\]
\[
M\left(x_{i}, \boldsymbol{z}_{i} ; \boldsymbol{\beta}\right)=K\left(\frac{\boldsymbol{\beta}^{T} \boldsymbol{z}_{i}-x_{i}}{\delta}\right)-
\frac{\delta^{\ell+1}}{h^{\ell+1}} U^{(\ell)}\left(\frac{\boldsymbol{\beta}^{T} \boldsymbol{z}_{i}-x_{i}}{h}\right).
\]
Recall that Lemma \ref{lemma_ori_normality_new} shows that the asymptotic variance of $(n \delta)^{1 / 2} S_{\delta}\left(\boldsymbol{\beta}^{*}\right) $ is $ \sigma^{* 2}=\boldsymbol{v}^{* T} \boldsymbol{\Sigma}^{*} \boldsymbol{v}^{*}$. Therefore, we can alternatively  estimate the asymptotic variance
\[
\delta \operatorname{Var}\left[\boldsymbol{v}^T \Bar{D}_{\delta, h}^{i}(\boldsymbol{\beta})\right]=
\delta\mathbb{E}\left[\left(w(Y) Y \frac{\boldsymbol{v}^{T} \boldsymbol{Z}}{\delta} M(X, \boldsymbol{Z} ; \boldsymbol{\beta})\right)^{2}\right]
-
\delta \left[\EE w(Y) Y \frac{\boldsymbol{v}^{T} \boldsymbol{Z}}{\delta} M(X, \boldsymbol{Z} ; \boldsymbol{\beta})\right]^{2}
\]
by replacing the expectations above with sample averages under $\widehat{\boldsymbol{\beta}} \text { and } \widehat{\boldsymbol{v}}$. Formally, denote
\[
\Sigma^n_{\delta,h}(\boldsymbol{\beta})=\frac{1}{n}\sum_{i=1}^{n}  \Bar{D}_{\delta, h}^{i}(\boldsymbol{\beta}) \Bar{D}_{\delta, h}^{i}(\boldsymbol{\beta})^T,
\]
\[
\Tilde{\Sigma}^n_{\delta,h}(\boldsymbol{\beta})= \left(\frac{1}{n}\sum_{i=1}^{n}  \Bar{D}_{\delta, h}^{i}(\boldsymbol{\beta}) \right)  \left(\frac{1}{n}\sum_{i=1}^{n}  \Bar{D}_{\delta, h}^{i}(\boldsymbol{\beta}) \right)^T,
\]
and
\[
M^n_{\delta, h}\left(\boldsymbol{\beta}, \boldsymbol{v}\right) =  \frac{1}{n} \sum_{i=1}^{n} \left( \boldsymbol{v}^T \Bar{D}_{\delta, h}^{i}(\boldsymbol{\beta})\right)^2 - \left(\frac{1}{n} \sum_{i=1}^{n} \boldsymbol{v}^T \Bar{D}_{\delta, h}^{i}(\boldsymbol{\beta})\right)^2= \boldsymbol{v}^T
\left(
\Sigma^n_{\delta,h}(\boldsymbol{\beta})+\Tilde{\Sigma}^n_{\delta,h}(\boldsymbol{\beta})
\right)
\boldsymbol{v},
\]
then our kernel-free variance estimator is $\delta M^n_{\delta, h}\left(\hat{\boldsymbol{\beta}}, \hat{\boldsymbol{v}}\right)$.

\begin{theorem}\label{thm_alternative_var}
Under the conditions in Lemmas \ref{alvar1} and \ref{alvar2} and
\[
\delta \vee \frac{\delta^{\ell+1}}{h^{\ell+1}} +
\sqrt{\frac{\log (d)}{n \delta}} +
 M_n\eta_1(n) = o(1),
\]
we have
\[
\left|\delta M^n_{\delta, h}\left(\hat{\boldsymbol{\beta}}, \hat{\boldsymbol{v}}\right) -
\boldsymbol{v}^{* T} \boldsymbol{\Sigma}^{*} \boldsymbol{v}^{*}\right|
=\|\boldsymbol{v}^*\|_1^2 \mathcal{O}_{\mathbb{P}}\left(
\eta_2(n) +
(\delta \vee \frac{\delta^{\ell+1}}{h^{\ell+1}}) +
\sqrt{\frac{\log (d)}{n \delta}} +
 M_n\eta_1(n)
\right).
\]
\end{theorem}

Compared to the convergence rate in Lemma \ref{prop_var_pilot}, we can see that the rate of the kernel-free estimator above may have a much slower rate.

\begin{proof}
Following Lemma~\ref{alvar1}, Lemma~\ref{alvar2}, and Lemma \ref{lemma_consistent_bias_pilot} where we showed that $$\left\|H(\hat{\boldsymbol{\beta}}) - H(\boldsymbol{\beta}^*)\right\|_{\max}
\lesssim  M_n \|\hat{\boldsymbol{\beta}} - \boldsymbol{\beta}^*\|_1,$$ we can obtain
\[
\left\|
\Sigma^n_{\delta,h}(\hat{\boldsymbol{\beta}})+\Tilde{\Sigma}^n_{\delta,h}(\hat{\boldsymbol{\beta}}) - \frac{ \Tilde{\mu}_k H(\boldsymbol{\beta}^*)}{\delta}
\right\|_{\max} = \mathcal{O}_{\mathbb{P}}\left(
1 \vee\frac{\delta^\ell}{h^{\ell+1}} +
\sqrt{\frac{\log (d)}{n \delta^{3}}} +
\frac{M_n\eta_1(n)}{\delta}
\right).
\]
Since $\left\|H\left(\boldsymbol{\beta}^{*}\right)\right\|_{\max }=\mathcal{O}(1)$, we have
$ \left\| \delta \left(\Sigma^n_{\delta,h}(\hat{\boldsymbol{\beta}})+\Tilde{\Sigma}^n_{\delta,h}(\hat{\boldsymbol{\beta}})\right) \right\|_{\max } = \mathcal{O}_{\mathbb{P}}(1)$ when
\[
\delta \vee \frac{\delta^{\ell+1}}{h^{\ell+1}} +
\sqrt{\frac{\log (d)}{n \delta}} +
 M_n\eta_1(n) = o(1).
\]

By triangle inequality we have
\begin{align*}
    & \left| \delta M^n_{\delta,h}\left(\hat{\boldsymbol{\beta}}, \hat{\boldsymbol{v}}\right) -
   \boldsymbol{v}^{* T} \boldsymbol{\Sigma}^{*} \boldsymbol{v}^{*}
    \right|\\
    \leq &
    \left\|\widehat{\boldsymbol{v}}^{(1)}-\boldsymbol{v}^{*}\right\|_{1}^{2}\delta
    \left\| \Sigma^n_{\delta,h}(\hat{\boldsymbol{\beta}})+\Tilde{\Sigma}^n_{\delta,h}(\hat{\boldsymbol{\beta}}) \right\|_{\max }
    +2 \left\|
    \boldsymbol{v}^{* T} \delta \left(\Sigma^n_{\delta,h}(\hat{\boldsymbol{\beta}})+\Tilde{\Sigma}^n_{\delta,h}(\hat{\boldsymbol{\beta}})\right) \right\|_{\infty}
    \left\|\widehat{\boldsymbol{v}}^{(1)}-\boldsymbol{v}^{*}\right\|_{1}\\
    +&\left|\boldsymbol{v}^{* T}\left(
    \delta\left(\Sigma^n_{\delta,h}(\hat{\boldsymbol{\beta}})+\Tilde{\Sigma}^n_{\delta,h}(\hat{\boldsymbol{\beta}})\right) - \boldsymbol{v}^{* T} \boldsymbol{\Sigma}^{*} \boldsymbol{v}^{*}
    \right) \boldsymbol{v}^{*}\right|.
\end{align*}
This implies that
\[
\left|\delta M^n_{\delta, h}\left(\hat{\boldsymbol{\beta}}, \hat{\boldsymbol{v}}\right) -
\boldsymbol{v}^{* T} \boldsymbol{\Sigma}^{*} \boldsymbol{v}^{*}\right|
=\|\boldsymbol{v}^*\|_1^2 \mathcal{O}_{\mathbb{P}}\left(
\eta_2(n) +
\delta \vee \frac{\delta^{\ell+1}}{h^{\ell+1}} +
\sqrt{\frac{\log (d)}{n \delta}} +
 M_n\eta_1(n)
\right).
\]
\end{proof}

\begin{lemma}\label{alvar1} Under Assumptions 1-4 and conditions in Lemma \ref{lemma_consistent_bias_pilot}, assume $\int U^{(l) 2}(u) du < C_1 < \infty$, $|U^{(\ell)}(u)| < C_2 < \infty$ for some constant $C_1,C_2$, for any fixed $\boldsymbol{\beta}$, we have
\[
\left\|\mathbb{E}\left( \Bar{D}_{\delta, h}^{i}(\boldsymbol{\beta}) \Bar{D}_{\delta, h}^{i}(\boldsymbol{\beta})^T\right) + \mathbb{E}\left(  \Bar{D}_{\delta, h}^{i}(\boldsymbol{\beta})\right) \mathbb{E}\left(  \Bar{D}_{\delta, h}^{i}(\boldsymbol{\beta})\right)^T
- \frac{\Tilde{\mu}_k H(\boldsymbol{\beta})}{\delta}\right\|_{\max} = \mathcal{O}(1 \vee \frac{\delta^\ell}{h^{\ell+1}} ).
\]
\end{lemma}
\begin{proof}
\begin{align*}
   \mathbb{E}\left( \Bar{D}_{\delta, h}^{i}(\boldsymbol{\beta}) \Bar{D}_{\delta, h}^{i}(\boldsymbol{\beta})^T\right)
    =&\mathbb{E}\left[w^2(Y) \frac{\boldsymbol{Z}\boldsymbol{Z}^T}{\delta^2} M^2(X, \boldsymbol{Z} ; \boldsymbol{\beta})\right] \\
    = &
     \mathbb{E}\left[\nabla \bar{R}_{\delta}^{i}\left(\boldsymbol{\beta}\right)\nabla \bar{R}_{\delta}^{i}\left(\boldsymbol{\beta}\right)^T
     \right]
     +\mathbb{E}\left[ w^2(Y) \boldsymbol{Z} \boldsymbol{Z}^T \frac{\delta^{2\ell}}{h^{2(\ell+1)}}\left(U^{(\ell)}\left(\frac{\boldsymbol{\beta}^{ T} \boldsymbol{Z}-X}{h}\right)\right)^2\right]\\
     - & 2 \mathbb{E}\left[\nabla \bar{R}_{\delta}^{i}\left(\boldsymbol{\beta}\right)
  \left(w(Y) \boldsymbol{Z}^T \frac{\delta^{\ell}}{h^{\ell+1}} U^{(\ell)}\left(\frac{\boldsymbol{\beta}^{ T} \boldsymbol{Z}-X}{h}\right)\right)\right].
\end{align*}
In the proof of Lemma 1 we showed that
\[
\max_{j,k}\left|  \left(\mathbb{E}\left[\nabla \bar{R}_{\delta}^{i}\left(\boldsymbol{\beta}\right)\nabla \bar{R}_{\delta}^{i}\left(\boldsymbol{\beta}\right)^T
     \right]\right)_{jk} - \left(\frac{\Tilde{\mu}_k H(\boldsymbol{\beta})}{\delta}\right)_{jk} \right|= \mathcal{O}(1).
\]

If $\int U^{(l) 2}(u) du < C_1 < \infty$, $|U^{(\ell)}(u)| < C_2 < \infty$ for some constant $C_1,C_2$, we have
\begin{align*}
    &\left(\mathbb{E}\left[ w^2(Y) \boldsymbol{Z} \boldsymbol{Z}^T \frac{\delta^{2\ell}}{h^{2(\ell+1)}}\left(U^{(\ell)}\left(\frac{\boldsymbol{\beta}^{ T} \boldsymbol{Z}-X}{h}\right)\right)^2\right]\right)_{jk}\\
    = &\sum_{y \in\{-1,1\}} \delta^{2\ell} w^2(y) \int \frac{z_jz_k}{h^{2(1+\ell)}} \int U^{(\ell) 2}\left(\frac{\boldsymbol{\beta}^{T} \boldsymbol{z}-x}{h}\right) f(x \mid y, \boldsymbol{z}) d x f(y, \boldsymbol{z}) d \boldsymbol{z}\\
    = & \sum_{y \in\{-1,1\}} \delta^{2\ell} w^2(y) \int \frac{z_jz_k}{h^{2\ell+1}} \int U^{(\ell) 2}\left(u\right) f(uh+\boldsymbol{\beta}^{T}\boldsymbol{z} \mid y, \boldsymbol{z}) d u f(y, \boldsymbol{z}) d \boldsymbol{z}\\
    =& \mathcal{O}(\frac{\delta^{2\ell}}{h^{2\ell+1}}),
\end{align*}
and
\begin{align*}
    & \left(\mathbb{E}\left[\nabla \bar{R}_{\delta}^{i}\left(\boldsymbol{\beta}\right)
  \left(w(Y) \boldsymbol{Z}^T \frac{\delta^{\ell}}{h^{\ell+1}} U^{(\ell)}\left(\frac{\boldsymbol{\beta}^{ T} \boldsymbol{Z}-X}{h}\right)\right)\right]\right)_{jk}\\
  = & \sum_{y \in\{-1,1\}}\delta^{\ell-1} w^2(y)\int \frac{z_jz_k}{h^{1+\ell}} \int  K\left(\frac{\boldsymbol{\beta}^{ T} \boldsymbol{z}_{i}-x_{i}}{\delta}\right) U^{(\ell)}\left(\frac{\boldsymbol{\beta}^{ T} \boldsymbol{z}_{i}-x_{i}}{h}\right)
    f(x \mid y, \boldsymbol{z}) d x f(y, \boldsymbol{z}) d \boldsymbol{z}\\
 =&  \sum_{y \in\{-1,1\}}\delta^{\ell}w^2(y)\int \frac{z_jz_k}{h^{1+\ell}}\int K(u)U^{(\ell)}(\frac{u\delta}{h})f(u\delta+\boldsymbol{\beta}^{T}\boldsymbol{z}\mid y,\boldsymbol{z})duf(y,\boldsymbol{z})d\boldsymbol{z}\\
  =& \mathcal{O}(\frac{\delta^{\ell}}{h^{\ell+1}}).
\end{align*}
Note that since $\frac{\delta}{h} = o(1)$, we have
\[
\max_{j,k}\left|  \left(\mathbb{E}\left( \Bar{D}_{\delta, h}^{i}(\boldsymbol{\beta}) \Bar{D}_{\delta, h}^{i}(\boldsymbol{\beta})^T\right)\right)_{jk} - \left(\frac{\Tilde{\mu}_k H(\boldsymbol{\beta})}{\delta}\right)_{jk} \right|= \mathcal{O}(1 \vee \frac{\delta^\ell}{h^{\ell+1}}).
\]
Together with Lemma~\ref{alvar2} which shows that $\left\|\mathbb{E}\left(\Bar{D}_{\delta, h}^{i}(\boldsymbol{\beta})\right) \right\|_{\infty}
= \mathcal{O}\left( \delta^\ell\right)$, we have,
\[
\left\|\mathbb{E}\left( \Bar{D}_{\delta, h}^{i}(\boldsymbol{\beta}) \Bar{D}_{\delta, h}^{i}(\boldsymbol{\beta})^T\right) + \mathbb{E}\left(  \Bar{D}_{\delta, h}^{i}(\boldsymbol{\beta})\right) \mathbb{E}\left(  \Bar{D}_{\delta, h}^{i}(\boldsymbol{\beta})\right)^T
- \frac{\Tilde{\mu}_k H(\boldsymbol{\beta})}{\delta}\right\|_{\max} = \mathcal{O}(1 \vee \frac{\delta^{\ell}}{h^{\ell+1}}).
\]
\end{proof}

\begin{lemma} \label{alvar2}Under Assumptions 1-4 and conditions in Lemma \ref{lemma_consistent_bias_pilot}, assume $|U^{(\ell)}(u)| < C < \infty$ for some constant $C$,  $\int K^4(t) dt < \infty$. For any fixed $\boldsymbol{\beta}$, we have
\[
\left\|\mathbb{E}\left[ \left( \Bar{D}_{\delta, h}^{i}(\boldsymbol{\beta}) \Bar{D}_{\delta, h}^{i}(\boldsymbol{\beta})^T\right) \right]\right\|_{\max}
= \mathcal{O}\left(\frac{1}{\delta}\right),
\]
\[
\left\|\mathbb{E}\left(\Bar{D}_{\delta, h}^{i}(\boldsymbol{\beta})\right) \right\|_{\infty}
= \mathcal{O}\left( \delta^\ell\right),
\]
\[
 \left\| \Sigma^n_{\delta,h}(\boldsymbol{\beta}) + \Tilde{\Sigma}^n_{\delta,h}(\boldsymbol{\beta})
 - \mathbb{E}\left( \Bar{D}_{\delta, h}^{i}(\boldsymbol{\beta}) \Bar{D}_{\delta, h}^{i}(\boldsymbol{\beta})^T\right)
  - \mathbb{E}\left(  \Bar{D}_{\delta, h}^{i}(\boldsymbol{\beta})\right) \mathbb{E}\left(  \Bar{D}_{\delta, h}^{i}(\boldsymbol{\beta})\right)^T
 \right\|_{\max} = \mathcal{O}_{\mathbb{P}}\left(\sqrt{\frac{\log (d)}{n \delta^{3}}}\right).
\]
\end{lemma}
\begin{proof}
Denote
\begin{align*}
    \Sigma_{ijk} = &\left( \Bar{D}_{\delta, h}^{i}(\boldsymbol{\beta}) \Bar{D}_{\delta, h}^{i}(\boldsymbol{\beta})^T- \mathbb{E}\left( \Bar{D}_{\delta, h}^{i}(\boldsymbol{\beta}) \Bar{D}_{\delta, h}^{i}(\boldsymbol{\beta})^T\right)\right)_{jk}\\
    = &w^2(y_i)\frac{z_{ij}z_{ik}}{\delta^2}M^2\left(x_{i}, \boldsymbol{z}_{i} ; \boldsymbol{\beta}\right) - \left(\mathbb{E}\left( \Bar{D}_{\delta, h}^{i}(\boldsymbol{\beta}) \Bar{D}_{\delta, h}^{i}(\boldsymbol{\beta})^T\right)\right)_{jk}.
\end{align*}
We have $\mathbb{E}(\Sigma_{ijk})=0$. Note that since $\frac{\delta}{h} = o(1)$, we have
\begin{align*}
    &\mathbb{E}\left[ \left( \Bar{D}_{\delta, h}^{i}(\boldsymbol{\beta}) \Bar{D}_{\delta, h}^{i}(\boldsymbol{\beta})^T\right)_{jk}^2 \right]\\
  =&\mathbb{E}\left[w^4(Y)  \frac{z_j^2 z_k^2}{\delta^4} M^4(X, \boldsymbol{Z} ; \boldsymbol{\beta})\right] \\
     = & \sum_{y} \int w^4(y) \frac{z_{j}^{2} z_{k}^{2}}{\delta^{4}}
    \left( K\left(\frac{\boldsymbol{\beta}^{T} \boldsymbol{z}-x}{\delta}\right)-
\frac{\delta^{\ell+1}}{h^{\ell+1}} U^{(\ell)}\left(\frac{\boldsymbol{\beta}^{T} \boldsymbol{z}-x}{h}\right)\right)^4
     f(x \mid y, \boldsymbol{z}) d x f(y, \boldsymbol{z}) d \boldsymbol{z} \\
     \lesssim &  \sum_{y} \int w^4(y) \frac{z_{j}^{2} z_{k}^{2}}{\delta^{4}}
     K^4\left(\frac{\boldsymbol{\beta}^{T} \boldsymbol{z}-x}{\delta}\right)
     f(x \mid y, \boldsymbol{z}) d x f(y, \boldsymbol{z}) d \boldsymbol{z} \\
     = & \sum_{y} \int w^4(y) \frac{z_{j}^{2} z_{k}^{2}}{\delta^{3}}
     K^4\left(u\right)
     f(u\delta+ \boldsymbol{\beta}^{T} \boldsymbol{z}\mid y, \boldsymbol{z}) d u f(y, \boldsymbol{z}) d \boldsymbol{z}
     = \mathcal{O}\left(\frac{1}{\delta^3}\right),
\end{align*}
where in the third step, we ignore the term involving the fourth moment of $U^{(\ell)}(\frac{\bbeta^T\bz-x}{h})$ because it is of a smaller order. Similarly we can show that $\left(\mathbb{E}\left( \Bar{D}_{\delta, h}^{i}(\boldsymbol{\beta}) \Bar{D}_{\delta, h}^{i}(\boldsymbol{\beta})^T\right)\right)_{jk}
= \mathcal{O}\left(\frac{1}{\delta}\right)$. Since $|\Sigma_{ijk}| \lesssim \frac{M_n^2}{\delta^2}$, by Bernstein inequality we can show that with probability greater than $1-\mathcal{O}\left(d^{-1}\right)$,
\[
 \left\| \Sigma^n_{\delta,h}(\boldsymbol{\beta}) - \mathbb{E}\left( \Bar{D}_{\delta, h}^{i}(\boldsymbol{\beta}) \Bar{D}_{\delta, h}^{i}(\boldsymbol{\beta})^T\right) \right\|_{\max} \lesssim \sqrt{\frac{\log (d)}{n \delta^{3}}}.
\]
Following the same proof of Lemma \ref{lemma_consistent_bias_pilot} we can show that
\begin{align*}
    &\left(\mathbb{E}\left(\Bar{D}_{\delta, h}^{i}(\boldsymbol{\beta})\right) \right)_j \\
    = &
     \sum_{y} \int w(y) \frac{z_{j}}{\delta}
    \left( K\left(\frac{\boldsymbol{\beta}^{T} \boldsymbol{z}-x}{\delta}\right)-
\frac{\delta^{\ell+1}}{h^{\ell+1}} U^{(\ell)}\left(\frac{\boldsymbol{\beta}^{T} \boldsymbol{z}-x}{h}\right)\right)
     f(x \mid y, \boldsymbol{z}) d x f(y, \boldsymbol{z}) d \boldsymbol{z} \\
     = & \mathcal{O}\left( \delta^\ell \right),
\end{align*}
and  with probability greater than $1-\mathcal{O}\left(d^{-1}\right)$,
\[
\left\| \frac{1}{n}\sum_{i=1}^{n}  \Bar{D}_{\delta, h}^{i}(\boldsymbol{\beta}) - \mathbb{E}\left(\Bar{D}_{\delta, h}^{i}(\boldsymbol{\beta})\right) \right\|_{\infty} \lesssim \sqrt{\frac{\log (d)}{n \delta}}.
\]
Therefore,  with probability greater than $1-\mathcal{O}\left(d^{-1}\right)$
\begin{align*}
     &\left\| \Tilde{\Sigma}^n_{\delta,h}(\boldsymbol{\beta}) - \mathbb{E}\left(  \Bar{D}_{\delta, h}^{i}(\boldsymbol{\beta})\right) \mathbb{E}\left(  \Bar{D}_{\delta, h}^{i}(\boldsymbol{\beta})\right)^T\right\|_{\max} \\
    = & \left\| \left( \frac{1}{n}\sum_{i=1}^{n}  \Bar{D}_{\delta, h}^{i}(\boldsymbol{\beta}) - \mathbb{E}\left(  \Bar{D}_{\delta, h}^{i}(\boldsymbol{\beta})\right)\right) \left( \frac{1}{n}\sum_{i=1}^{n}  \Bar{D}_{\delta, h}^{i}(\boldsymbol{\beta}) + \mathbb{E}\left(  \Bar{D}_{\delta, h}^{i}(\boldsymbol{\beta})\right)\right)^T\right\|_{\max}\\
    \leq &
    \left\| \frac{1}{n}\sum_{i=1}^{n}  \Bar{D}_{\delta, h}^{i}(\boldsymbol{\beta}) - \mathbb{E}\left(\Bar{D}_{\delta, h}^{i}(\boldsymbol{\beta})\right) \right\|_{\infty}
    \left( 2 \left(\mathbb{E}\left(  \Bar{D}_{\delta, h}^{i}(\boldsymbol{\beta})\right)\right)_\infty +  \left\| \frac{1}{n}\sum_{i=1}^{n}  \Bar{D}_{\delta, h}^{i}(\boldsymbol{\beta}) - \mathbb{E}\left(\Bar{D}_{\delta, h}^{i}(\boldsymbol{\beta})\right) \right\|_{\infty} \right)\\
    \lesssim & \left( \sqrt{\frac{\log (d)}{n \delta}} \vee \delta^\ell \right)\sqrt{\frac{\log (d)}{n \delta}}.
\end{align*}
Combining this with the result for $\Sigma^n_{\delta,h}(\boldsymbol{\beta})$ we finish the proof.
\end{proof}

\section{Practical Considerations}\label{section_practical}
In practice, in order to apply the smoothed decorrelated score test presented in Section \ref{method}, we need to select the bandwidth and regularization parameter in the initial estimators $\hat\bomega$ and $\hat\bbeta$, two bandwidths for the plug-in asymptotic bias and variance estimators in (\ref{pilot_bias}) and (\ref{eq_est_var}), and the final bandwidth $\delta$ for the score test statistic. In this section, we discuss practical considerations for the proposed test as well as the data-driven selection procedure for $\delta$.

We use a path-following algorithm to compute the initial estimator $\hat\bbeta$ and apply a two way cross-validation approach to choose the tuning parameters $(\delta,\lambda)$ in the optimization (\ref{eq_hatbeta}); see \cite{feng2022nonregular} for details. The initial values of the path-following algorithm used in the simulation are shown in Section \ref{appendix_simulation}.
To compute the Dantzig estimator $\hat\bomega$, we use the ``flare" package in R. The same cross-validation method can be used to select the tuning parameters $(\delta,\lambda')$ in (\ref{eq_def_dantzig}). Meanwhile, we find the empirical performance of the Dantzig estimator $\hat\bomega$ is not very sensitive to the choice of $(\delta,\lambda')$. In the simulation and real data analysis, to ease the computation, we set $\lambda' = 2(\log d/n)^{1/5}$ and $\delta = 1$ in (\ref{eq_def_dantzig}). Notice that in rare cases, the Dantzig solver may become unstable numerically when the Hessian $\nabla^2 R^n_\delta(\wh\bbeta)$ is ill-conditioned. In this case, we will firstly project it onto the cone of positive definite matrices and then plug-in the projected matrix into (\ref{eq_def_dantzig}) to solve $\wh\bomega$.


For the bandwidth parameter $h$ in the plug-in bias estimator, Lemma \ref{lemma_consistent_bias_pilot} implies that the theoretical optimal order for $h$ is $(\log(d)/n)^{2\ell + 2\zeta + 1}$ (if $M_n\eta_1(n) = o(h)$ which holds under mild conditions). In practice, we recommend choosing $h = (\log(d)/n)^{2\ell' + 3}$, where $\ell'$ is the order of the kernel function $K$. For variance estimator, we use the kernel-free approach, as it is more convenient in practice.

Finally, for the main bandwidth parameter $\delta$ appearing in the score test statistic, if a pre-specified value is preferred, we recommend choosing $\delta = c n^{-1/(2\ell' + 1)}$ for some constant $c >0$ (e.g., $c=1$). When using
our proposed data-driven bandwidth selector, notice that another bandwidth parameter $b$, whose optimal order is $(\log(d)/n)^{1/(2\ell + 2r + 1)}$ (see the discussion after Theorem \ref{theorem_adapt_main}), is required for estimating the squared bias $SB(\delta)$. Similarly, we suggest choosing $b=c_0(\log(d)/n)^{1/(2\ell' + 2r' + 1)}$ for some $c_0$ (e.g., $c_0=0.5$), where $r'$ is the order of kernel $J$. 

\section{Additional Numerical Results}\label{appendix_simulation}

\subsection{Computational details of the path-following algorithm}
In Section \ref{simulation}, we apply the path-following algorithm proposed in \citet{feng2022nonregular} to estimate the true coefficient $\bbeta$. We fix the number of stages $N = 25$, $\nu = 0.25$, $\eta = 0.25$ ,$\epsilon_{tgt} = 0.0001$ and set $\Omega = \{\bbeta: \norm{\bbeta}_2 \leq 10^3\}$ for each parameter in Algorithm 2.1 and 2.2 in  \citet{feng2022nonregular}.

\subsection{Normal Q-Q plots in Section \ref{simulation_fixed}}\label{appendix_qqplots_fixed}

Figure \ref{QQ_manski} contains the Normal Q-Q plots for the test statistics obtained from the experiments in Section \ref{simulation_fixed} under both scenarios. We can see that the Q-Q plots from the proposed SDS method appear close to Gaussian, while the Q-Q plots for DS and Honest method deviate from standard Gaussian.

\subsection{Empirical power in Section \ref{simulation_fixed}}\label{appendix_power_fixed}
Figure \ref{fig_Power_2} shows the empirical power under the Heteroskedastic Uniform scenario in Section \ref{simulation_fixed}. Similar to the figures for the Heteroskedastic Gaussian scenario, the empirical power converges to 1 as the magnitude of $\beta^*$ becomes larger. In addition, as the correlation between $\bZ$ becomes larger, the power also decreases.

\subsection{Bandwidth selection in Section \ref{sec_sim_band}}\label{appendix_band}
Finally, we investigate whether the data-driven bandwidth is close to the theoretically optimal bandwidth that minimizes the MSE. We focus on the Heteroskedastic Gaussian scenario and consider the setting with $n = 800, d = 50, s = 10, \rho = 0.2$ and set  $\beta^*_1 = \dotso = \beta^*_{10} = 1/\sqrt{10}$. We compare $M(\delta)$ with $\wh M(\delta)$, where $\wh M(\delta)$ is obtained using the implementation discussed in Section \ref{section_practical}, but with different bandwidth $b = 0.15,0.2,0.25,0.3$ for squared bias estimation. Notice that since directly obtaining $M(\delta)$ is hard, we apply Monte Carlo integration to approximate this quantity. 

Figure \ref{fig_MSE_compare} shows the true $M(\delta)$ as well as $\wh M(\delta)$ in each case. The red line is the true MSE $M(\delta)$ approximated by Monte Carlo method, and the black solid line is the average estimated MSE $\wh M(\delta)$ with different bandwidths $b$ over 100 repetitions. The black dashed line is the one standard deviation band around $\wh M(\delta)$ over these repetitions.
As we can see, for all cases, $M(\delta)$ lies within the one standard deviation band around $\wh M(\delta)$ and the minimizer of $M(\delta)$
is very close to the minimizer of the average $\wh M(\delta)$. This experiment confirms that the data-driven bandwidth $\wh\delta$ is close to the optimal $\delta^*$, provided that the bandwidth $b$ for squared-bias estimation lies in a suitable range.

\subsection{Bandwidth selection in high dimension}\label{sec:newsimu_high}

We run the data-driven bandwidth selection procedure under the high dimensional setting. Specifically, we consider the same Heteroskedastic Gaussian and Heteroskedastic Uniform scenario with $n = 300,d = 500, s = 3$, and $\beta^*_0 \in \{0, 0.05,\dotso,0.3\}$. Similar to Section \ref{sec_sim_band}, we seek for the minimizer of the estimated MSE over $\delta \in [0.1,1.2]$, and then plug-in this estimated bandwidth $\hat{\delta}$ into the decorrelated score function. 
We repeat this experiment 250 times for each case. Table \ref{table_adaptive_type1} and Figure \ref{fig_adaptive} show the empirical Type-I error and empirical power, respectively. The performance for both Type-I error and power is similar as in the low-dimensional case.

\subsection{Sensitivity analysis}\label{sec:newsimu_sensi}

We run a set of experiments with different choices of kernels. We consider the experiments similar to Section 6.1 with $n = 300, d \in \{100, 500\}$ and $s =3$. Instead of using Gaussian kernel in the decorrelated score function, we consider the following two choices:
\begin{itemize}
    \item Epanechnikov(parabolic) kernel: $ K(u) = \frac34(1-u^2),|u|\leq 1$,
    \item Quartic (biweight) kernel: $K(u) = \frac{15}{16}(1-u^2)^2,|u|\leq 1$.
\end{itemize}

Table \ref{table_adaptive_type2} shows the empirical Type-I error using these two kernels under both low and high dimensional regions. We can see that the empirical Type-I errors are close to the nominal level 0.05, and seem not largely affected by the choice of kernels in this setting.

In the next set of experiments, we study how the choice of bandwidth affects the empirical Type-I error. Similar to the above experiments, we fix $n = 300,d=100,s=3$, and consider the fixed bandwidth setting with $\delta \in \{0.3,0.5,0.75,1.0\}$ in the decorrelated score function. The results are shown in Table  \ref{table_bandwidth}. We can see that the empirical performance of our testing procedure using these bandwidth parameters is in general close to the nominal level 0.05.

Finally, we also examine how the bandwidth in the bias estimator influences the empirical performance of the score test. We run additional experiments with a range of different bandwidths in the plug-in bias estimator. Specifically, we consider the same Heteroskedastic Gaussian scenario with $n = 300,d = 100,500$ and $s = 3$. To estimate $\hat \mu$, we set the bandwidth $h = 0.5,0.8,1.0,1.2$ in the plug-in estimator and evaluate the empirical Type-I error with 250 repeated experiments for each case. From Table \ref{table_variance_type1}, we can see that although the empirical Type-I errors vary with different bandwidths, they are generally around the nominal level 0.05.

\subsection{Additional results from ChAMP analysis}\label{appendix_champ}

Table \ref{table_champ_DS} lists the five most significant variables using the DS and Honest approach. Notice that the significant variables identified by the proposed approach are also in the list for the DS approach. Honest yields very different set of variables compared to the other two.

\begin{figure}[H]
	\centering
	{	\vspace{-0.8cm}
	\subfigure[]{\includegraphics[width=70mm,height = 60mm]{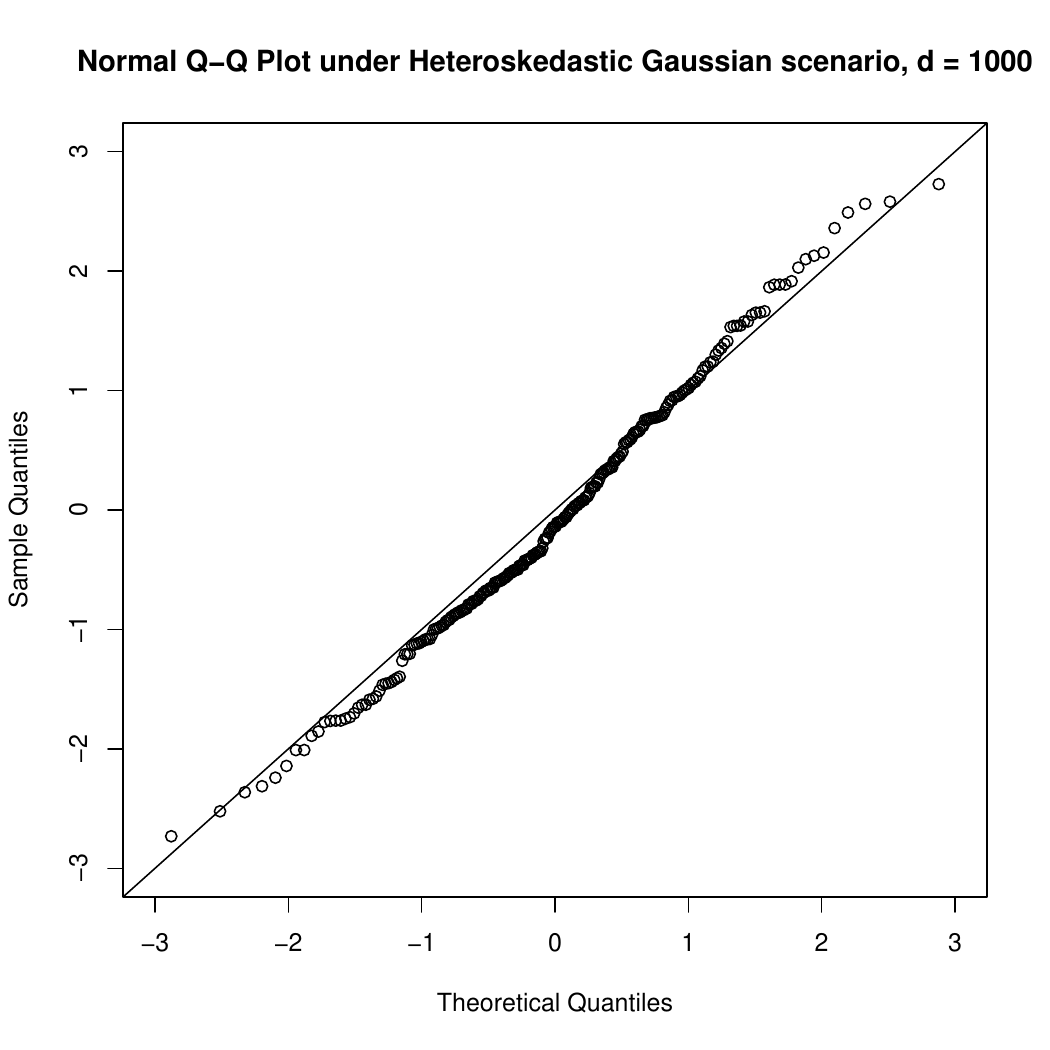}}
	}
	{\subfigure[]{\includegraphics[width=70mm,height = 60mm]{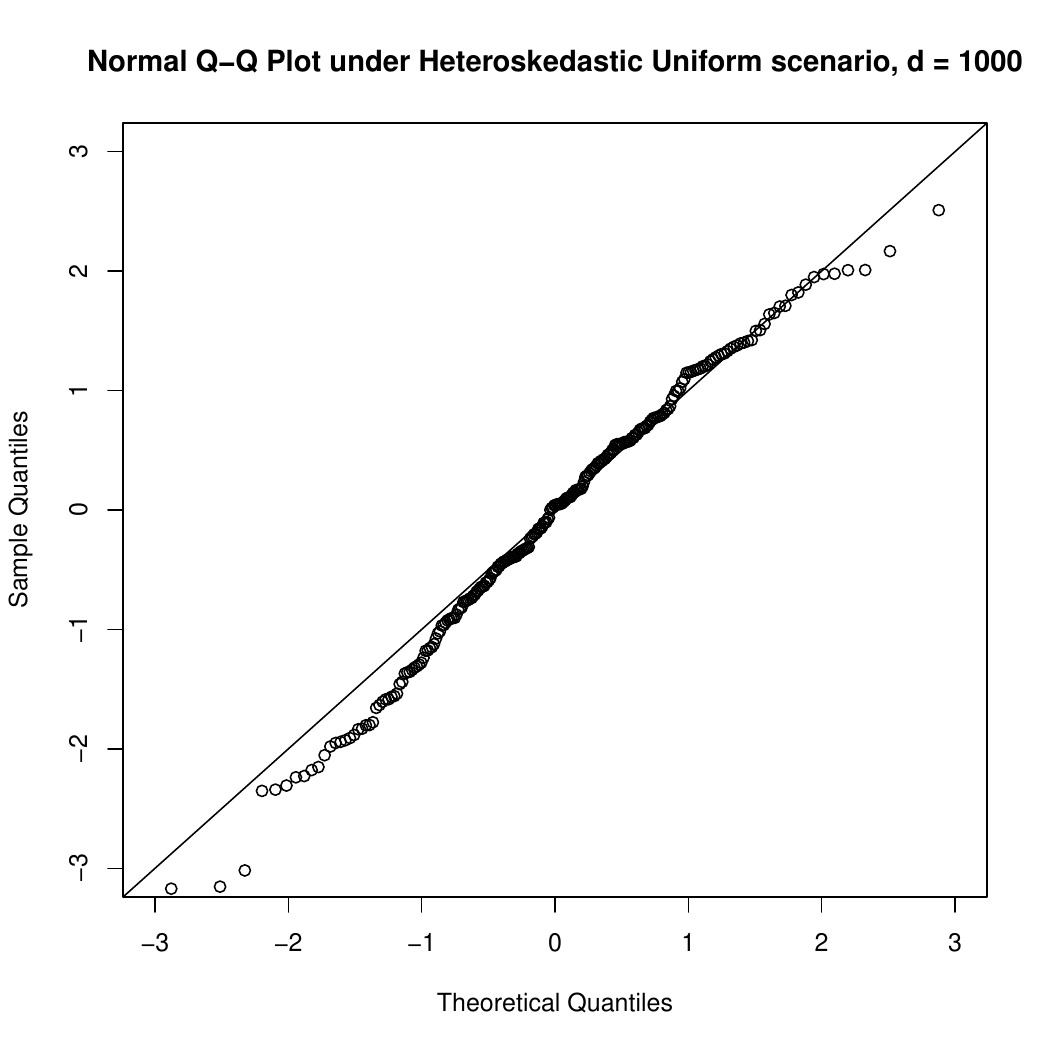}}
	}
	\vspace{-0.8cm}
	{\subfigure[]{\includegraphics[width=70mm,height = 60mm]{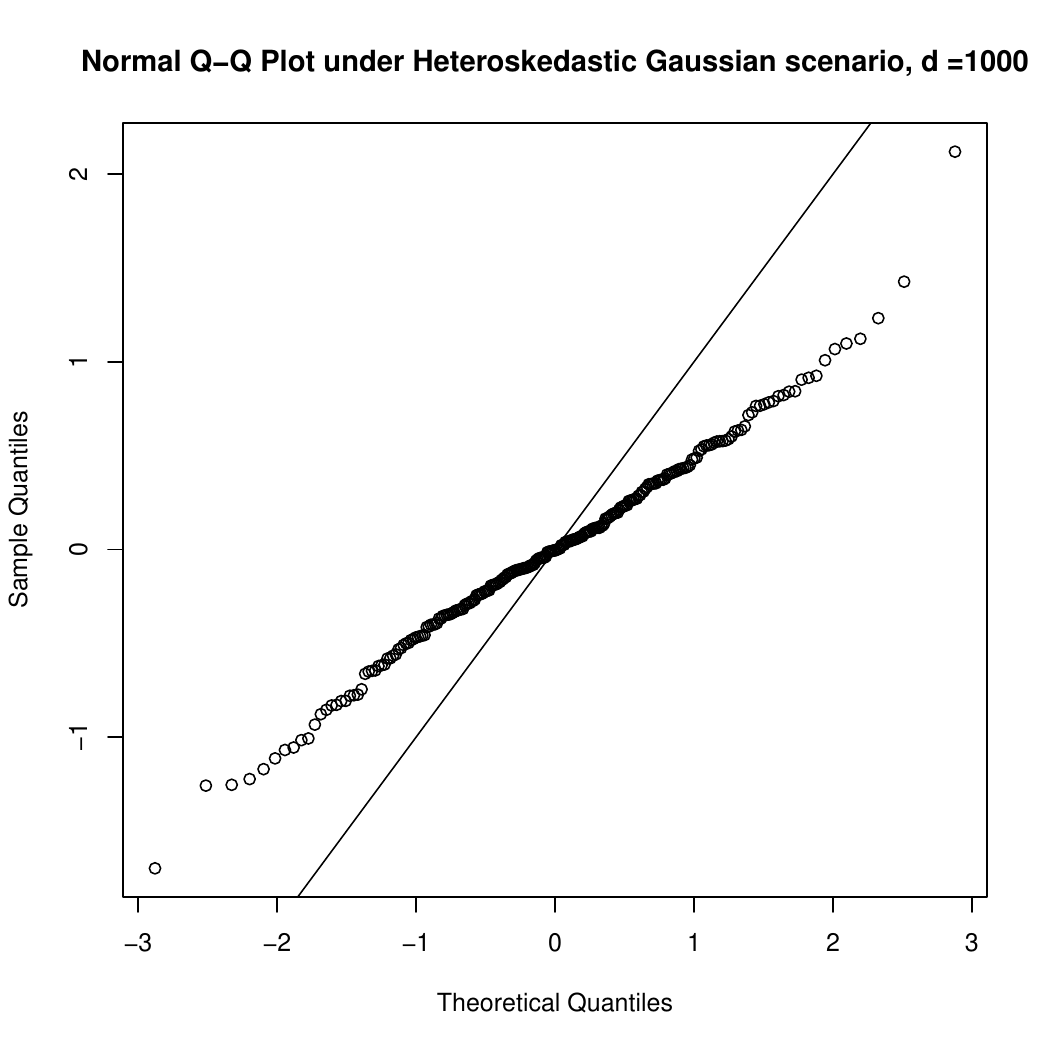}}
	}
	{\subfigure[]{\includegraphics[width=70mm,height = 60mm]{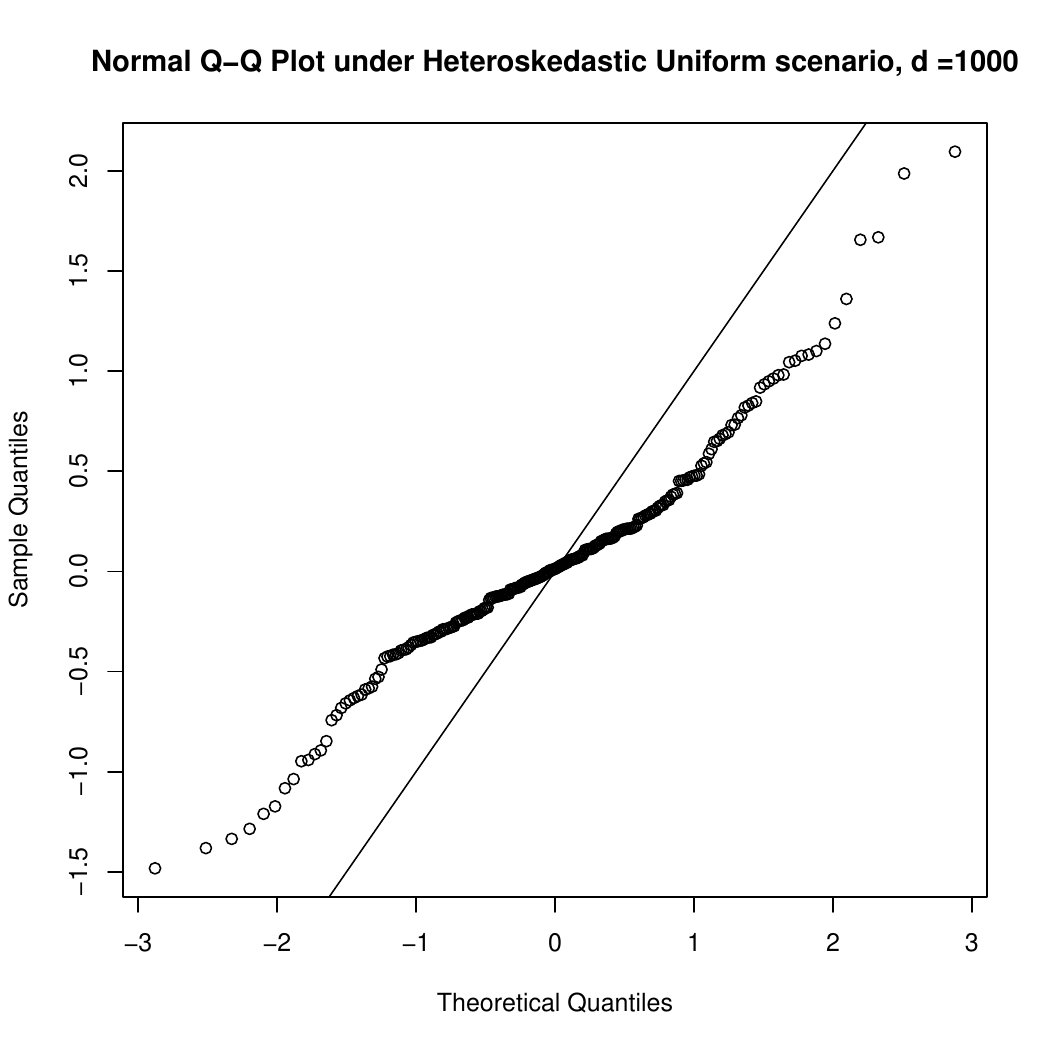}}
	}
\vspace{-0.4cm}
	{\subfigure[]{\includegraphics[width=70mm,height = 60mm]{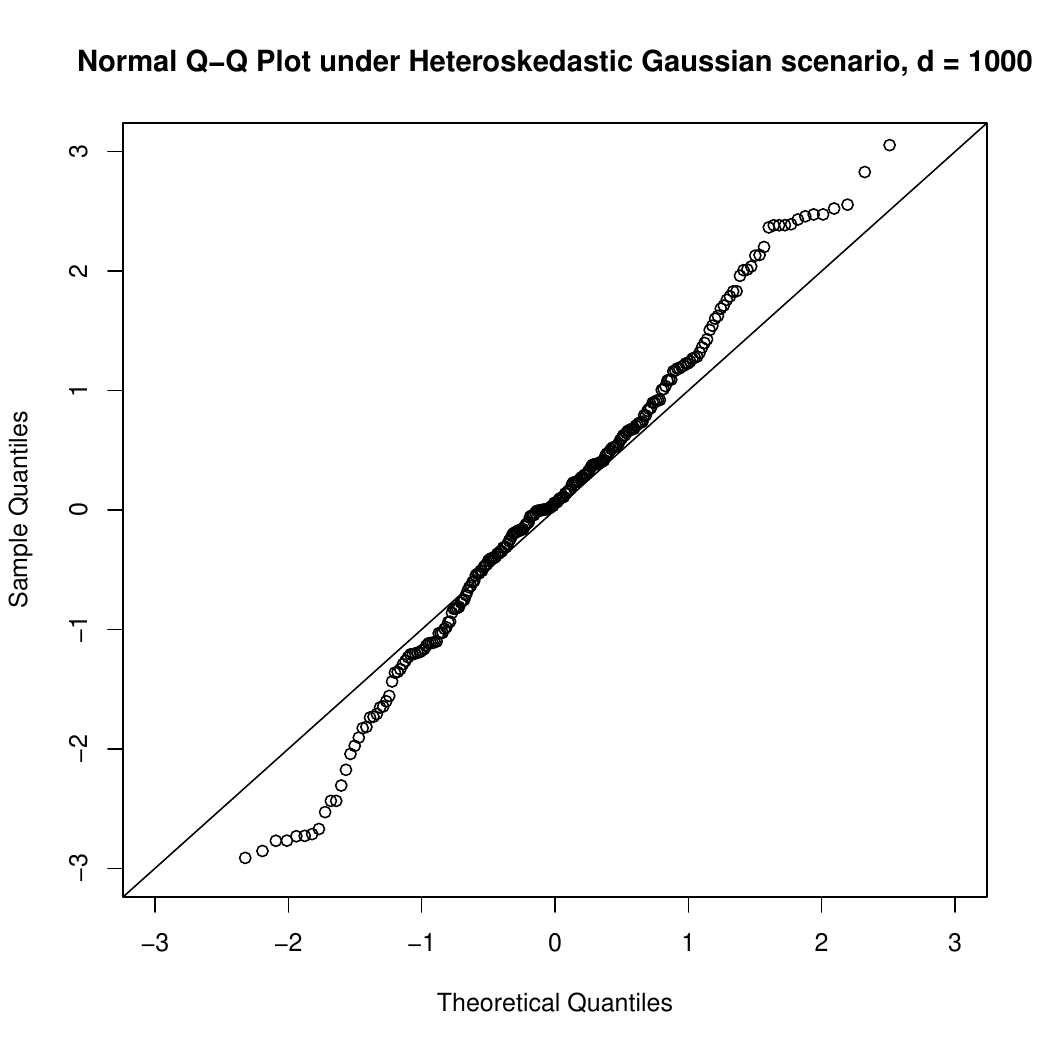}}
	}
	{\subfigure[]{\includegraphics[width=70mm,height = 60mm]{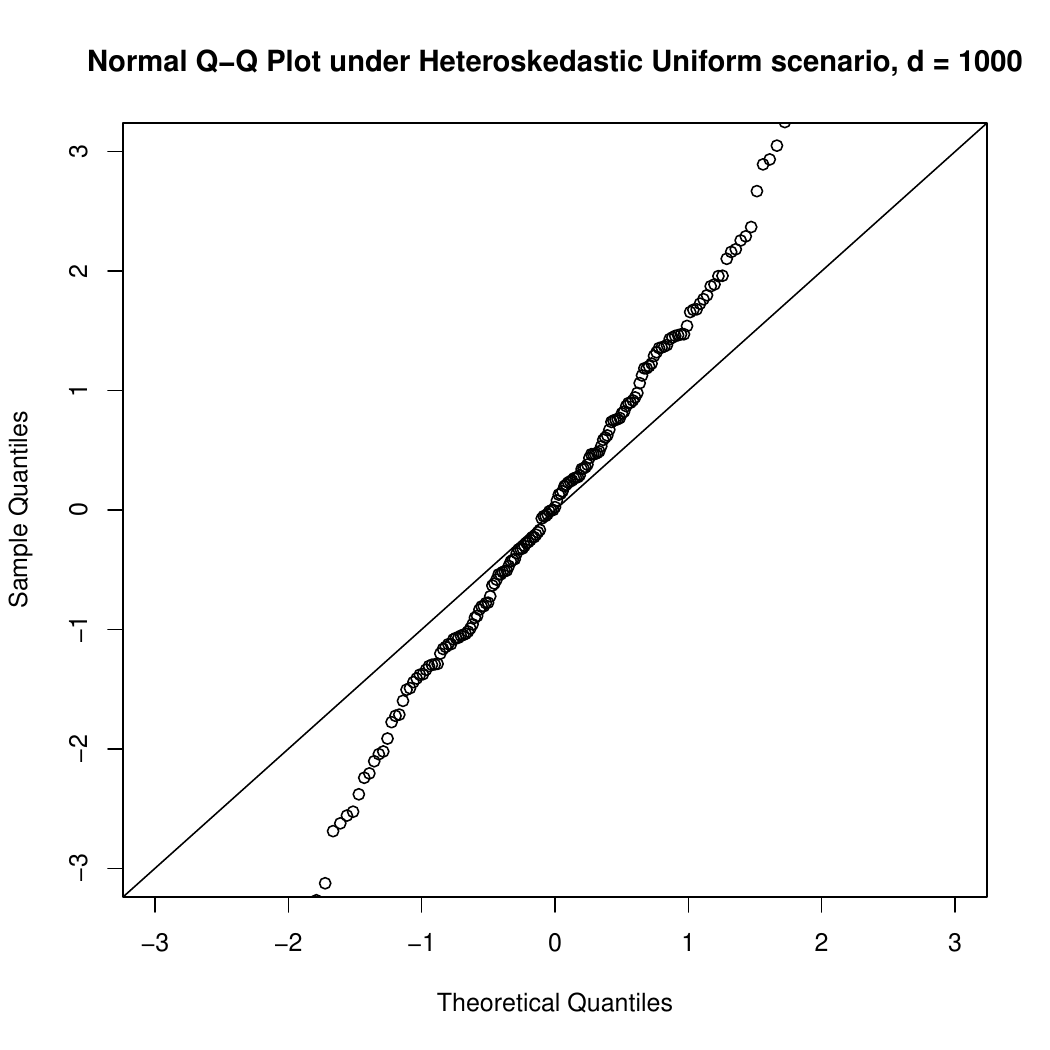}}
	}
	\caption{Gaussian Q-Q plot of the test statistics under the setting $d = 1000, s = 10,\rho = 0.5$ from: (a, b) SDS method, (c, d) DS method, and (e, f) Honest method.}   \label{QQ_manski}
\end{figure}

\begin{figure}[H]
	\centering
	{\subfigure[]{\includegraphics[width=80mm,height = 80mm]{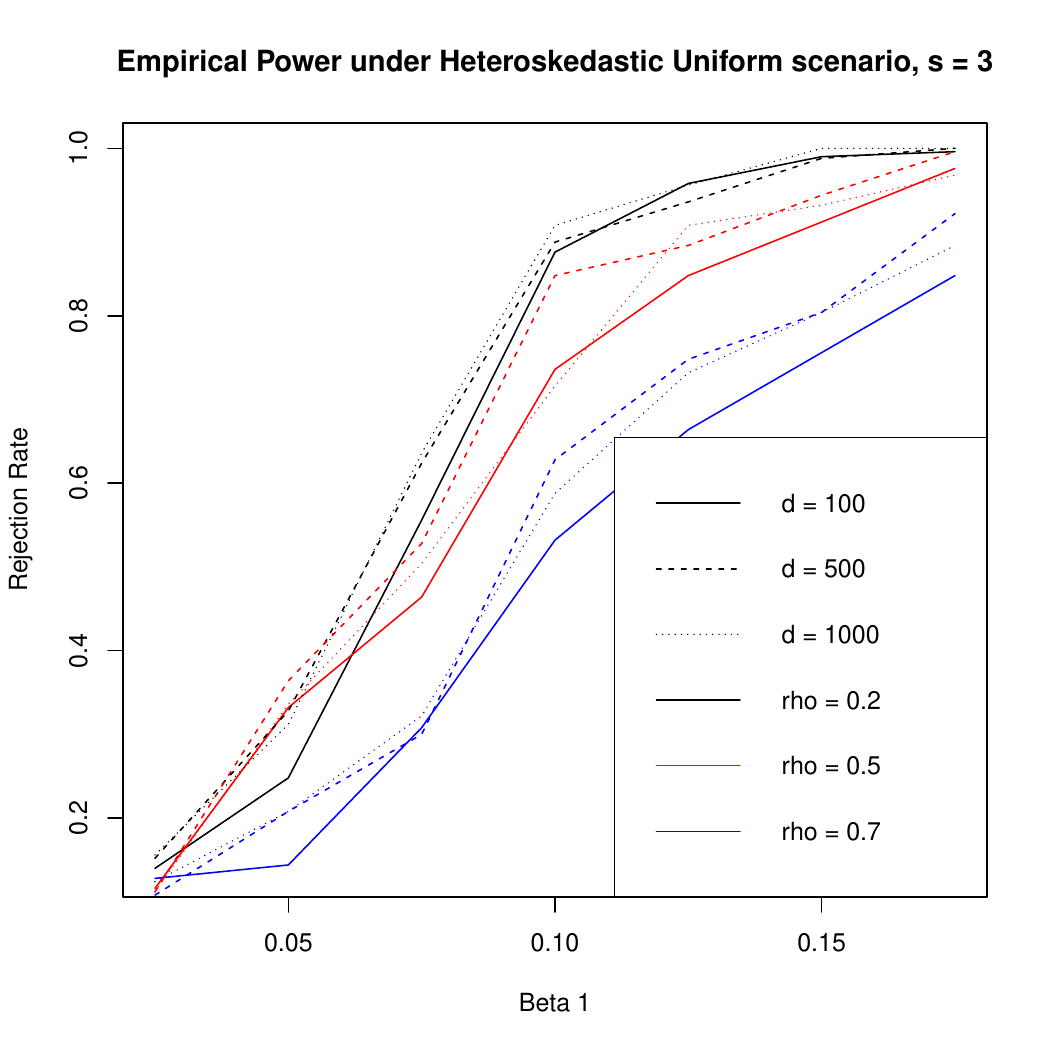}}
	}
	{\subfigure[]{\includegraphics[width=80mm,height = 80mm]{{plots/Power_Uniform_10}.pdf}}
	}
	\caption{Empirical rejection rate of the proposed test under both scenarios with $s = 3$, $d = 100, 500, 1000$ and $\rho = 0.2,0.5,0.7$.}   \label{fig_Power_2}
\end{figure}

\begin{figure}[H]
	\centering
	{\subfigure[]{\includegraphics[width=80mm,height = 70mm]{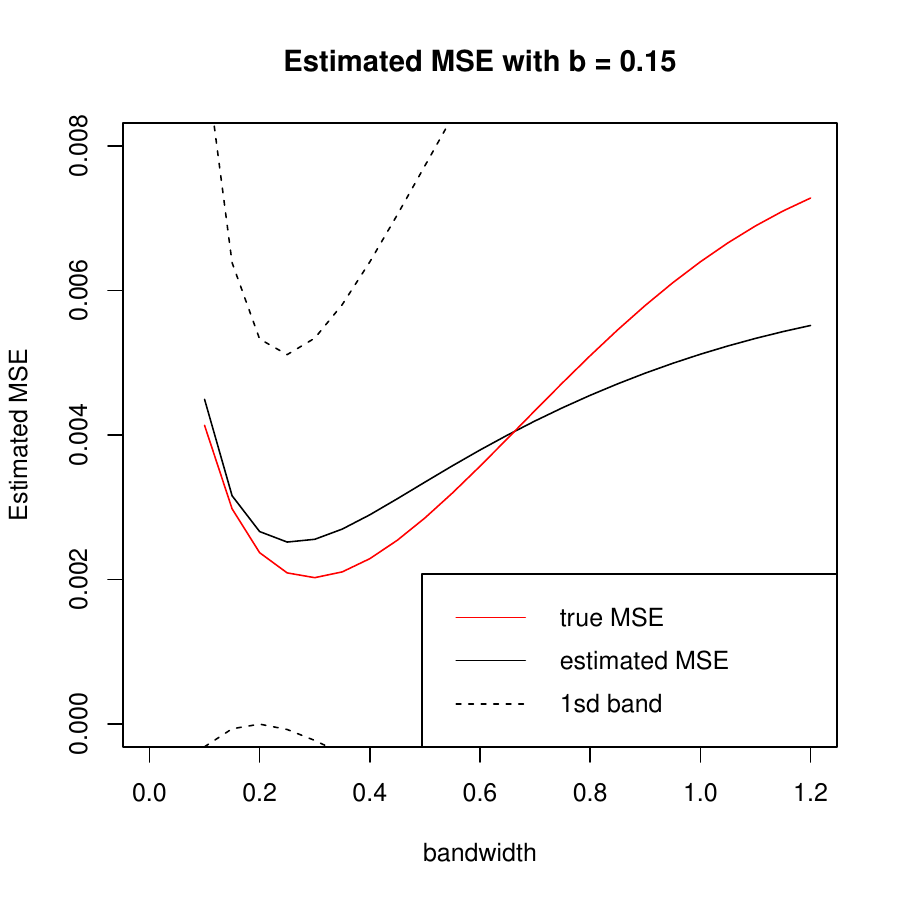}}
	}
	{\subfigure[]{\includegraphics[width=80mm,height = 70mm]{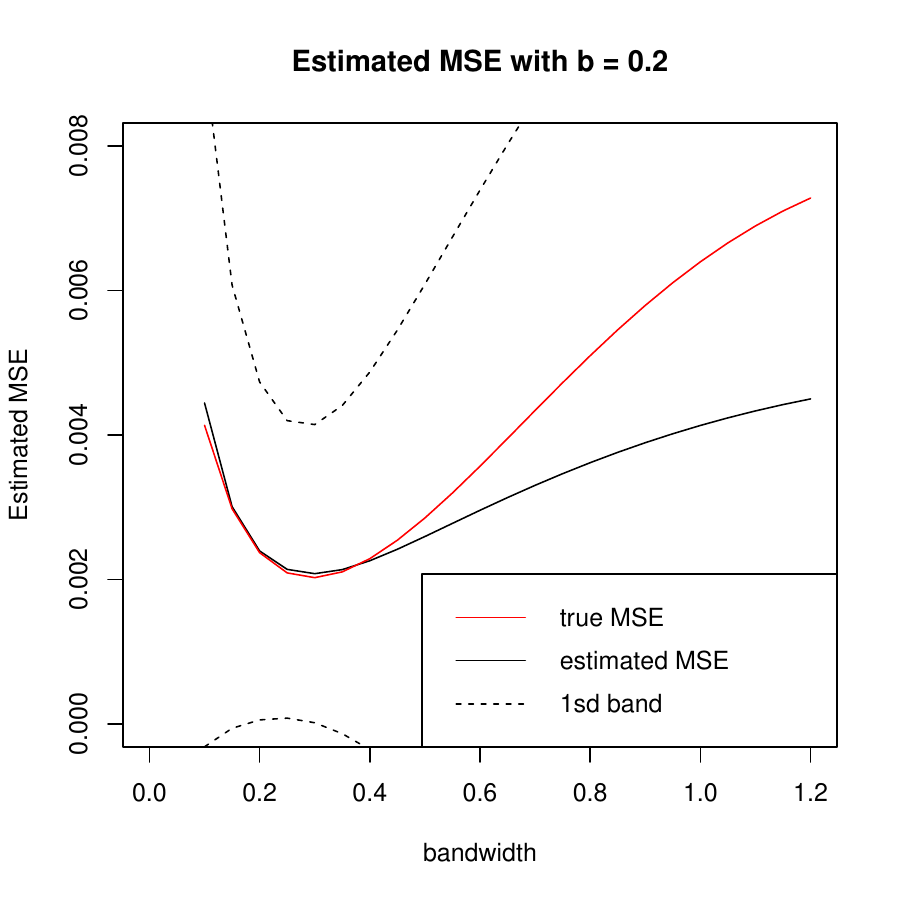}}
	}
\vskip \baselineskip
{\subfigure[]{\includegraphics[width=80mm,height = 70mm]{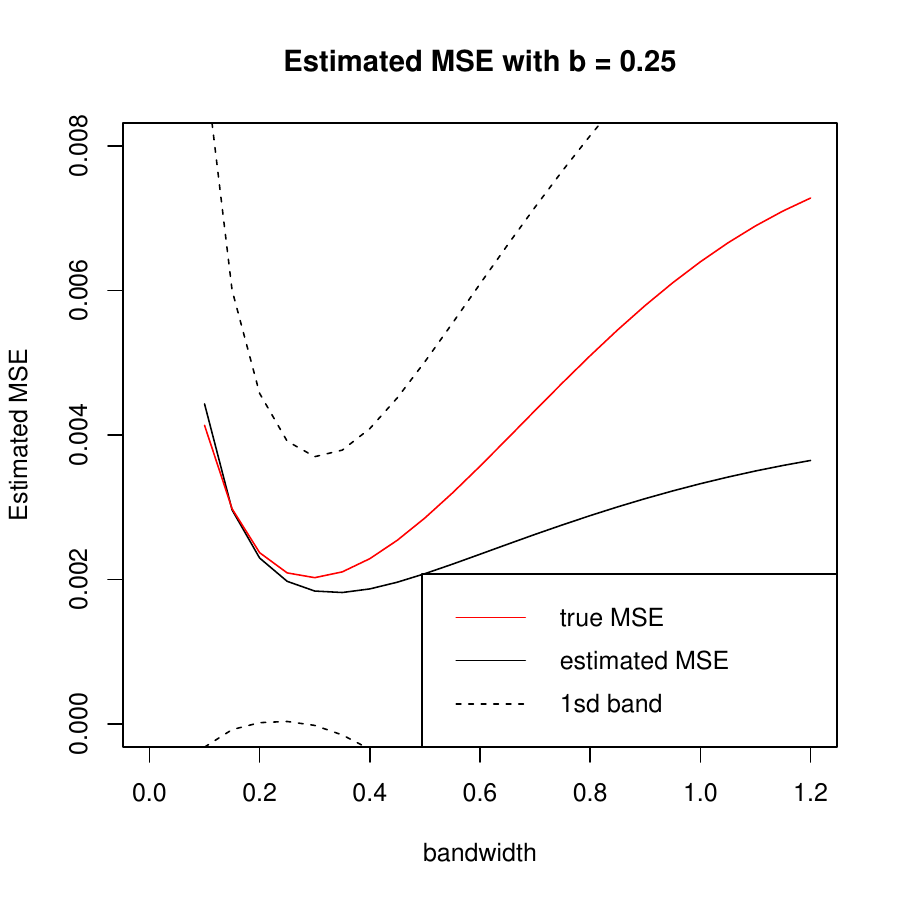}}
}
{\subfigure[]{\includegraphics[width=80mm,height = 70mm]{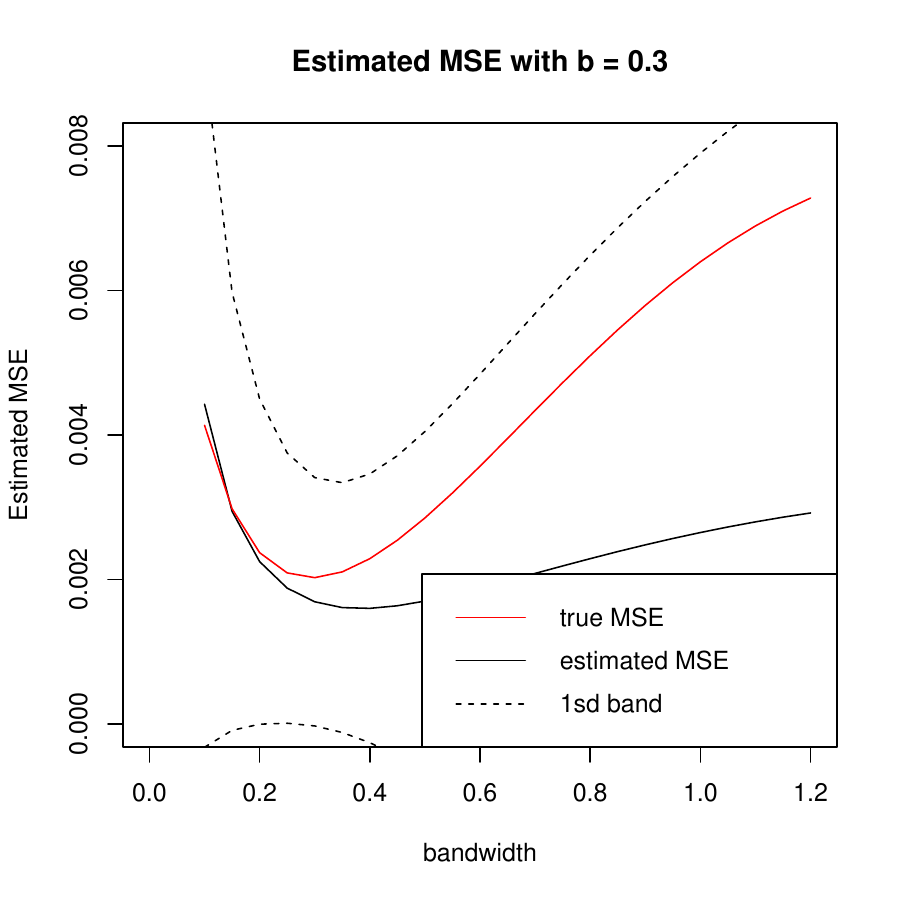}}
}
	\caption{The true $M(\delta)$ approximated by Monte Carlo Method (red solid line) and the average estimated $\wh M(\delta)$ (black solid line) with bandwidth $b = 0.15,0.20,0.25,0.30$ for squared-bias estimation over 100 repetitions. The black dashed line is one standard deviation bands around $\wh M(\delta)$.}   \label{fig_MSE_compare}
\end{figure}

\begin{table}[H]
	\caption{The empirical Type I error rate of the tests under the Heteroskedastic Gaussian and Uniform scenarios with data-driven bandwidth $\hat\delta$.}\label{table_adaptive_type1}
	\setlength\extrarowheight{1pt}
	\begin{center}
	\begin{tabular}{ c | c c c }
			\hline
			& &$s = 3$&\\
			Data generating process&$\rho = 0.2$&$\rho = 0.5$&$\rho = 0.7$\\
			\hline
			Heteroskedastic Gaussian&7.6\%&8.0\%&6.4\%\\
			Heteroskedastic Uniform&7.6\%&6.8\%&8.4\%\\
			\hline
		\end{tabular}
	\end{center}	
\end{table}

\begin{figure}[H]
	\centering
	\includegraphics[width=120mm,height = 120mm]{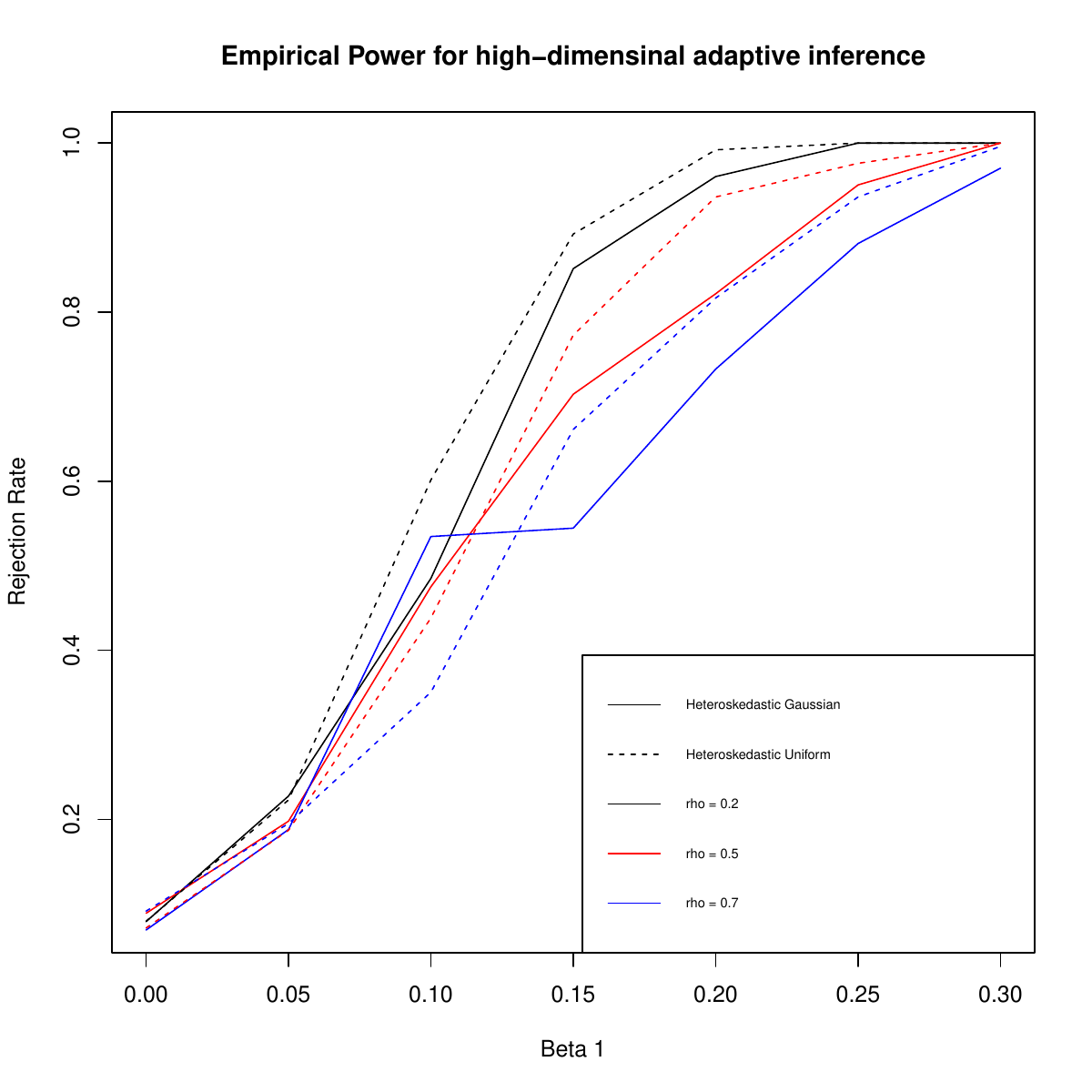}
	\caption{Empirical power of the tests under the Heteroskedastic Gaussian and Uniform scenarios with data-driven bandwidth $\hat\delta$. \label{fig_adaptive}}
\end{figure}

\begin{table}[H]
	\caption{The empirical Type I error rate of the tests under the Heteroskedastic Gaussian and Uniform scenarios with alternative kernel functions.}\label{table_adaptive_type2}
	\setlength\extrarowheight{1pt}
	\begin{center}
	\begin{tabular}{c | c | c c c | c c c}
			\hline
			&& &Epanechnikov&&&Quartic&\\
			d&Data generating process&$\rho = 0.2$&$\rho = 0.5$&$\rho = 0.7$&$\rho = 0.2$&$\rho = 0.5$&$\rho = 0.7$\\
			\hline
			100&Heteroskedastic Gaussian&4.8\%&3.6\%&5.2\%&4.4\%&2.8\%&5.6\%\\
			500&Heteroskedastic Uniform&4.4\%&4.0\%&5.6\%&3.6\%&3.2\%&6.0\%\\
			\hline
		\end{tabular}
	\end{center}	
\end{table}

\begin{table}[H]
	\caption{The empirical Type I error rate of the tests under the Heteroskedastic Gaussian and Uniform scenarios with different bandwidth.}\label{table_bandwidth}
	\setlength\extrarowheight{1pt}
	\begin{center}
	\begin{tabular}{c | c | c c c}
			\hline
			$\delta$&Data generating process&$\rho = 0.2$&$\rho = 0.5$&$\rho = 0.7$\\
			\hline
			0.3&Heteroskedastic Gaussian&6.4\%&5.2\%&3.6\%\\
			0.5&Heteroskedastic Gaussian&5.2\%&6.8\%&8.0\%\\
			0.75&Heteroskedastic Gaussian&8.0\%&3.6\%&7.6\%\\
			1.0&Heteroskedastic Gaussian&7.6\%&6.4\%&7.6\%\\
			\hline
			0.3&Heteroskedastic Uniform&5.2\%&3.6\%&8.4\%\\
			0.5&Heteroskedastic Uniform&8.4\%&6.4\%&4.8\%\\
			0.75&Heteroskedastic Uniform&6.0\%&5.6\%&6.4\%\\
			1.0&Heteroskedastic Uniform&9.2\%&6.4\%&7.2\%\\
			\hline
		\end{tabular}
	\end{center}	
\end{table}

\begin{table}[H]
	\caption{The empirical Type I error rate of the tests under the Heteroskedastic Gaussian with different bandwidth $h$ in the plug-in bias estimator.}\label{table_variance_type1}
	\setlength\extrarowheight{1pt}
	\begin{center}
	\begin{tabular}{ c | c c c| c c c }
			\hline
			& &$d = 100$&&&$d = 500$&\\
			h &$\rho = 0.2$&$\rho = 0.5$&$\rho = 0.7$&$\rho = 0.2$&$\rho = 0.5$&$\rho = 0.7$\\
			\hline
			0.5&6.8 \%&8.0 \%&5.6 \%&6.4 \%&5.2 \%&2.8 \%\\
			0.8&7.6 \%&4.8 \%&7.2 \%&8.4 \%&4.8 \%&5.6 \%\\
			1.0&8.0 \%&4.8 \%&3.6 \%&7.2 \%&5.2 \%&3.6 \%\\
			1.2&6.8 \%&4.0 \%&7.6 \%&6.0 \%&3.2 \%&6.0 \%\\
			\hline
		\end{tabular}
	\end{center}	
\end{table}

\begin{table}[H]
	\caption{Five most significant variables using DS and Honest approach, sorted from most significant (left) to less significant (right), with their corresponding p-values }\label{table_champ_DS}
	\setlength\extrarowheight{1pt}
	\begin{center}
		\begin{tabular}{ c |c c c c c }
			\hline
			Method&&&Significant variables&&\\
			\hline
			DS& MMenB & X(WOMAC Pain Score) &WFunc\_6mo & KQOL\_6wk& KSymp\_3mo\\
			&9.754e-05&3.402e-03&1.649e-02&1.685e-02 &2.136e-02\\
			Honest&ActualGrp&MMenTear&TrochDam&TrochCenLes&MMenMgt\\
			&$<$1e-16&$<$1e-16&$<$1e-16&$<$1e-16&$<$1e-16\\
			\hline
		\end{tabular}
	\end{center}	
\end{table}

\end{document}